\documentclass[a4paper,UKenglish]{lipics-v2019}

\usepackage{microtype}
\usepackage{bm} 
\usepackage{soul} 
\newcommand{\xto}[1]{\xrightarrow{#1}}
\newcommand{\ppa}{\textup{PPA}}
\newcommand{\ccfg}{\textup{CCFG}}
\newcommand{\leven}{\textup{$\mathcal L_{even}$} }
\usepackage{thmtools,,stmaryrd}
\usepackage{thm-restate}
\usepackage{tikz}\usetikzlibrary{automata}
\usetikzlibrary{decorations.pathmorphing}
\usepackage{bm}

\makeatletter
\@mparswitchfalse%
\let\c@theorem\relax

\declaretheorem[name={Theorem},style=plain]{theorem}
\let\c@proposition\relax

\declaretheorem[name={Proposition},style=plain,numberlike=theorem]{proposition}
\let\c@lemma\relax

\declaretheorem[name={Lemma},style=plain,numberlike=theorem]{lemma}
\let\c@definition\relax

\declaretheorem[name={Definition},style=definition,numberlike=theorem]{definition}
\let\c@corollary\relax

\declaretheorem[name={Corollary},style=definition,numberlike=theorem]{corollary}
\let\c@remark\relax

\declaretheorem[name={Remark},style=remark,numberlike=theorem]{remark}
\let\c@example\relax

\declaretheorem[name={Example},style=remark,numberlike=theorem]{example}

\usepackage{xspace}
\usepackage{hyperref}

\usepackage{cleveref}

\definecolor{itemi}{RGB}{60,60,61}

\title{Weakly-unambiguous Parikh automata and their link to holonomic series}

\newcommand{\crochet}[1]{\left[\kern-0.15em\left[  #1  \right]\kern-0.15em\right]}
\newcommand{\intint}[2]{\left[\kern-0.15em\left[  #1, #2  \right]\kern-0.15em\right]}
\newcommand{\NN}{\mathbb{N}}
\renewcommand{\AA}{\mathcal{A}}
\newcommand{\BB}{\mathcal{B}}
\newcommand{\CC}{\mathcal{C}}
\newcommand{\EE}{\mathcal{E}}

\newcommand{\OO}{\mathcal{O}}
\newcommand{\mdeg}{\deg_m}

\colorlet{darkred}{red!60!black}
\colorlet{darkgreen}{green!60!black}
\newcommand{\intro}[1]{\textcolor{darkred}{\emph{#1}}}
\newcommand{\vect}[1]{\boldsymbol{#1}}

\usepackage{mathtools}

\newcommand{\era}[1]{\xrightarrow{#1}}
\newcommand{\eRa}[1]{\xRightarrow{#1}}

\newcommand{\eRb}[2]{\xRightarrow[#2]{#1}}
\newcommand{\qinit}{q_I}
\newcommand{\normF}{S_\infty}

\newcommand{\eqdef}{:=}

\DeclareMathOperator{\RCM}{RCM}

\DeclareMathOperator{\RBCM}{RBCM}

\DeclareMathOperator{\PA}{PA}

\DeclareMathOperator{\lcm}{lcm}

\newif\ifdraft\draftfalse
\ifdraft

\newcommand{\ac}[1]{\textcolor{darkgreen}{[ #1 - Arnaud]}}

\newcommand{\cn}[1]{\textcolor{blue}{[ #1 - Cyril]}}

\newcommand{\fk}[1]{\textcolor{red}{[ #1 - Florent]}}

\newcommand{\ab}[1]{\textcolor{orange}{[ #1 - Alin]}}
\else
\newcommand\cn[1]{}
\newcommand\ac[1]{}

\newcommand{\fk}[1]{}

\newcommand{\ab}[1]{}
\fi
\newcommand{\reviewer}[1]{ \marginpar{\it\footnotesize #1}}

\titlerunning{Weakly-unambiguous Parikh automata and their link to holonomic series}

\author{Alin Bostan}{Inria and Universit\'e Paris-Saclay, 1 rue Honor\'e d'Estienne d'Orves, 91120 Palaiseau, France}{Alin.Bostan@inria.fr}{}{}

\author{Arnaud Carayol}{LIGM, Univ Gustave Eiffel, CNRS, F77454 Marne-la-Vallée, France}{arnaud.carayol@u-pem.fr}{}{}

\author{Florent Koechlin}{LIGM, Univ Gustave Eiffel, CNRS, F77454 Marne-la-Vallée, France}{florent.koechlin@u-pem.fr}{}{}

\author{Cyril Nicaud}{LIGM, Univ Gustave Eiffel, CNRS, F77454 Marne-la-Vallée, France}{cyril.nicaud@u-pem.fr}{}{}

\authorrunning{A. Bostan, A. Carayol, F. Koechlin and C. Nicaud}

\Copyright{A. Bostan, A. Carayol, F. Koechlin and C. Nicaud}

\ccsdesc[500]{Theory of computation~Formal languages and automata theory}\ccsdesc[500]{Mathematics of computing~Generating functions}

\keywords{generating series, holonomy, ambiguity, reversal bounded counter machine, {P}arikh automata}

\category{}

\relatedversion{}

\supplement{}

\funding{}

\nolinenumbers 
\hideLIPIcs  

\begin{document}

\maketitle

\begin{abstract}
We investigate the connection between properties of formal languages and
properties of their generating series, with a focus on the class of
\emph{holonomic} power series. We first prove a strong version of a conjecture
by Castiglione and Massazza: weakly-unambiguous Parikh automata are equivalent
to unambiguous two-way reversal bounded counter machines, and
their multivariate generating series are holonomic. We then show that the
converse is not true: we construct a language whose generating series is
algebraic (thus holonomic), but which is inherently weakly-ambiguous as a
Parikh automata language. Finally, we prove an effective decidability result
for the inclusion problem for weakly-unambiguous Parikh automata, and provide an upper-bound on to its complexity.
 \end{abstract}
\newpage
\section{Introduction}
This article investigates the link between \emph{holonomic}  (or \emph{D-finite}) power series and formal languages. We consider the classical setting in which this connection is established via the generating series~$L(x)=\sum_{n\geq 0} \ell_n x^n$ counting the number  $\ell_n$ of words of length~$n$ in a given language~$\mathcal L$.

On the  languages side, the Chomsky–Schützenberger hierarchy~\cite{chomschut63} regroups  languages in classes of increasing complexity: regular, context-free, context-sensitive and recursively enumerable. For power series, a similar hierarchy exists, consisting of the rational, algebraic and holonomic series. The first two levels of each hierarchy share a strong connection,
as the generating series of a regular (resp. unambiguous context-free) language is a rational (resp. algebraic) power series.

This connection has borne fruits both in formal language theory and in combinatorics. In combinatorics, finite automata and unambiguous grammars are routinely used to establish rationality and algebraicity of particular power series. In formal languages, this connection was (implicitly) used to give polynomial-time algorithms for the inclusion and universality tests for unambiguous finite automata~\cite{Stearns85}. In~\cite{flajolet87}, Flajolet uses the connection between unambiguous context-free grammars and algebraic series to prove the inherent ambiguity  of certain context-free languages, solving several conjectures with this tool. Using analytic criteria on the series (for instance, the existence of infinitely many singularities), he establishes that the series of these context-free languages is not algebraic. Hence these languages cannot be described by unambiguous context-free grammars and are therefore inherently ambiguous.

\medskip

In this article we propose to extend the connection to holonomic series. Holonomic series enjoy non-trivial closure properties whose algorithmic counterparts are actively studied in computer algebra. Our aim is to show that these advances can be leveraged to obtain non-trivial results in the formal languages and verification worlds.
The work of extending the connection was already  initiated by Massazza in~\cite{Massazza93}, where he introduces two families of languages, named RCM and LCL, whose generating series are holonomic. These classes are, however, not  captured by well-known models of automata, and this limits their appeal.  Recently, Castiglione and Massazza addressed this issue and conjectured that RCM contains the languages accepted by deterministic one-way reversal bounded machines (RBCM for short)~\cite{Castiglione17}; Massazza proved the result for RBCM for two subclasses of one-way deterministic RBCM~\cite{Massazza17,Massazza18}.
This conjecture hints that the class RCM is related to models of automata such as  RBCM, which are used in  program verification.

\medskip
Our first contribution is to prove a stronger version of this conjecture. We show that RCM and LCL  respectively correspond to the languages accepted  by weakly-unambiguous\footnote{We use the term weakly-unambiguous here to avoid a possible confusion with the class of unambiguous PA defined in~\cite{Cadilhac13} which is strictly contained in our class. Our notion of non-ambiguity is the standard one: every word has at most one accepting computation (this is detailed in Remark~\ref{rmk:comp}).} version of Parikh automata (PA, for short) \cite{Klaedtke03} and pushdown Parikh automata. In terms of RBCM, these classes correspond to unambiguous two-way RBCM and unambiguous one-way RBCM enriched with a stack. Parikh automata are also  commonly used in program verification. In view of the literature, these results might seem expected but they still require a careful adaptation of the standard techniques in the absence of a stack and become even more involved when a stack is added.

After having established the relevance of the classes of languages under study, we provide two consequences of the holonomicity of their associated generating series.

The first consequence follows Flajolet’s approach mentioned previously and gives criteria to establish the inherent weak-ambiguity for languages accepted by PA or pushdown PA, by proving that their generating series are not holonomic. These criteria are sufficient but not necessary; this is not surprising as the inherent ambiguity is undecidable for languages accepted by PA. Yet, the resulting method captures non-trivial examples with quite short and elegant proofs. In contrast, we give an example of inherently weakly-ambiguous PA language having a holonomic series (and therefore not amenable to the analytic method) for which we prove inherent weak-ambiguity \emph{by hand}. The proof is quite involved but shows the inherent ambiguity of this language for a much larger class of automata (i.e., PA whose semi-linear sets are replaced by arbitrary recursive sets). 

The second consequence is of an algorithmic nature. We focus on the inclusion
problem for weakly-unambiguous PA, whose decidability can be deduced from
Castiglione and Massazza's work~\cite{Castiglione17}. Here our contribution is
an effective decidability result: we derive a concrete bound~$B$, depending on
the size of the representation of the two PAs, such that the inclusion holds
if and only if the languages are included when considering words up to the
length~$B$. This bound $B$ is obtained by a careful analysis of the proofs
establishing the closure properties of holonomic series (in several
variables), notably under Hadamard product and specialisation. We do this by
controlling various parameters (order, size of the polynomial coefficients,
\ldots) of the resulting partial differential equations.

\section{Primer on holonomic power series in several variables}
\label{sec:primer}

{In this section, we introduce power series in several variables and the classes of rational, algebraic and holonomic power series. We recall the connection with regular and context-free languages via the notions of generating series in one or several variables.}

Let $\mathbb{Q}[x_1,\ldots,x_k]$ be the ring of polynomials in the variables $x_1,\ldots,x_k$ with coefficients in $\mathbb{Q}$ and let $\mathbb{Q}(x_1,\ldots,x_k)$ be the associated field of rational fractions.

The \intro{generating series of a sequence} $(f_n)_{n\in\mathbb{N}}$ is the (formal) power series in the variable~$x$ defined by $F(x)=\sum_{n\in\mathbb{N}}f_nx^n$. More generally, the generating series of a sequence $a(n_1,\ldots,n_k)$ is a multivariate (formal) power series in the variables $x_1,\ldots, x_k$ defined by $A(x_1, \ldots, x_k)=\sum_{(n_1, \ldots, n_k)\in\mathbb{N}^k}a(n_1,\ldots,n_k)x_1^{n_1}\ldots x_k^{n_k}$. 
In this article, we only consider power series whose coefficients belong to the field~$\mathbb Q$. 
The set of such $k$-variate power series is denoted $\mathbb{Q}[[x_1,\ldots,x_k]]$.
Power series are naturally equipped with a sum  and a product which generalize those of polynomials, for which $\mathbb{Q}[[x_1,\ldots,x_k]]$ is a ring. 
We use the bracket notation for the coefficient extraction: $[x_1^{n_1}\ldots x_k^{n_k}]A(x_1, \ldots, x_k) = a(n_1,\ldots,n_k)$. 
The \intro{support} of $A\in\mathbb{Q}[[x_1,\ldots,x_k]]$ is the set of $(n_1, \ldots, n_k)$ such that $[x_1^{n_1}\ldots x_k^{n_k}]A\not= 0$. 
The inverse $1/A$ of a series $A\in \mathbb{Q}[[x_1,\ldots,x_k]]$ is well-defined when its constant term $[x_1^0 \ldots x_k^0]A$ is not zero. 
For instance, the inverse of $A(x_1,x_2)=1-x_1x_2^2$ is $\frac{1}{1-x_1x_2^2} = \sum_{n\geq 0}x_1^nx_2^{2n}$. 

The \intro{generating series of a language} $\mathcal L$  over the alphabet $\Sigma=\{a_1, \ldots, a_k\}$  is the univariate power series  $L(x)=\sum_{w\in\mathcal L}x^{|w|}=\sum_{n\in\mathbb{N}}\ell_n x^n$, where $\ell_n$ counts the number of words of length $n$ in $\mathcal L$. Similarly the \intro{multivariate generating series} of $\mathcal L$ defined by
$L(x_{a_1}, \ldots, x_{a_k})=\sum_{(n_1, \ldots, n_k)\in\mathbb{N}^k}\ell(n_1,\ldots,n_k)x_{a_1}^{n_1}\ldots x_{a_k}^{n_k}$ where $\ell(n_1,\ldots,n_k)$ denotes the number of words $w$ in~$\mathcal L$ such that $|w|_{a_1}=n_1$, $|w|_{a_2}=n_2$, \ldots, and $|w|_{a_k}=n_k$, and $|w|_a$ denotes the number of occurrences of $a\in\Sigma$ in $w$. This way, we create one dimension per letter, so that each letter $a\in\Sigma$ has a corresponding variable $x_a$. %

Observe that the univariate generating series of a language is exactly $L(x,\ldots,x)$, obtained by setting each variable to $x$ in its multivariate generating series.

\begin{example}
The generating series of the language $\mathcal P$ of well-nested parentheses defined by the grammar $S\rightarrow aSbS+\varepsilon$
is $P(x_a,x_b) = \frac{1-\sqrt{1-4x_ax_b}}{2x_ax_b}$ and its counting series is\footnote{As for the inverse, the square root of a power series with nonzero constant term can be defined using the usual  Taylor formula %
$\sqrt{1-x} = \sum _{n\geq 0} \frac {1}{(1-2\,n) \, {4}^n} {2\,n\choose n} x^n$.} $P(x) = \frac{1-\sqrt{1-4x^2}}{2x^2}$.

\end{example}

		A power series $A(x_1, \ldots, x_k)=\sum_{n_1, \ldots, n_k}a(n_1, \ldots, n_k)x_1^{n_1}\ldots x_k^{n_k}$ is \intro{rational} if it satisfies an equation of the form:
		$
		P(x_1, \ldots, x_k)A(x_1, \ldots, x_k)=Q(x_1, \ldots, x_k),
		$ with $P, Q \in \mathbb Q[x_1, \ldots, x_k]$ and $P \neq 0$.
		The generating series (both univariate and multivariate) of regular languages (i.e., languages accepted by a finite state automaton) are rational power series~\cite{Bousquet-Melou06}. 
        It is well-known that the generating series
		can be effectively computed from a deterministic automaton accepting the language (see for instance \cite[\S I.4.2]{Flajsedg} for a detailed proof).
		For example, the multivariate generating series of the regular language $(abc)^*$ is $\frac{1}{1-x_ax_bx_c}=\sum_{n\geq0} x_a^nx_b^nx_c^n$. Its univariate generating series is $\frac{1}{1-x^3}$.
		
		The connection between rational languages and rational power series is not tight. For instance, the generating series of
		the non-regular context-free language $\{a^n b^n \;:\; n \geq 0\}$ is $\frac{1}{1-x_a x_b}$, which is rational. In fact, it has the same generating series as the regular language $(ab)^*$. Also there exist rational power series with coefficients in $\mathbb{N}$ which are not the generating series of any rational language. It is the case for $\frac{x+5x^2}{1+x-5x^2-125x^3}$ as shown in \cite{Bousquet-Melou06}. 

		A power series $A(x_1, \ldots, x_k)$ is  \intro{algebraic} if there exists a non-zero polynomial $P\in\mathbb Q[x_1, \ldots, x_k, Y]$ such that $P(x_1, \ldots, x_k, A(x_1, \ldots, x_k))=0$. All rational series are algebraic.

		\begin{example} 
			The series $\frac{1}{1-x_1x_2}=\sum_{n \geq 0} x_1^nx_2^n$ is rational, as it satisfies the equation $(1-x_1x_2)A(x_1,x_2)=1$. The series $A(x_1,x_2)=\sqrt{1-x_1x_2}$ is algebraic but not rational, since $A(x_1,x_2)^2 + (x_1x_2-1) = 0$ and there is no similar algebraic equation of degree~$1$.
		\end{example}
		The reader is referred to \cite{Stanley,Flajsedg} for a  detailed account on  rational and algebraic series.

		In the same manner that rational series satisfy 
		linear equations and algebraic series satisfy polynomial equations, holonomic series satisfy
	    linear differential equations with polynomial coefficients. To give a precise definition, we need to introduce the formal partial derivation of power series.
	    The differential operator
	$\partial_{x_i}$ with respect to the variable $x_i$ is defined by	\[
		\partial_{x_i}A(x_1,\ldots, x_k) = \sum_{n_1,\ldots,n_k}n_i\,a(n_1,\ldots,n_k)x_1^{n_1}\ldots x_{i-1}^{n_{i-1}}x_{i}^{n_{i}\,-\,1}x_{i+1}^{n_{i+1}}\ldots x_k^{n_k}.
		\]
			The composed operator $\partial_{x_i}^j$ is inductively defined for $j\geq 1$ by $\partial_{x_i}^1=\partial_{x_i}$ and $\partial_{x_i}^{j+1}=\partial_{x_i}\circ\partial_{x_i}^j$.
		
\begin{definition}[see \cite{STANLEY1980, LIPSHITZ1989}]\label{def:holonom}
A power series $A(x_1, \ldots, x_k)$ is  \intro{holonomic} or \intro{D-finite}{\footnote{{A priori, these notions differ: 
a function $\, A(x_1,\ldots,x_k)$ is called \emph{D-finite\/} 
if all its partial derivatives $\, \partial_{x_1}^{n_1}\cdots \partial_{x_k}^{n_k} \cdot A $ 
generate a finite dimensional space over $\, \mathbb{Q}(x_1,\ldots,x_k)$, 
and \emph{holonomic\/} if the functions
$\, x_1^{\alpha_1}\,  \cdots \, x_k^{\alpha_k} \partial_{x_1}^{\beta_1}\,  \cdots \, \partial_{x_k}^{\beta_k} \, \cdot\,  A $
subject to the constraint 
$\, \alpha_1\,  +\,  \cdots \, + \alpha_k \, + \beta_1 \, + \cdots \, + \beta_k\,  \leq\,  N$ 
span a vector space whose dimension over~$\mathbb{Q}$ grows like $\, O(N^k)$. The
equivalence of these notions is proved by deep results of 
Bern{\v{s}}te{\u\i}n~\cite{Bernstein} and 
Kashiwara~\cite{Kashiwara,Takayama}.}}} if the $\mathbb Q(x_1, \ldots, x_k)$-vector space spanned by the family $\{\partial_{x_1}^{i_1}\ldots\partial_{x_k}^{i_k}A(x_1,\ldots,x_k)\ :\ (i_1, \ldots, i_k)\in\NN^k\}$ has a finite dimension. 
Equivalently, for every variable $z \in \{x_1,\ldots,x_k\}$, $A(x_1,\ldots,x_k)$ satisfies a linear differential equation of the form $P_{r}(x_1, \ldots, x_k)\partial_{z}^{r}A(x_1,\ldots,x_k)+\ldots + P_{0}(x_1, \ldots, x_k)A(x_1,\ldots,x_k)=0$, where the $P_{i}$'s are polynomials of $\mathbb Q[x_1, \ldots, x_k]$ with $P_{r}\neq 0$.
\end{definition}

		\begin{example}\label{ex:L3} A simple example of holonomic  series is $A(x)=e^{x^2}=\sum_{n \geq 0} \frac{x^{2n}}{n!} $. It is holonomic (in one variable) since it verifies $\partial_x A(x)-2xA(x)=0$.
		
		For a more involved example, consider the language $\mathcal{L}_3=
		\{w \in \{a,b,c\}^* \;:\; |w|_a=|w|_b=|w|_c \}$, containing the words having the same number of occurrences of $a$'s, $b$'s and $c$'s.
	This language is classically not context-free. Moreover there are $\binom{3n}{n,n,n}$ words of length $3n$ in~$\mathcal{L}_3$
and the power series $\sum_{n} \binom{3n}{n,n,n} x^n$ is transcendental~\cite[\S7]{flajolet87}. Its multivariate generating series $L_3(x_a,x_b,x_c)$ is equal to $\sum _{n=0}^{\infty }{\frac { \left( 3\,n \right) !\, \left( x_a x_b x_c
 \right) ^{n}}{ \left( n! \right) ^{3}}}$ and satisfies the partial differential equation: 
		\[
		(27x_a^2x_bx_c-x_a)\partial_{x_a}^2f(x_a,x_b,x_c) + (54x_ax_bx_c-1)\partial_{x_a}f(x_a,x_b,x_c) +6x_bx_cf(x_a,x_b,x_c)=0,
		\]
		and the symmetric ones for the other variables $x_b$ and $x_c$.
		\end{example}
		
		Holonomic series are an extension of the hierarchy we presented, as stated in the following  proposition  (see \cite{comtet} for a proof, and~\cite{BCLSS07} for bounds, algorithms and historical remarks).

		\begin{proposition}
			Multivariate algebraic power series are holonomic.
		\end{proposition}

In the univariate case\footnote{The
generalization of this equivalence to the multivariate case is not
straightforward (see \cite{LIPSHITZ1989} for more details) and will not be used in this article.}, a power series $A(x)=\sum_{n}a_nx^n$ is
holonomic if and only if its coefficients satisfy a linear recurrence of the
form $p_{r}(n)a_{n+r}+\ldots+p_0(n)a_n=0$, where every $p_i$ is a polynomial
with rational coefficients~\cite[Th.~1.2]{STANLEY1980}.  

In the sequel, except for Section~\ref{sec:complexity}, we rely on the closure properties of the holonomic series and will not need to go back to Definition~\ref{def:holonom}. We now focus  on these closure properties.

\begin{proposition}[{\cite{STANLEY1980}}]
	\label{prop:holonomicclosuresumproduct}
	Multivariate holonomic series are closed under sum and product.
\end{proposition}

	Holonomic series are also closed under substitution by algebraic series as long as the resulting series is well-defined\footnote{Note that the substitution of a power series into another power series might not yield a power series: for instance substituting $x$ by $1+y$ in $\sum_{n \geq 0} x^n$ does not result in a power series as the constant term would be infinite.}.
		
		\begin{proposition}[{\cite[Prop. 2.3]{LIPSHITZ1989}}]\label{prop:substitution}
		    Let $A(x_1, \ldots, x_k)$ be a power series and let 
		    $g_i(y_1, \ldots, y_\ell)$ be algebraic power series such that $B(y_1, \ldots, y_\ell)=A(g_1(y_1, \ldots, y_\ell), \ldots, g_k(y_1, \ldots, y_\ell))$ is well-defined as a power series.  If $A$ is holonomic, then $B$ is also holonomic.
		\end{proposition}
	    A sufficient condition for the substitution to be valid is that $g_i(0, \ldots, 0)=0$ for all $i$ (see \cite[Th.~2.7]{STANLEY1980}). For the case $g_1=\cdots=g_k=1$, called the \intro{specialization to 1},\reviewer{A detailed proof for the specialization  can be found in Appendix Prop.~\ref{prop:specialisation_holonome}.} a sufficient condition is that for every index $(i_1, \ldots, i_k)$, $[x_1^{i_1}\ldots x_k^{i_k}]A$ is a polynomial in $y_1, \ldots, y_\ell$.

The Hadamard product is the coefficient-wise multiplication of power series. If $A(x_1, \ldots, x_k)$ and $B(x_1,\ldots,x_k)$ are the generating series of the sequences $a(n_1,\ldots,n_k)$ and $b(n_1,\ldots,n_k)$,
 the \intro{Hadamard product} $A\odot B$ of $A$ and $B$ is the power series defined by
\[\hfill A\odot B (x_1, \ldots, x_k) = \sum_{n_1, \ldots, n_k\in\mathbb{N}^k}a(n_1,\ldots,n_k)b(n_1,\ldots,n_k)x_1^{n_1}\ldots x_k^{n_k}.\hfill\]
Observe that the support of $F\odot G$ is the intersection of the supports of $F$ and $G$. 

\begin{theorem}[\cite{LIPSHITZ1988}]\label{thm:hadamard}
	Multivariate holonomic series are closed under Hadamard product.
\end{theorem}

\begin{example} \label{holonomicexample}
The generating series of the language $\mathcal L_3$ of Example~\ref{ex:L3}, which is not context-free, can be expressed using the Hadamard product: since $\frac{1}{1-x_ax_bx_c}$ is the support series of the subset $\{(n, n, n): n\in\mathbb{N}\}$, and since $\frac{1}{1-(x_a+x_b+x_c)}$ is the multivariate series of all the words on $\{a,b,c\}$, we have $L_3(x_a,x_b,x_c)=\frac{1}{1-(x_a+x_b+x_c)}\odot\frac{1}{1-x_ax_bx_c}$, which is not algebraic. 
\end{example}
 
One of our main technical contribution is to provide bounds on the sizes of the polynomials in the differential equations of the holonomic representation of the Hadamard product of two rational series\reviewer{See Section~\ref{sec:hadamard} in the Appendix.}
$\frac{P_1}{Q_1}$ and $\frac{P_2}{Q_2}$:
we prove that their maxdegree is at most $(kM)^{\OO(k)}$ and
that the logarithm of their largest coefficient is at most $(kM)^{\OO(k^2)}(1+\log\normF)$, where $M$ (resp. $\normF$) is the maxdegree plus one (resp. largest coefficient) in $P_1$, $Q_1$, $P_2$ and $Q_2$.
 
\section{Weakly-unambiguous Parikh automata}
\label{sec:wupa}

In this section, we introduce weakly-unambiguous Parikh automata and show that their multivariate generating series are holonomic. We establish that they accept the same languages as unambiguous two-way reversal bounded counter machines \cite{Ibarra78}. Finally, we prove that the class of accepted languages coincides with Massazza's RCM class \cite{Massazza93,Castiglione17}.

Parikh automata (PA for short) were introduced in
\cite{Klaedtke02,Klaedtke03}. Informally, a PA is a finite automaton whose
transitions are labeled by pairs $(a,\vect{v})$ where $a$ is a letter of the
input alphabet and $\vect{v}$ is a vector in $\mathbb{N}^d$. A run
$q_0 \era{a_1,\vect{v_1}} q_1 \era{a_2,\vect{v_2}} q_2 \cdots
q_{n-1} \era{a_{n},\vect{v_n}} q_n$ computes the word $a_1\cdots
a_{n}$ and the vector $\vect{v_1} + \cdots + \vect{v_{n}}$ where the sum is
done component-wise. The acceptance condition is given by a set of final
states and a semilinear set of vectors. A run is accepting if it reaches a
final state and if its vector belongs to the semilinear set.

\noindent\begin{minipage}{.35\textwidth}
\begin{tikzpicture}
\node (init) at (0,.8) {};
\node[draw,circle] (q0) at (0,0) {$0$};
\node[draw,circle] (q1) at (2,0) {$1$};
\node[draw,circle,accepting] (q2) at (4,0) {$2$};

\draw[->,thick] (q0) -- node[above] {\small $a\big(\substack{1\\0\\0}\big)$} (q1);
\draw[->,thick] (q1) -- node[above] {\small $a\big(\substack{1\\0\\0}\big)$} (q2);
\draw[->,thick] (q1) edge[loop, above] node[above] {\small $a\big(\substack{1\\0\\0}\big), b\big(\substack{0\\1\\0}\big), c\big(\substack{0\\0\\1}\big)$} (q1);

\draw[->, thick] (init) -- (q0);
\end{tikzpicture}
\end{minipage}
\begin{minipage}{.64\textwidth}
The PA depicted on the left, equipped with the semilinear constraint $\{ (n_1,n_2,n_1+n_2) \;:\; n_1,n_2 \geq 0\}$, accepts the set of words $w$ over $\{a,b,c\}$ that start and end with $a$ and that are such that $|w|_a+|w|_b=|w|_c$.
\end{minipage}

\vspace{-.4cm}
\subsection{Semilinear sets and their characteristic series}
\label{sec:semilinear}

A set $L \subseteq \mathbb{N}^d$ is \intro{linear} if it is of the form
$
\vect{c}+P^* := \{ \vect{c} + \lambda_1 \vect{p_1} + \cdots + \lambda_k \vect{p_k} \;\mid\; \lambda_1,\ldots,\lambda_k \in \mathbb{N} \}
$,
where $\vect{c} \in \mathbb{N}^d$ is the \intro{constant} of the set  and $P = \{\vect{p_1},\ldots,\vect{p_k}\} \subset \mathbb{N}$ its set of \intro{periods}. 
A set $C \subseteq \mathbb{N}^d$ is \intro{semilinear} if it is a finite union of linear sets. 
For example, the semilinear set $\{ (n_1,n_2,n_1+n_2) \;:\; n_1,n_2 \geq 0\}$ is in fact a linear set $(0,0,0) + \{ (1,0,1),(0,1,1) \}^*$.

In \cite{EILENBERG1969,Ito69}, it is shown that every semilinear set admits an unambiguous presentation. A presentation $\vect{c}+P^*$ with $P=\{\vect{p_1},\ldots,\vect{p_k}\}$ of a linear set $L$ is unambiguous if for all $\vect{x} \in L$, the $\lambda_i$'s such that $\vect{x} = \vect{c} + \lambda_1 \vect{p_1} + \cdots + \lambda_k \vect{p_k}$ are unique. An \intro{unambiguous presentation of a semilinear set} is given by a disjoint union of unambiguous linear sets. A bound on the size of the equivalent unambiguous presentation is given in \cite{Chistikov16}.

Semilinear sets are ubiquitous in theoretical computer science and admit numerous characterizations. They are the rational subsets of the commutative monoid $(\mathbb{N}^d,+)$, the unambiguous rational subset of $(\mathbb{N}^d,+)$ \cite{EILENBERG1969,Ito69}, the Parikh images of context-free languages \cite{Parikh1966}, the sets definable in Presburger arithmetic 
\cite{Pres29}, the sets defined by boolean combinations of linear inequalities, equalities and equalities modulo constants. 

For example, the semilinear set $\{ (2n,3n,5n) \;:\; n \geq 0 \}$ is equal to $(0,0,0) + \{(2,3,5)\}^*$. It is also the Parikh image of the regular (hence context-free) language $(aabbbccccc)^*$. In $(\mathbb{N},+)$, it is defined by the Presburger formula $\phi(x,y,z) = \exists w, x=w+w \wedge y=w+w+w \wedge z=w+w+w+w+w$. Finally it is characterized in $\mathbb{N}^3$ by the equalities $3x=2y$ and $5x=2z$.
   
For a semilinear set $C \subseteq \mathbb{N}^d$, we consider  
its  \intro{characteristic generating series} 
$C(x_1,\ldots,x_d) = \sum_{(i_1,\ldots,i_d) \in C} x_1^{i_1} \cdots x_d^{i_d}$. It is well-known \cite{EILENBERG1969,Ito69} that this power series is rational\footnote{%
The characteristic series  of an unambiguous linear set $\vect{c} + P^*\subseteq \mathbb{N}$ with $P=\{\vect{p_1},\ldots,\vect{p_k}\}$ is 
$x_1^{\vect{c}(1)}\ldots x_d^{\vect{c(d)}}\prod_{i=1}^k
\big(1-x_1^{\vect{p_i}(1)}\cdots x _d^{\vect{p_i}(d)}\big)^{-1}
$
and hence is rational. As an unambiguous semilinear set is the disjoint union of unambiguous linear sets, its characteristic series is the sum of their series and it is therefore rational.}.

\subsection{Weakly-unambiguous PAs and their generating series}
\label{sec:PA GS}

We now introduce PA and their weakly-unambiguous variant. We discuss the relationship with the class of unambiguous PA introduced by Cadilhac et al. in \cite{Cadilhac13} and the closure properties of this class.
	
A \intro{Parikh automaton} of dimension $d \geq 1$ is a tuple $\mathcal A=(\Sigma, Q, q_I, F, C, \Delta)$ where $\Sigma$ is the  alphabet, $Q$ is the set of states, $q_I \in Q$ is the initial state, $F \subseteq Q$ is the set of final states, $C \subseteq \mathbb{N}^d$ is the semilinear constraint and $\Delta \subseteq Q \times (\Sigma \times \mathbb{N}^d) \times Q$ is the transition relation.

A run of the automaton is a sequence $q_0 \era{a_1,\vect{v_1}} q_1 \era{a_2,\vect{v_2}} q_2 \cdots q_{n-1} \era{a_n,\vect{v_n}} q_n$
where for all $i \in [1,n]$, $(q_{i-1},(a_i,\vect{v_i}),q_i)$ is a transition in $\Delta$. The run is labeled by the pair $(a_1\cdots a_n, \vect{v_1}+\cdots+\vect{v_n}) \in \Sigma^* \times \mathbb{N}^d$. It is accepting if $q_0=q_I$, the state $q_n$ is final and if the vector $\vect{v_1}+\cdots+\vect{v_n}$ belongs to $C$. The word $w$ is then said to be accepted by $\mathcal{A}$. The language accepted by $\mathcal{A}$ is denoted by $\mathcal L(\mathcal A)$. %

To define a notion of size for a PA, we assume that the constraint set is given by an unambiguous presentation
$\uplus_{i=1}^p \vect{c_i} + P_i^*$. We denote by \intro{$|\AA|$} $ := |Q|+|\Delta|+p+\sum_{i} |P_i|$ and by \intro{$\|\AA\|_{\infty}$} the maximum coordinate of a vector appearing in $\Delta$, the $\vect{c_i}$'s and the $P_i$'s.  

\begin{definition}\label{def:weakly-unambiguous}
	A Parikh automaton is said to be \intro{weakly-unambiguous} if for every word there is at most one accepting run.
\end{definition}

A language is \intro{inherently weakly-ambiguous} if it is not accepted by any weakly-unambiguous PA.
The language $\mathcal{S}$ (defined in Section~\ref{sec:2 examples}) is an example\footnote{Using the equivalences between weakly-unambiguous PA and RCM established in Proposition~\ref{prop:unambparbcm} and PA and RBCM \cite{Klaedtke02}, it also gives an example of a language accepted by a RBCM with a non-holonomic generating series (strengthening Theorem 12 of \cite{Castiglione17}) and a witness for the strict inclusion of RCM in RBCM announced in Theorem 11 of \cite{Castiglione17}. Remark that their proof of this theorem only shows that there exists no recursive translation from RBCM to RCM.} of a language accepted by a non-deterministic PA which is inherently weakly-ambiguous.

\begin{remark}
\label{rmk:comp}
    We consider here the standard notion of unambiguity for 
finite state machines. However we decided to use the name \emph{weakly-unambiguous} to avoid the confusion with the class of unambiguous PA introduced by 
Cadilhac et al. in \cite{Cadilhac12a, Cadilhac13}. Their notion of unambiguity is more restrictive than ours:
 they call a Parikh automaton \emph{unambiguous} if the underlying automaton on letters, where the vectors have been erased, is unambiguous. Clearly such automata are weakly-unambiguous. However the converse is not true.\reviewer{More details on this automaton in Appendix~Sec.~\ref{app:rmk-comp}} \vspace{1mm}
\begin{minipage}{.4\textwidth}
\hspace{-.5cm}
\begin{tikzpicture}
\node (init) at (0,-.8) {};
\node[draw,circle] (q0) at (0,0) {$0$};
\node[draw,circle,accepting] (q1) at (2,0) {$1$};
\node[draw,circle,accepting] (q2) at (4,0) {$2$};

\draw[->,thick] (q0) edge node[below] {\small $a\big(\substack{0\\1\\0}\big), b\big(\substack{0\\0\\1}\big)$} (q1);
\draw[->,thick] (q1) edge node[below] {\small $a\big(\substack{0\\0\\0}\big), b\big(\substack{0\\0\\0}\big)$} (q2);

\draw[->,thick] (q0) edge[loop, above] node[above] {\small $c\big(\substack{1\\0\\0}\big)$} (q0);
\draw[->,thick] (q1) edge[loop, above] node[above] {\small $a\big(\substack{0\\1\\0}\big), b\big(\substack{0\\0\\1}\big)$} (q1);
\draw[->,thick] (q2) edge[loop, above] node[above] {\small $a\big(\substack{0\\0\\0}\big), b\big(\substack{0\\0\\0}\big)$} (q1);

\draw[->, thick] (init) -- (q0);
\end{tikzpicture}
  
 \end{minipage}
 \begin{minipage}{.59\textwidth}
 Consider 
 the language $\mathcal{L} = \{c^nw \;:\; w = x_1x_2\cdots x_m \in\{a,b\}^*\ \wedge\ n>0 \wedge\ |x_1x_2\ldots x_n|_{a} < |x_1x_2\ldots x_n|_{b}\} $ over the alphabet $\{a,b,c\}$. Using results from \cite{Cadilhac13}, one can show that it is not recognized by any unambiguous Parikh automata. However, it is accepted by the weakly-unambiguous automaton depicted on the left with the semilinear $\{ (n_1,n_2,n_3) \;:\; n_1=n_2+n_3 \;\text{and}\; n_2 < n_3 \}$.
  \end{minipage}

\vspace{1mm}

 The lack of expressivity of unambiguous PAs is counter-balanced by their closure under boolean operations, which is explained by their link with a class of determinisitic PA \cite{Filliot19}.
\end{remark}

Using a standard product construction when \reviewer{The construction for the intersection is given in Appendix Sec.\ref{app:inter-pa}.} 
 the vectors are concatenated and using the concatenation of the constraints, it is easy to show that weakly-unambiguous PA are closed under intersection. In \cite{Castiglione17}, the authors claim that the class\footnote{Actually, their claim is for the class RCM, which we will show to be equivalent in Section \ref{sec:rcm}.} is closed under union. However their construction has an irrecoverable flaw and we do not know if weakly-unambiguous PA are closed under union or under complementation.\reviewer{A counter-example to the construction of \cite{Castiglione17} is given in Apprendix Section~\ref{app:counter-example-massazza}.} 

We now give a very short proof of the fact that weakly-unambiguous languages in $\PA$ have holonomic generating series. The idea of the proof can be traced back to \cite{LIPSHITZ1988}. A similar proof was given in \cite{Massazza93} for languages in the class $\RCM$ but using the closure under algebraic substitutions instead of specialization (see Remark~\ref{rmk:Massazza}). %

Our approach puts into light a different multivariate power series associated with a weakly-unambiguous PA $\AA$ of dimension $d$. The multivariate \intro{weighted generating series} $G(x,y_1,\ldots,y_d)$ of $\AA$ is such that for all index $(n,i_1,\ldots,i_d)$,  $[x^ny_1^{i_1}\ldots y_d^{i_d}]G$  counts the number of words of length $n$ accepted by $\mathcal{A}$ with a run labeled 
by the vector $(i_1,\ldots,i_d)$.

\begin{proposition}\label{prop:unambPAholonomic}
	The generating series of the language recognized by a weakly-unambiguous Parikh automaton is holonomic.
\end{proposition}

\begin{proof}
Let $\mathcal{A}$ be a weakly-unambiguous PA with a constraint set $C \subseteq \mathbb{N}^d$. We first prove that its 
weighted series $G(x,y_1,\ldots,y_d)$ is holonomic. As holonomic series are closed under Hadamard product (see Theorem~\ref{thm:hadamard}), it suffices to express $G$ as the Hadamard product of two rational series $\overline{A}$ and $\overline{C}$ in the variables $x,y_1,\ldots,y_{d}$.

The first series $\overline{A}(x,y_1,\ldots,y_d)$ is such that for all $n,i_1,\ldots,i_d \geq 0$, $[x^ny_1^{i_1}\ldots y_d^{i_d}]\overline{A}$ counts the number of runs of $\mathcal{A}$ starting in $q_I$, ending in a final state and labeled with a word of length $n$ and the vector $(i_1,\ldots,i_d)$. Note that we do not require that $(i_1,\ldots,i_d)$ belongs to $C$. As this series simply counts the number of runs in an automaton, its rationality is proved via the standard translation of the automaton into a linear system of equations.

For the second series, we take $\overline{C}(x,y_1,\ldots,y_d) \eqdef \frac{1}{1-x}C'(y_1,\ldots,y_d)$ where $C'$ is the  support series of $C$, which is rational (see \cite{EILENBERG1969,Ito69}). A direct computation yields that for all $n,i_1,\ldots,i_d \geq 0$, $[x^ny_1^{i_1}\ldots y_d^{i_d}] \overline{C}$ is equal to $1$ if $(i_1,\ldots,i_d)$ belongs to $C$ and $0$ otherwise.

The Hadamard product of $\overline{A}$ and $\overline{C}$ counts the number of runs accepting a word of length~$n$ with the vector $(i_1,\ldots,i_d)$. As $\mathcal{A}$ is weakly-unambiguous, this quantity is equal to the number of
words of length $n$ accepted with this vector. Hence $G=\bar A \odot \bar C$.

The univariate series  $A(x)$ of $\mathcal{A}$ is equal to $G(x,1,\ldots,1)$. Indeed, for all $n\geq 0$, $[x^n]G(x,1,\ldots,1)$ is the sum over all vector $\vect{i} \in \mathbb{N}^d$ of the number of words of length $n$ accepted with the vector $\vect{i}$. As $\mathcal{A}$ is weakly-unambiguous, each word is accepted with at most one vector and this sum is therefore equal to the total number of accepted words of length $n$. 
Thanks to Proposition~\ref{prop:substitution}, $A(x)=G(x,1,\ldots,1)$ is holonomic.
\end{proof}

\subsection{Equivalence with unambiguous reversal bounded counter machines}

A \intro{$k$-counter machine} \cite{Ibarra78} is informally a Turing machine with one read-only tape that contains the input word, and $k$ counters. Reading a letter $a$ on the input tape, in a state $q$, the machine can check which of its counters are zero, increment or decrement its counters, change its state, and move its read head one step to the left or right, or stay on its current position. Note that the machine does not have access to the exact value of its counters. A $k$-counter machine is said \intro{$(m,n)$-reversal bounded} if its reading head can change direction between left and right at most $m$ times, and if every counter can alternate between incrementing and decrementing at most $n$ times each.  
Finally, a \intro{reversal bounded counter machine} (RBCM) is a $k$-counter machine which is $(m,n)$-reversal bounded for some $m$ and $n$. A RBCM is \intro{unambiguous} if for every word there is at most one accepting computation.

RBCM are known to recognize the same languages as Parikh automata (see \cite{Klaedtke03,Klaedtke02}). This equality does not hold anymore for their deterministic versions \cite[Prop. 3.14]{Cadilhac12a}. However, the proof of the equivalence for the general case can be slightly modified to preserve unambiguity.

\begin{restatable}{proposition}{equivparbcm}\label{prop:unambparbcm}
The class of languages accepted by unambiguous $\RBCM$ and weakly-unambiguous $\PA$ coincide.
\end{restatable}

\begin{proof}[Proof sketch]\reviewer{A detailed proof can be found in Appendix~Sec.~\ref{app:unambparbcm}.}
Unambiguous RBCMs are shown to be equivalent to one-way unambiguous RBCMs. In turn these are shown to be equivalent 
to weakly-unambiguous PA with $\varepsilon$-transitions, which in turn are equivalent to weakly-unambiguous PA. This $\varepsilon$-removal step needs to be adapted to preserve weak-unambiguity.
\end{proof}

\subsection{Equivalence with RCM}
\label{sec:rcm}

If we fix an alphabet $\Gamma=\{a_1,\ldots,a_d\}$, we can associate with every semilinear set $C$ of dimension $d$, the language $[C]=\{ w \in \Gamma^* \;:\; (|w|_{a_1},\ldots,|w|_{a_d}) \in C\}$ of words whose numbers of occurrences of each letter satisfy the constraint expressed by $C$.  For instance, if we take the semilinear set $C_0 = \{ (n,m,n,m) \;:\; n,m \geq 0 \}$ and the alphabet $\{a,b,c,d\}$, $[C_0]$ consists of all words having as many $a$'s as $c$'s and as many $b$'s as $d$'s.

A language $\cal L$ over $\Sigma$ belongs to \intro{$\textnormal{RCM}$} if there exist  a regular language $\cal R$ over $\Gamma=\{a_1, \ldots, a_d\}$, a semilinear set\footnote{Our definition may seem a little more general than Massazza's, which only uses semilinear constraints without quantifiers, but it can be shown that the classes are equivalent.}\reviewer{See Appendix Sec.~\ref{app:eq-def-RCM} for a proof of the footnote.} $C \subseteq \mathbb N^d$ and a length preserving morphism $\mu:\Gamma^*\longrightarrow\Sigma^*$, that is injective over $\mathcal{R}\cap [C]$, so that $\mathcal{L}=\mu(\mathcal{R}\cap [C])$.
For example, $\mathcal L_{abab}=\{a^nb^ma^nb^m : n, m\in\mathbb N\}$ can be shown to be in RCM by taking $\Gamma=\{a,b,c,d\}, \Sigma=\{a,b\}$, $\mu(a)=\mu(c)=a, \mu(b)=\mu(d)=b$, $\mathcal{R}=a^*b^*c^*d^*$ and the semilinear set $C_0$ defined in the previous paragraph.

\begin{restatable}{theorem}{eqparcm}\label{prop:rcmpa}
$\mathcal L\in \RCM$ iff $\mathcal L$ is recognized by a weakly-unambiguous Parikh automaton.
\end{restatable}
\begin{proof}[Proof sketch]\reviewer{See Appendix Sec.~\ref{app:rcm-wupa} for the complete proof.}Every language in RCM can be accepted by a weakly-unambiguous PA that guesses the underlying word over $\Gamma$: the weak-unambiguity is guaranteed by the injectivity of the morphism. Conversely a language accepted by a weakly-unambiguous PA is in RCM by taking for  $\mathcal R$ the set of runs of the PA and translating the constraint: the injectivity of the morphism is guaranteed by the weak-unambiguity of the PA.
\end{proof}

In \cite{Castiglione17}, the authors conjectured that the class RCM contains the one-way deterministic RBCM. %
From Theorem~\ref{prop:rcmpa} and Proposition~\ref{prop:unambparbcm} we get a stronger result:

\begin{corollary}
The languages in RCM are the languages accepted by unambiguous RBCM.
\end{corollary}

\subsection{Weakly-unambiguous pushdown Parikh automata}
\label{sec:wuppa}	

\newcommand{\PPA}{$\mathbb{P}\mathrm{A}$\xspace} 
A \intro{pushdown Parikh automaton} (\PPA for short) is  a PA where the finite automaton is replaced by a pushdown automaton. A \intro{weakly-unambiguous \PPA} has at most one accepting run for each word.  Most results obtained previously  can be adapted for weakly-unambiguous \PPA.\reviewer{The proofs of non-closure are in Appendix Prop.~\ref{prop:non-closure-wuppa}} However, unsurprisingly, the class of languages accepted by weakly-unambiguous \PPA is not closed under union and  intersection. This can be shown using the inherent weak-ambiguity of the language $\mathcal{D}$ proved in Section~\ref{sec:2 examples}. The closure under complementation is left open.

\begin{proposition}\label{prop:PPA_holonomic}
The generating series of a weakly-unambiguous \PPA is holonomic.
\end{proposition}

\begin{proof}%
The proof is almost identical to the proof of Proposition~\ref{prop:unambPAholonomic}. The only difference is that
the series $\overline{A}$ is algebraic and not rational. Indeed it counts the number of runs in a pushdown automaton and the language of runs is a deterministic context-free language even if the pushdown automaton is not deterministic. 
\end{proof}

\reviewer{More details can be found in Appendix Sec.~\ref{app:parikh-tree}.}
Remark that using the same techniques, we can prove that the generating series of weakly-unambiguous Parikh tree automata are holonomic. As we proved that all these series are also generating series of \PPA's, we do not elaborate on this model in this extended abstract.

RBCMs can be extended with a pushdown storage to obtain a \intro{RBCM with a stack} \cite{Ibarra78}.

\begin{restatable}{theorem}{weakpparbcm}
Weakly-unambiguous \PPA are equivalent to unambiguous one-way RBCM with a stack. 
\end{restatable}

\begin{proof}[Proof sketch]\reviewer{A detailed proof can be found in Appendix Sec.~\ref{sec:C}.}
We first establish that unambiguous one-way RBCM with a stack are equivalent
to weakly-unambiguous \PPA with $\varepsilon$-transitions. Contrarily to the PA case, the removal of $\varepsilon$-transitions is quite involved and uses weighted context-free grammars.  
\end{proof}

The class\footnote{In \cite{Massazza93}, LCL is defined without the injective morphism but we adapt it following \cite{Castiglione17}.} \intro{LCL} of \cite{Massazza93} is defined as RCM is, except that the regular language is replaced by an unambiguous context-free\footnote{Using deterministic context-free languages instead of unambiguous ones in the definition of LCL would result in the same class.} language. Similarly to the PA case, one can prove:

\begin{restatable}{proposition}{ppalcl}
LCL is the set of languages accepted by weakly-unambiguous \PPA.
\end{restatable}

\section{Examples of inherently weakly-ambiguous languages}\label{sec:non-ambiguity}

\reviewer{See Appendix Prop.~\ref{prop:decidability-weak-unambiguity} for a description of the algorithm.}There is a polynomial-time algorithm to decide whether a given PA is weakly-unambiguous. But inherent weak-unambiguity is undecidable, as a direct application of a general theorem from~\cite{Greibach1968}.\reviewer{Undecidability is proved in Appendix Prop.~\ref{prop:undecidability-inherent-weak-ambiguity}} This emphasises that inherent weak-ambiguity is a difficult problem in general. 

\subsection{Two examples using an analytic criterion}\label{sec:2 examples}
 Following an idea from Flajolet~\cite{flajolet87} for context-free languages, the link between weakly-unambiguous PA and holonomic series yields sufficient criteria to establish inherent weak-ambiguity, of analytic flavor: the contraposition of Proposition~\ref{prop:unambPAholonomic} indicates that if $\mathcal L$ is recognized by a PA but its generating series is not holonomic, then $\cal L$ is inherently weakly-ambiguous.
  Hence, any criterion of non-holonomicity can be used to establish the inherent weak-ambiguity. Many such criteria can be obtained when considering the generating series as analytic functions (of complex variables). See~\cite{flajolet87, flajolet2005} for several examples. For the presentation of this method in this extended abstract,  we only rely on the following property:
  
  \begin{proposition}[\cite{STANLEY1980}]\label{thm:finite singularities}
  A holonomic function in one variable has finitely many singularities.
  \end{proposition}
 
Our first example is the language $\mathcal D$, defined over the alphabet $\{a,b\}$ as follows:
\[\mathcal D= \{\underline{n_1}\ \underline{n_2}\ldots\underline{n_{k}}\ :\ k\in\NN^*,\ n_1=1\ \textnormal{and}\ \exists j<k, n_{j+1} \not= 2n_{j} \},\text{ where }\underline{n}=a^nb.\]
This language is recognized by a weakly-ambiguous Parikh automaton, which guesses the correct $j$, and then verifies that $n_{j+1} \not= 2n_{j}$. Let $\overline{\mathcal D}=ab(a^*b)^*\setminus\mathcal D$, and suppose by contradiction that $\mathcal D$ can be recognized by a weakly-unambiguous PA. Then its generating series should be holonomic by Proposition~\ref{prop:unambPAholonomic}. Since the generating series of $\overline{\mathcal D}$ is
$\overline{D}(x_a,x_b)=\frac{x_ax_b}{1-\frac{x_b}{1-x_a}} - D(x_a,x_b)$, it should be holonomic too. Looking closely at the form of the words of $\overline{\cal D}$, we get that its generating series is $\sum_{k\geq 1}x_a^{2^k-1}x_b^k$. It is not holonomic as $x\overline{D}(x,1)+x$ has infinitely many singularities, see~\cite[p.~296--297]{flajolet87}.

	Our second example is Shamir's language $\mathcal{S}=\{a^nbv_1a^nv_2: n\geq 1, v_1, v_2\in\{a,b\}^*\}$.
	One can easily design a PA recognizing $\cal S$, where one coordinate stands for the length of the first run of $a$'s and the other one for the second run of $a$'s, the automaton guessing when the second run starts.
	Flajolet proved that $\cal S$ is inherently ambiguous as a context-free language, since its generating series $S(z)=\frac{z(1-z)}{1-2z}\sum_{n\geqslant 1}\frac{z^{2n}}{1-2z+z^{n+1}}$ has an infinite number of singularities~\cite[p.~296--297]{flajolet87}. This also yields its inherent weak-ambiguity as a PA language.

\subsection{Limit of the method: an example using pumping techniques}
	As already mentioned, the analytic method presented is not always sufficient to prove inherent ambiguity. In this section, we develop an example where it does not apply. We consider the
	following language \leven, which is accepted both by a deterministic pushdown automaton and a non-deterministic PA (where $\underline{n}=a^nb$ as in Section~\ref{sec:2 examples}): \[\leven=\left\{\underline{n_1}\ \underline{n_2}\ldots\underline{n_{2k}}\ :\ k\in\NN,\ \forall i\leq 2k, n_i>0, \textnormal{ and }   \exists j\leq k, n_{2j} = n_{2j-1} \right\}.\] 
In other words, \leven is the language of sequences of encoded numbers having two consecutive equal values,  the first one being at an odd position. This language is accepted by a non-deterministic PA but is also deterministic context-free.\reviewer{See in Appendix, Sec.~\ref{sec:leven} for more details.}  This means that its generating series is algebraic and hence holonomic. This puts it out of the reach of our analytic method. 

In this section we establish the following result:	
	
\begin{theorem}\label{pro:leven ambiguous}
The language \leven is inherently weakly-ambiguous as a PA language.
\end{theorem}

The remainder of this section is devoted to sketch the proof of this proposition. By contradiction, we 
suppose that $\leven$ is recognized by a weakly-unambiguous PA $\AA$.

An \intro{$a$-piece} $\omega$ of $\mathcal A$ is a non-empty simple path of  $a$-edges in $\mathcal A$, starting and ending at the same state: the states of the path are pairwise distinct, except for its extremities. The \intro{origin} of $w$ is its starting (and ending) state.
Let $\Pi(\mathcal{A})$ be the (finite) set of $a$-pieces in $\mathcal A$.

We see a run in $\mathcal A$ as a sequence of transitions forming a path in $\mathcal A$. An \intro{$a$-subpath} of a run $R$ in $\AA$ is a maximal consecutive subsequence of $R$ whose transitions are all labeled by $a$'s, that is, it cannot be extended further to the left nor to the right in $R$ using $a$'s.  

Let $R$ be an accepting run in $\AA$. One can show that every $a$-subpath $S$ of $R$ can be decomposed as $S=w_1\sigma_1^{s_1}w_2\sigma_2^{s_2}\cdots w_f\sigma_f^{s_f}w_{f+1}$, where the $\sigma_i$'s are $a$-pieces of  $\Pi(\mathcal{A})$, the $s_i$'s are positive integers, and the $w_i$'s are paths not using twice the same state. Moreover, this decomposition is unique if we add the condition that if $w_i=\varepsilon$, then $\sigma_i\neq\sigma_{i-1}$ and the only state in common in $w_i$ and $\sigma_i$ is the origin of $\sigma_i$. This is done by repeatedly following the path until a state $q$ is met twice, factorizing this segment of the form $w\sigma$, where $\sigma$ is an $a$-piece of origin $q$. We call this decomposition the \intro{canonical form} of $S$, and  the \intro{signature} of $S$ is the tuple $(w_1,\sigma_1,w_2,\ldots,w_f,\sigma_f,w_{f+1})$, i.e., we dropped the $s_i$'s of the canonical form.

From the weakly-unambiguity\reviewer{See Sec.~\ref{sec:leven} in the Appendix for the proof.} of $\AA$ we can prove that there are at most $c$ distinct possible signatures, where $c$ only depends on $\AA$, and that $f$ is always at most $|Q_\AA|$.

Ramsey's Theorem~\cite{Ramsey} guarantees that there exists an integer $r$ such that any complete undirected graph with at least $r$ vertices, whose edges are colored using $c^2$ different colors, admits a monochromatic triangle. We fix two positive integers $n$ and $k$ sufficiently large, which will be chosen later on, depending on $\AA$ only. For $\ell\in\{1,\ldots,r\}$, let $w_\ell$ be the word $w_\ell = \underline{n_1}\,\underline{n_2}\ldots\underline{n_{2r}}$, where $n_i=n$ for odd $i$ and $n_{2i}=n+k$ if $i\neq\ell$ and $n_{2\ell}=n$. Each $w_\ell$ is in \leven, with a match at position $2\ell$ only. By weak-unambiguity, each $w_\ell$ has a unique accepting run $\mathcal{R}_\ell$ in $\AA$, each such run having $2r$ $a$-subpaths by construction. For $i\neq j$ in $\{1,\ldots,r\}$, let $\lambda_{ij}$ be the signature of the $2j$-th $a$-subpath of $\mathcal R_i$, which is a path of length $n+k$ by definition. The complete undirected graph  of vertex set  $\{1,\ldots,r\}$ where each edge $ij$, with $i<j$, is colored by the pair $(\lambda_{ij},\lambda_{ji})$ admits a monochromatic triangle of vertices $\alpha < \beta <\gamma$. In particular, $\lambda_{\alpha\beta}=\lambda_{\alpha\gamma}$ and $\lambda_{\gamma\alpha}=\lambda_{\gamma\beta}$. %

We choose $k=\lcm(\{|\sigma|: \sigma\in\Pi(\mathcal A)\})$, and $n$ sufficiently large so that any $a$-subpath of an accepting run contains an $a$-piece $\sigma$ repeated at least $k+1$ times. This is possible as the $w_i$'s have bounded length, and there are at most $|Q_\AA|+1$ of them. Hence, the $2\gamma$-th $a$-subpath of $w_\alpha$ contains a $a$-piece $\sigma$ that is repeated more than $s$ times, where $s=k/|\sigma|$. As $\lambda_{\alpha\beta}=\lambda_{\alpha\gamma}$, the piece $\sigma$ is also in the $2\beta$-th $a$-subpath of $w_\alpha$. If we alter the accepting path $\mathcal R_\alpha$ into $\mathcal R'_\alpha$ by looping $s$  more times in $\sigma$ in the $a$-subpath at position $2\beta$ and $s$  less times in $\sigma$ at position $2\gamma$, we obtain a run for the word $w=\cdots\underline{n}\underset{2\alpha}{\underline{n}}\cdots\underline{n}\,\underset{2\beta}{\underline{n+2k}}\cdots\underline{n}\underset{2\gamma}{\underline{n}}\cdots$. This run is accepting as the PA computes the same vector as for $w_\alpha$, by commutativity of vectors addition. And the signatures remains unchanged, as there are sufficiently many repetitions of $\sigma$ at position $2\gamma$ in $w_\alpha$. Similarly, as $\lambda_{\gamma\alpha}=\lambda_{\gamma\beta}$, we can alter the accepting path $\mathcal R_\gamma$ into an accepting path $\mathcal R'_\gamma$ of same signatures as $\mathcal R_\gamma$ for the same word $w$, by removing $s'=k/|\sigma'|$ iterations of an $a$-piece $\sigma'$ at position $2\alpha$ and adding them at position $2\beta$. 

We have built two paths $\mathcal R'_\alpha$ and $\mathcal R'_\gamma$ that both accept the same word $w$. Therefore, as $\AA$ is weakly-unambiguous, they are equal. As the signatures have not changed, this implies that the signature at position $2\alpha$ in $\mathcal R'_\alpha$ is $\lambda_{\gamma\alpha}$,
which is equal to $\lambda_{\gamma\beta}$ (monochromaticity), 
which is equal to  $\lambda_{\alpha\beta}$ ($\mathcal R'_\alpha$ and $\mathcal R'_\gamma$ have same signatures). This is a contradiction as we could remove one $a$-piece at position $2\alpha$ in $w_\alpha$ and add it at position $2\beta$, while
computing the same vector with the same starting and ending states: but this word is not in $\leven$.

\begin{remark}
The proof relies on manipulations of paths in the automaton, and we only use the commutativity of the  addition for the vector part. Thus, it still holds if we consider  automata where we use a recursively enumerable set instead of a semilinear set for acceptance.
\end{remark}
\vspace{-4mm}
\section{Algorithmic consequence of holonomicity}
\label{sec:complexity}

Generating series of languages have already been used to obtain efficient
algorithms on unambiguous models of automata. For instance, they were used by
Stearns and Hunt as a basic tool to obtain bounds on the length of a word
witnessing the non-inclusion between two unambiguous word
automata~\cite{Stearns85}. More precisely, the proof in~\cite{Stearns85}
relies on the recurrence equation satisfied by the coefficients of the
generating series (which is guaranteed to exist by holonomicity in one
variable). In the rational case, this recurrence relation can be derived from
the automaton and does not require advanced results on holonomic series. In
this section, our aim is to obtain a similar bound for the inclusion problem
for weakly-unambiguous Parikh automata. {The inclusion problem for RCM (and hence for weakly-unambiguous PA) is shown to be 
decidable in~\cite{Castiglione17} but no complexity bound is provided.} Note that this problem is known to be
undecidable for non-deterministic PA~\cite{Ibarra78}. We follow the same approach as
for the rational case \cite{Stearns85} and for RCM \cite{Castiglione17}. In
stark contrast with the rational case, it is necessary to closely inspect
holonomic closure properties in order to give concrete bounds.

Fix $\AA$ and $\BB$ two weakly-unambiguous PA. We can construct a weakly-unambiguous PA $\CC$ accepting $L(\AA)\cap L(\BB)$. We rely on the key fact that the series $D(x)=A(x)-C(x)$ counts the number of words of length $n$ in $L(\AA) \setminus L(\BB)$. In particular, $L(\AA) \subseteq L(\BB)$ if and only if $D(x)=0$.

As $D(x)$ is the difference of two holonomic series, it is holonomic. As equality between holonomic series is decidable, \cite{Castiglione17} concludes that
the problem is decidable. But without further analysis, no complexity upper-bound can be derived. The coefficients of $D(x)=\sum_{n \geq 0} d_n x^n$ satisfy a recurrence equation of the form $p_0(n) d_n = \sum_{k=1}^{r} p_k(n) d_{n-k}$ for $n \geq r$ with $p_0(x)\neq 0$.
This equation fully determines $d_n$ in terms of its $r$ previous values $d_{n-1},\ldots,d_{n-r}$, 
\textbf{provided that} $p_0(n)\not=0$. In particular, if the $r$ previous values are all equal to $0$, then $d_n=0$.
Consequently, if $d_n$ is equal to $0$ for all $n \leq r+R$ where $R$ denotes the largest positive root of $p_0$ (which is bounded from above by its $\infty$-norm, as a polynomial on $\mathbb{Z}$) then\footnote{The proof of Theorem~7 in \cite{Castiglione17} wrongly
suggests that we can take the order of the differential equation for~$D$ as a
bound on the length of a witness for non-inclusion. In general, this is not
the case. For instance consider $D(x)=x^{1000}$ which satisfies the
first-order differential equation $1000 D(x) - x \partial_x D(x)=0$.
It is clear that the coefficients~$D_0=0$ and $D_1=0$ are not enough to decide
that $D$ is not zero.
}
 $D(x)=0$.

Taking $W := r+R$, we have that if $L(\AA) \not\subseteq L(\BB)$ then
there exists a word witnessing this non-inclusion of length at most $W$. We now aim at computing an upper-bound on $W$
on the size of the inputs $\AA$ and $\BB$. 

For this, we first bound the order of the linear
recurrence satisfied by $A(x)$ and $C(x)$, as well as the degrees and norm of
the polynomials involved. This is stated in
Proposition~\ref{prop:bound-rec-pa}, whose proof follows the one of
Proposition~\ref{prop:unambPAholonomic}, while establishing such bounds for
the multivariate Hadamard product and the specialization to 1.

\begin{restatable}{proposition}{propboundrecpa} \label{prop:bound-rec-pa} The generating series $A(x)$ of
a weakly-unambiguous PA $\AA$ of dimension $d \geq 1$ satisfies a non-trivial 
linear differential equation $q_s(x)\partial_x^s A(x)+\cdots+ q_0(x)A(x)=0$, 
with  $s \leq ((d+1)|A|\,\|A\|_\infty)^{O(d)}$ and for all $i \in [0,s]$,
$\deg(q_i)\leq ((d+1)|A|\,\|A\|_\infty)^{O(d)}$ and $\log \|q_i\|_\infty\leq ((d+1)|A|\,\|A\|_\infty)^{O(d^2)}$  using the notations of Section~\ref{sec:PA GS}. \end{restatable}

Finally, we transfer these bounds to the series $D(x)$ of $L(\AA)\setminus L(\BB)$ using the analysis of
\cite{Kauers2014} for the sum of holonomic series in one variable.

\begin{restatable}{theorem}{theoreminclusion}\label{thm:inclusionbound}
 Given two weakly-unambiguous PA $\AA$ and $\BB$ of respective dimensions $d_\AA$ and $d_\BB$, if $L(\AA)$ is not included in $L(\BB)$ then there exists a word in $L(\AA)\setminus L(\BB)$ of length at most $2^{2^{O(d^2\log(dM))}}$ where $d = d_\AA+d_\BB$ and $M= |\AA|\, |\BB|\, \|\AA\|_\infty\, \|\BB\|_\infty$.
\end{restatable}
Using dynamic programming to compute the number of accepted words, we can decide inclusion.

\begin{restatable}{corollary}{corcomplexite} \label{cor:complexite}
	 Given two weakly-unambiguous PA $\AA$ and $\BB$ of dimensions $d_\AA$ and $d_\BB$, we can decide  if $L(\AA) \subseteq L(\BB)$ in time $2^{2^{O(d^2\log(dM))}}$ where $d = d_\AA+d_\BB$ and $M= |\AA|\, |\BB|\, \|\AA\|_\infty\, \|\BB\|_\infty$.
\end{restatable}

\begin{remark}\label{rmk:Massazza} In \cite{Castiglione17}, the author proposes a different
construction to prove the holonomicity of the generating series of languages
in RCM. This proof uses the closure of holonomic series under Hadamard product
and algebraic substitution $x_1 = x_2 = \cdots = x_n=x$. It is natural to
wonder if this approach would lead to better bounds in
Proposition~\ref{prop:bound-rec-pa} (using the equivalence between
weakly-unambiguous PA and RCM). It turns out that the operation $x_1 = x_2 =
\cdots = x_n=x$ is more complicated than it seems at a first glance. Indeed,
to our knowledge, no proof of the closure under algebraic substitution
explains what happens if, during the substitution process, the equations
become trivial. This issue can be overcome by doing the substitution step by
step: $x_2 = x_1$, then $x_3=x_1$, etc. However, this naive approach would
produce worse bounds. 
\end{remark}

\section{Perspectives}

The bounds obtained in Section~\ref{sec:complexity} are derived directly from constructions given in the proofs of the closure properties. In particular, we did not use any information on the special form of our series. The bounds  are certainly perfectible using more advanced tools from computer algebra. Also it seems that the complexity of the closure under the algebraic substitution deserves more investigation, as discussed in Remark~\ref{rmk:Massazza}.

A more ambitious perspective is to find larger classes of automata whose generating series are holonomic. This would certainly require new ideas, as for instance any holonomic power series with coefficients in $\{0,1\}$ is known to be the characteristic series of some semilinear \cite{bell}.

\bibliography{main.bib}

@preamble{"\def\cprime{$'$} "}

@inproceedings{Abney1999,
	Author = {Abney, Steven P. and McAllester, David A. and Pereira, Fernando},
	Booktitle = {27th Annual Meeting of the Association for Computational Linguistics, {ACL 1999}},
	Doi = {10.3115/1034678.1034759},
	Pages = {542--549},
	Publisher = {{ACL}},
	Title = {Relating Probabilistic Grammars and Automata},
	Year = {1999}}

@incollection{BCLSS07,
	Author = {Bostan, Alin and Chyzak, Fr\'{e}d\'{e}ric and Lecerf, Gr\'{e}goire and Salvy, Bruno and Schost, \'{E}ric},
	Booktitle = {Symbolic and Algebraic Computation, International Symposium, {ISSAC} 2007},
	Doi = {10.1145/1277548.1277553},
	Pages = {25--32},
	Publisher = {{ACM}},
	Title = {Differential equations for algebraic functions},
	Year = {2007}}

@inproceedings{10.1007/Bonsangue2012,
	Author = {Marcello M. Bonsangue and Jan J. M. M. Rutten and Joost Winter},
	Booktitle = {Coalgebraic Methods in Computer Science - 11th International Workshop, {CMCS} 2012},
	Doi = {10.1007/978-3-642-32784-1\_2},
	Pages = {20--39},
	Publisher = {Springer},
	Series = {Lecture Notes in Computer Science},
	Title = {Defining Context-Free Power Series Coalgebraically},
	Volume = {7399},
	Year = {2012}}

@inproceedings{massazza18,
	Author = {Paolo Massazza},
	Booktitle = {Proceedings of the 19th Italian Conference on Theoretical Computer Science, ({ICTCS 2018})},
	Pages = {191--202},
	Publisher = {CEUR-WS.org},
	Series = {{CEUR} Workshop Proceedings},
	Title = {On the Generating Functions of Languages Accepted by Deterministic One-reversal Counter Machines},
	Volume = {2243},
	Year = {2018}}

@article{bostan2017,
	Author = {Bostan, Alin and Chyzak, Fr{\'{e}}d{\'{e}}ric and van Hoeij, Mark and Kauers, Manuel and Pech, Lucien},
	Doi = {10.1016/j.ejc.2016.10.010},
	Journal = {Eur. J. Comb.},
	Pages = {242--275},
	Publisher = {Elsevier},
	Title = {Hypergeometric expressions for generating functions of walks with small steps in the quarter plane},
	Volume = {61},
	Year = {2017}}

@article{Brown,
	Author = {Brown, William S.},
	Doi = {10.1145/321662.321664},
	Journal = {J. Assoc. Comput. Mach.},
	Pages = {478--504},
	Title = {On {E}uclid's algorithm and the computation of polynomial greatest common divisors},
	Volume = {18},
	Year = {1971}}

@article{Emiris,
	Author = {Emiris, Ioannis Z. and Pan, Victor Y.},
	Doi = {10.1016/j.jco.2004.03.003},
	Issn = {0885-064X},
	Journal = {J. Complexity},
	Number = {1},
	Pages = {43--71},
	Title = {Improved algorithms for computing determinants and resultants},
	Volume = {21},
	Year = {2005}}

@article{Goldstein,
	Author = {Goldstein, A. J. and Graham, R. L.},
	Doi = {10.1137/1016065},
	Fjournal = {SIAM Reivew},
	Journal = {SIAM Rev.},
	Number = {3},
	Pages = {394--395},
	Title = {A {H}adamard-type bound on the coefficients of a determinant of polynomials},
	Volume = {16},
	Year = {1974}}

@article{Ramsey,
	Author = {Ramsey, Frank P.},
	Doi = {10.1112/plms/s2-30.1.264},
	Journal = {Proc. London Math. Soc. (2)},
	Number = {4},
	Pages = {264--286},
	Title = {On a {P}roblem of {F}ormal {L}ogic},
	Volume = {30},
	Year = {1929}}

@inproceedings{Kauers2014,
	Author = {Kauers, Manuel},
	Booktitle = {{P}roceedings of the 39th {I}nternational {S}ymposium on {S}ymbolic and {A}lgebraic {C}omputation, (I{SSAC} 2014)},
	Doi = {10.1145/2608628.2608634},
	Pages = {288--295},
	Publisher = {ACM, New York},
	Title = {Bounds for {D}-finite closure properties},
	Year = {2014}}

@article{Bernstein,
	Author = {Bernstein, Joseph N.},
	Doi = {10.1007/BF01077645},
	Journal = {Funct. Anal. Appl.},
	Number = {4},
	Pages = {273--28},
	Title = {Analytic continuation of generalized functions with respect to a parameter},
	Volume = {6},
	Year = {1972}}

@article{Kashiwara,
	Author = {Kashiwara, Masaki},
	Doi = {10.1007/BF01403082},
	Journal = {Invent. Math.},
	Number = {2},
	Pages = {121--135},
	Title = {On the holonomic systems of linear differential equations. {II}},
	Volume = {49},
	Year = {1978}}

@article{Takayama,
	Author = {Takayama, Nobuki},
	Doi = {10.1016/0747-7171(92)90039-7},
	Journal = {J. Symbolic Comput.},
	Number = {2-3},
	Pages = {265--282},
	Title = {An approach to the zero recognition problem by {B}uchberger algorithm},
	Volume = {14},
	Year = {1992}}

@incollection{Bousquet-Melou06,
	Author = {Bousquet{-}M{\'{e}}lou, Mireille},
	Booktitle = {International {C}ongress of {M}athematicians ({ICM} 2006)},
	Doi = {10.4171/022-3/40},
	Pages = {789--826},
	Publisher = {Eur. Math. Soc., Z\"{u}rich},
	Title = {Rational and algebraic series in combinatorial enumeration},
	Volume = {3},
	Year = {2006}}

@article{Stearns85,
	Author = {Richard Edwin Stearns and Harry B. Hunt III},
	Doi = {10.1137/0214044},
	Journal = {SIAM J. Comput.},
	Number = {3},
	Pages = {598--611},
	Title = {On the equivalence and containment problems for unambiguous regular expressions, regular grammars and finite automata},
	Volume = {14},
	Year = {1985}}

@book{Stanley,
	Author = {Stanley, Richard P.},
	Doi = {10.1017/CBO9780511609589},
	Isbn = {0-521-56069-1; 0-521-78987-7},
	Pages = {581},
	Publisher = {Cambridge University Press, Cambridge},
	Series = {Cambridge Studies in Advanced Mathematics},
	Title = {Enumerative combinatorics. {V}ol. 2},
	Volume = {62},
	Year = {1999}}

@article{Ito69,
	Author = {Ryuichi Ito},
	Doi = {10.1016/S0022-0000(69)80014-0},
	Journal = {J. Comput. Syst. Sci.},
	Number = {2},
	Pages = {221--231},
	Title = {Every Semilinear Set is a Finite Union of Disjoint Linear Sets},
	Volume = {3},
	Year = {1969}}

@article{comtet,
	Author = {Comtet, Louis},
	Journal = {Enseignement Math. (2)},
	Pages = {267--270},
	Title = {Calcul pratique des coefficients de {T}aylor d'une fonction alg{\'e}brique},
	Volume = {10},
	Year = {1964}}

@article{Greibach1968,
	Author = {Greibach, Sheila},
	Doi = {10.1007/BF01691341},
	Journal = {Mathematical Systems Theory},
	Number = {1},
	Pages = {1--6},
	Title = {A note on undecidable properties of formal languages},
	Volume = {2},
	Year = {1968}}

@article{bell,
	Author = {Bell, Jason P. and Chen, Shaoshi},
	Doi = {10.1016/j.jcta.2017.05.002},
	Journal = {J. Comb. Theory, Ser. {A}},
	Pages = {241 - 253},
	Title = {Power series with coefficients from a finite set},
	Volume = {151},
	Year = {2017}}

@article{flajolet2005,
	Author = {Philippe Flajolet and Stefan Gerhold and Bruno Salvy},
	Doi = {10.37236/1894},
	Journal = {Electr. J. Comb.},
	Number = {2},
	Title = {On the Non-Holonomic Character of Logarithms, Powers, and the nth Prime Function},
	Volume = {11},
	Year = {2005}}

@phdthesis{cheung,
	Author = {Cheung, Siu-nang Bruce},
	School = {University of Hong Kong, Pokfulam, Hong Kong SAR},
	Title = {A theory of automatic language acquisition},
	Url = {http://hdl.handle.net/10722/34911},
	Year = {1994}}

@article{ponty2012,
	Author = {Ponty, Yann},
	Hal_Id = {hal-00693600},
	Hal_Version = {v1},
	Journal = {CoRR},
	Title = {{Rule-weighted and terminal-weighted context-free grammars have identical expressivity}},
	Type = {Research Report},
	Volume = {abs/1205.0627},
	Year = {2012}}

@article{STANAT1972,
	Author = {Donald F. Stanat},
	Doi = {10.1016/S0022-0000(72)80003-5},
	Journal = {J. Comput. Syst. Sci.},
	Number = {3},
	Pages = {217--232},
	Title = {A homomorphism theorem for weighted context-free grammars},
	Volume = {6},
	Year = {1972}}

@article{Pres29,
	Author = {Presburger, Moj{\.z}esz},
	Journal = {C.R. 1er congres des Mathematiciens des pays slaves},
	Pages = {92-101},
	Title = {{U}ber die {V}olstandigkeit eines gewissen {S}ystems der {A}rithmetik ganzer zahlen, in welchem die {A}ddition als einzige {O}peration hervortritt},
	Year = {1929}}

@article{Filliot19,
	Archiveprefix = {arXiv},
	Author = {Emmanuel Filiot and Shibashis Guha and Nicolas Mazzocchi},
	Eprint = {1907.09362},
	Journal = {CoRR},
	Title = {{Two-way Parikh Automata}},
	Volume = {abs/1907.09362},
	Year = {2019}}

@inproceedings{Chistikov16,
	Author = {Dmitry Chistikov and Christoph Haase},
	Booktitle = {43rd International Colloquium on Automata, Languages, and Programming, {ICALP} 2016},
	Doi = {10.4230/LIPIcs.ICALP.2016.128},
	Pages = {128:1--128:13},
	Publisher = {Schloss Dagstuhl - Leibniz-Zentrum f{\"u}r Informatik},
	Series = {LIPIcs},
	Title = {The Taming of the Semi-Linear Set},
	Volume = {55},
	Year = {2016}}

@article{EILENBERG1969,
	Author = {Eilenberg, Samuel and Sch{\"u}tzenberger, Marcel-Paul},
	Doi = {10.1016/0021-8693(69)90070-2},
	Issn = {0021-8693},
	Journal = {J. Algebra},
	Number = {2},
	Pages = {173 - 191},
	Title = {Rational sets in commutative monoids},
	Volume = {13},
	Year = {1969}}

@article{Parikh1966,
	Author = {Parikh, Rohit J.},
	Doi = {10.1145/321356.321364},
	Journal = {J. ACM},
	Number = {4},
	Numpages = {12},
	Pages = {570--581},
	Publisher = {ACM},
	Title = {On Context-Free Languages},
	Volume = {13},
	Year = {1966}}

@article{LIPSHITZ1989,
	Author = {Lipshitz, Leonard},
	Doi = {10.1016/0021-8693(89)90222-6},
	Journal = {J. Algebra},
	Number = {2},
	Pages = {353 - 373},
	Title = {{D}-finite power series},
	Volume = {122},
	Year = {1989}}

@article{LIPSHITZ1988,
	Author = {Lipshitz, Leonard},
	Doi = {10.1016/0021-8693(88)90166-4},
	Journal = {J. Algebra},
	Number = {2},
	Pages = {373 - 378},
	Title = {The diagonal of a {D}-finite power series is {D}-finite},
	Volume = {113},
	Year = {1988}}

@article{STANLEY1980,
	Author = {Richard P. Stanley},
	Doi = {10.1016/S0195-6698(80)80051-5},
	Journal = {Eur. J. Comb.},
	Number = {2},
	Pages = {175 --188},
	Title = {Differentiably Finite Power Series},
	Volume = {1},
	Year = {1980}}

@article{Cadilhac12a,
	Author = {Micha{\"{e}}l Cadilhac and Alain Finkel and Pierre McKenzie},
	Doi = {10.1051/ita/2012013},
	Journal = {{RAIRO} - Theor. Inf. and Applic.},
	Number = {4},
	Pages = {511--545},
	Title = {Affine {P}arikh automata},
	Volume = {46},
	Year = {2012}}

@article{Cadilhac13,
	Author = {Micha{\"{e}}l Cadilhac and Alain Finkel and Pierre McKenzie},
	Doi = {10.1142/S0129054113400339},
	Journal = {Int. J. Found. Comput. Sci.},
	Number = {7},
	Pages = {1099--1116},
	Title = {Unambiguous constrained automata},
	Volume = {24},
	Year = {2013}}

@article{Castiglione17,
	Author = {Giusi Castiglione and Paolo Massazza},
	Doi = {10.1016/j.tcs.2016.07.022},
	Journal = {Theor. Comput. Sci.},
	Pages = {74--84},
	Title = {On a class of languages with holonomic generating functions},
	Volume = {658},
	Year = {2017}}

@article{Ibarra78,
	Author = {Oscar H. Ibarra},
	Doi = {10.1145/322047.322058},
	Journal = {J. {ACM}},
	Number = {1},
	Pages = {116--133},
	Title = {Reversal-Bounded Multicounter Machines and Their Decision Problems},
	Volume = {25},
	Year = {1978}}

@inproceedings{Klaedtke03,
	Author = {Felix Klaedtke and Harald Rue{\ss}},
	Booktitle = {Automata, Languages and Programming, 30th International Colloquium, {ICALP} 2003},
	Doi = {10.1007/3-540-45061-0_54},
	Pages = {681--696},
	Publisher = {Springer},
	Series = {Lecture Notes in Computer Science},
	Title = {Monadic Second-Order Logics with Cardinalities},
	Volume = {2719},
	Year = {2003}}

@techreport{Klaedtke02,
	Author = {Felix Klaedtke and Harald Rue{\ss}},
	Institution = {Freiburg University},
	Number = {177},
	Title = {Parikh automata and {M}onadic {S}econd-{O}rder logics with linear cardinality constraints},
	Year = {2002}}

@article{Massazza93,
	Author = {Paolo Massazza},
	Doi = {10.1051/ita/1993270201491},
	Journal = {{ITA}},
	Number = {2},
	Pages = {149--161},
	Title = {Holonomic Functions and Their Relation to Linearly Constrained Languages},
	Volume = {27},
	Year = {1993}}

@inproceedings{Massazza17,
	Author = {Paolo Massazza},
	Booktitle = {Implementation and Application of Automata - 22nd International Conference, {CIAA} 2017},
	Doi = {10.1007/978-3-319-60134-2_15},
	Pages = {175--187},
	Publisher = {Springer},
	Series = {Lecture Notes in Computer Science},
	Title = {On the Conjecture $\mathcal{L}_{\mathrm{DFCM}} \subsetneq \mathrm{RCM}$},
	Volume = {10329},
	Year = {2017}}

@book{Flajsedg,
	Author = {Flajolet, Philippe and Sedgewick, Robert},
	Edition = {First},
	Isbn = {0521898064, 9780521898065},
	Publisher = {Cambridge University Press},
	Title = {Analytic Combinatorics},
	Year = {2009}}

@book{chomschut63,
	Author = {Noam Chomsky and Marcel-Paul Sch{\"u}tzenberger},
	Booktitle = {Computer Programming and Formal Systems},
	Doi = {10.1016/S0049-237X(08)72023-8},
	Issn = {0049-237X},
	Pages = {118 -- 161},
	Publisher = {Elsevier},
	Series = {Studies in Logic and the Foundations of Mathematics},
	Title = {The Algebraic Theory of Context-Free Languages},
	Volume = {35},
	Year = {1963}}

@article{flajolet87,
	Author = {Philippe Flajolet},
	Doi = {10.1016/0304-3975(87)90011-9},
	Issn = {0304-3975},
	Journal = {Theor. Comput. Sci.},
	Number = {2},
	Pages = {283 - 309},
	Title = {Analytic models and ambiguity of context-free languages},
	Volume = {49},
	Year = {1987}}

@article{cramersrule,
	Author = {Akhtyamov, Azamat and Amram, Meirav and Dagan, Miriam and Mouftahkov, Artour},
	Journal = {The Teach. Math.},
	Number = {1},
	Pages = {13--19},
	Title = {Cramer's rule for nonsingular $m \times n$ matrices},
	Volume = {20},
	Year = {2017}}

\newpage
\appendix
\section{Proof omitted in Section~\ref{sec:primer}}

\begin{proposition}\label{prop:specialisation_holonome}
			Let $A(x_1,\ldots, x_k, y_1, \ldots, y_\ell)$ be an holonomic series such that for every index $(i_1, \ldots, i_k)$, $[x_1^{i_1}\ldots x_k^{i_k}]A$ is a polynomial in $y_1, \ldots, y_\ell$. Then the series defined by $B(x_1, \ldots, x_k)=A(x_1, \ldots, x_k, 1, \ldots, 1)$ is also holonomic.
		\end{proposition}

		\begin{proof}
					We write $A(x_1,\ldots, x_k, y_1, \ldots, y_\ell)=\sum_{i_1, \ldots, i_k}P_{i_1, \ldots, i_k}(y_1,\ldots, y_\ell)x_1^{i_1}\ldots x_k^{i_k}$, where each $P_{i_1, \ldots, i_k}$ is a polynomial in $y_1$, \ldots, $y_\ell$. 
					
		First observe that $B(x_1, \ldots, x_k)$ is well defined, as $[x_1^{i_1}\ldots x_k^{i_k}]B$ is the evaluation at $(y_1,\ldots,y_\ell)=(1,\ldots,1)$ of  $P_{i_1, \ldots, i_k}$.

	We prove the result by setting each $y_i$ to $1$, one at a time. So we focus on the operation $y_\ell=1$.
			Let $C(x_1,\ldots, x_k, y_1, \ldots, y_{\ell-1})=A(x_1,\ldots, x_k, y_1, \ldots, y_{\ell-1},1)$.
			
			By a direct induction, for every non-negative integer $k$ we have that for every $z\in\{x_1,\ldots, x_n, y_1, \ldots, y_{d-1}\}$: \[(\partial_z^k A)(x_1,\ldots, x_k, y_1, \ldots, y_{\ell-1},1)=\partial_z^k C(x_1,\ldots, x_k, y_1, \ldots, y_{\ell-1}).\]

			Then setting $y_\ell=1$ in the system of differential equations of $A$ (after removing the differential equation in $y_\ell$)  yields a system of equations for $C$. To ensure that none of these equations is trivial when setting  $y_\ell=1$, we first divide each of them by $(y_\ell-1)$, as many times as possible (that is, by the largest power of $(y_\ell-1)$ that divides all the polynomial coefficients of the equation).
		\end{proof}
\section{Proof omitted in Section~\ref{sec:wupa}}
\label{app:wupa}

\subsection{Proofs and details for Remark~\ref{rmk:comp}}
\label{app:rmk-comp}

Recall 
 the language $\mathcal{L} = \{c^nw \;:\; w = x_1x_2\cdots x_m \in\{a,b\}^*\ \wedge\ n>0 \wedge\ |x_1x_2\ldots x_n|_{a} < |x_1x_2\ldots x_n|_{b}\} $ over the alphabet $\{a,b,c\}$. 
 
 It is accepted by the weakly-unambiguous automaton depicted below with the semilinear $\{ (n_1,n_2,n_3) \;:\; n_1=n_2+n_3 \;\text{and}\; n_2 < n_3 \}$. This automaton  stores the number $n$ of $c$'s seen in its 1st dimension and upon seeing the first $a$ or $b$ starts storing the number of $a$'s in its 2nd dimension and the number of $b$'s in its 3rd dimension. At any point, it can non-deterministically choose to stop counting the $a$'s and $b$'s and move to an accepting state. The semilinear constraint ensures that there is only one correct guess for when to stop counting the $a$'s and $b$'s and hence that the automaton is weakly-unambiguous. Remark that if we erase the vectors, the resulting automaton is ambiguous and therefore this automaton is not unambiguous in the sense of Cadilhac et al.

 \begin{center}
\begin{tikzpicture}
\node (init) at (0,-.8) {};
\node[draw,circle] (q0) at (0,0) {$0$};
\node[draw,circle,accepting] (q1) at (2,0) {$1$};
\node[draw,circle,accepting] (q2) at (4,0) {$2$};

\draw[->,thick] (q0) edge node[below] {\small $a\big(\substack{0\\1\\0}\big), b\big(\substack{0\\0\\1}\big)$} (q1);
\draw[->,thick] (q1) edge node[below] {\small $a\big(\substack{0\\0\\0}\big), b\big(\substack{0\\0\\0}\big)$} (q2);

\draw[->,thick] (q0) edge[loop, above] node[above] {\small $c\big(\substack{1\\0\\0}\big)$} (q0);
\draw[->,thick] (q1) edge[loop, above] node[above] {\small $a\big(\substack{0\\1\\0}\big), b\big(\substack{0\\0\\1}\big)$} (q1);
\draw[->,thick] (q2) edge[loop, above] node[above] {\small $a\big(\substack{0\\0\\0}\big), b\big(\substack{0\\0\\0}\big)$} (q1);

\draw[->, thick] (init) -- (q0);
\end{tikzpicture}
 \end{center} 
 
\begin{proposition}
The language $\mathcal{L} = \{c^nw \;:\; w = x_1x_2\cdots x_m \in\{a,b\}^*\ \wedge\ n>0 \wedge\ |x_1x_2\ldots x_n|_{a} < |x_1x_2\ldots x_n|_{b}\} $ over the alphabet $\{a,b,c\}$ is not accepted by any unambiguous PA in the sense of \cite{Cadilhac13}.
\end{proposition}

\begin{proof}
Let's suppose that $\mathcal{L}$ is recognized by an unambiguous PA in the sense of  \cite{Cadilhac13}. Then, by stability of unambiguous PA (see \cite[Proposition 7]{Cadilhac13}) under intersection and left quotient by any language recognized by a PA (in particular any regular language), $(\{c^*\}^{-1}L)\cap \{a,b\}^*$ should be unambiguous. But it is not, by Proposition 11 of \cite{Cadilhac13}. 
\end{proof}

\subsection{Proof of closure under intersection for weakly-unambiguous PA}
\label{app:inter-pa}

\begin{proposition}\label{prop:intersectionPA}
Given two weakly-unambiguous PA $\mathcal{A}$ and $\mathcal{B}$, the language $L(\mathcal{A})\cap L(\mathcal{B})$ is accepted by a weakly-unambiguous PA $\mathcal{C}$. Moreover $\CC$ can be constructed so that $|\CC| \leq |\AA| |\BB|$, $\|\CC\|_\infty = \max(\|\AA\|_\infty,\|\BB\|_\infty)$ and its dimension is the sum of the dimensions of $\AA$ and $\BB$.  
\end{proposition}

\begin{proof}
This is done by a classical construction of a product automaton. We create the automaton $\CC$, of dimension $d_\CC=d_\AA+d_\BB$, set of states $Q_\CC=Q_{\AA}\times Q_{\BB}$, such that for every pair of transitions $(t_1, t_2)\in \delta_\AA\times\delta_\BB$, labeled by the same letter $a\in\Sigma$, with $t_1=(q_1,(a,\vect{v_1}),q'_1)$ and $t_2=(q_2,(a,\vect{v_2}),q'_2)$, then $((q_1,q_2),(a,(\vect{v_1}, \vect{v_2})),(q'_1,q'_2))\in\delta_\CC$. We also pose $F_\CC=F_\AA\times F_\BB$, $q_{0,\CC}=(q_{0,\AA},q_{0,\BB})$.
Finally, we take the semilinear set $C_\CC$ to be 
\[
\{(x_1,\ldots,x_{d_\AA},x_{d_\AA+1},\ldots,x_{d_\CC}) \;:\; (x_1,\ldots,x_{d_\AA}) \in C_\AA \;\text{and}\; (x_{d_\AA+1},\ldots,x_{d_\CC}) \in C_\BB  \}
\]
The runs of $\CC$ are classically in bijection with the couples of runs in $\AA$ and $\BB$ labeled with the same word, so that $\CC$  is also weakly-unambiguous, and recognizes the intersection of the two languages.

If $C_\AA$ and $C_\BB$ are given by unambiguous presentations:
\[
\begin{array}{lcl}
C_\AA & =&  \biguplus_{i \in [1,k_\AA]} \vect{c_i} + P_i^* \\
C_\BB & = & \biguplus_{i \in [1,k_\BB]} \vect{d_i} + R_i^* \\
\end{array}
\]
An unambiguous presentations for the semilinear set $C$ can be constructed as follows:
\[
C = \biguplus_{i \in [1,k_\AA] \\ j \in [1,k_\BB]} \vect{c_i} \cdot \vect{d_j} + Q_{i,j}^*
\]
where for $i \in [1,k_\AA]$ and $j \in [1,k_\BB]$, \[
Q_{i,j} = \{ p \cdot \underbrace{(0,\ldots,0)}_{\text{$d_\BB$ times}} \;:\; p \in P_i\} \cup \{  \underbrace{(0,\ldots,0)}_{\text{$d_\AA$ times}} \cdot r \;:\; r \in R_j\}
\]
From the definition, we have $|\CC| = |Q_\CC| + |\delta_\CC| + k_\AA k_\BB + \sum_{i,j} |P_i| + |R_j|$. As $|Q_\CC|=|Q_\AA| |Q_\BB|$ and $|\delta_\CC| \leq |\delta_\AA| |\delta_\BB|$, $|\CC| \leq |\AA| |\BB|$.  We also have that $\|\CC\|_\infty = \max(\|\AA\|_\infty,\|\BB\|_\infty)$.

\end{proof}

\subsection{A counter-example construction for union announced in \cite{Castiglione17}}
\label{app:counter-example-massazza}

We give a counter-example for the proof of the closure of union of RCM in \cite{Castiglione17}.

We follow the construction for the union of $L=(R,C,\mu)$ and 
$L'=(R',C',\mu)$ with:
\begin{itemize}
    \item $R = \{a,A\}$ over $\Gamma=\{a,A\}$,
    \item $C = (\mathbb{N} \times \{0\}) \setminus (0,0)$,
    \item $\Sigma=\{c\}$,
    \item $\mu : \Gamma^* \mapsto \Sigma^*$ defined by $\mu(a)=\mu(A)=c$.
    \item $R' = \{b\}$ over $\Gamma'=\{b\}$,
    \item $C' = \mathbb{N} \setminus \{0\}$,
    \item $\mu' : (\Gamma')^* \mapsto \Sigma^*$ defined by $\mu(b)=c$.
\end{itemize}
As a result $L=\mu(R \cap [C]) = \{c\}$ and $L'=\mu'(R' \cap [C'])=\{c\}$. 

The language $M = \mu(R) \cap \mu(R')=\{c\}$. Hence $R_2=\emptyset$, $R_2'=\emptyset$, $R_1=R$ and $R_1'=R'$. The construction produces $M'=\{\tau_{ab},\tau_{Ab}\}$ and $\mu''(\tau_{ab})=\mu''(\tau_{Ab})=c$ and $C''=\mathbb{N}^2\setminus (0,0)$. The union is announced to be 
$\mu''(M' \cap [C''])$ but the morphism is not injective on $M' \cap [C'']$.

Notice that this only invalidates the proof given in \cite{Castiglione17}. We do not know yet if RCM is closed under union. 

\subsection{Omitted results and definitions on PA, generalized PA and PA with $\varepsilon$-transitions}

In this section, we introduce a model, we call generalized PA, in which transitions are labeled by semilinear sets instead of vectors and show that it accepts the same languages as standard $\PA$ in the weak-unambiguous case (see Section~\ref{sec:generalized-pa}). Using the model of generalized $\PA$, we show $\varepsilon$-transitions can be eliminated in weakly-unambiguous PA while preserving weak-unambiguity (see Section~\ref{sec:eps-pa}). We start with formal definitions omitted in the main text.

\textbf{Related work.}
This model is implicit in \cite{Klaedtke02,Klaedtke03}. We make it explicit here for two reasons. First, it allows for a cleaner presentation of the $\varepsilon$-removal procedure. Second, the proof of equivalence between generalized $\PA$ and $\PA$ is quite generic and will be reused for \PPA. The $\varepsilon$-removal procedure detailed in Section~\ref{sec:eps-pa} is essentially the one of \cite{Klaedtke03} but tailored to ensure the preservation of weak-unambiguity.

\subsubsection{Runs of PA}

A \intro{run of a Parikh automaton} $\AA=(\Sigma, Q, q_I, F, C, \Delta)$ of dimension $d$ is a sequence 
\[
\pi=p_0 \,(a_0,\vect{v_0})\, p_1 \cdots p_{n-1}\, (a_{n-1},\vect{v}_{n-1})\,p_n
\]
 with $n \geq 0$, $p_i \in Q$, $a_i \in \Sigma$ and $\vect{v}_i \in \NN^d$ such that for all $i \in [0,n-1]$, $(p_i,(a_i,\vect{v_i}),p_{i+1})$ is a transition of $\AA$. Such a run will be denoted by  $p_0 \era{a_0,\vect{v_0}} p_1 \cdots p_{n-1} \era{a_{n-1},\vect{v_{n-1}}} p_n$
 and is said to start from $p_0$ and to end in $p_n$ and labeled by $(a_0\cdots a_{n-1},\vect{v_0}+\cdots+\vect{v_{n-1}})$. To convey this information in a synthetic manner, we simply write $p \eRb{w,\vect{v}}{\pi} q$ where $w=
 a_0\cdots a_{n-1}$ and $\vect{v}=\vect{v_0}+\cdots+\vect{v_{n-1}}$.
  Remark that $\pi=q$ is a valid run starting and ending in $q$ and labeled by $(\varepsilon,\vect{0})$ where $\vect{0}$ denotes the constant vector of dimension $d$ equal to $0$.

For a run $\pi_1$ starting in $p$ and ending in $q$ and a run $\pi_2$ starting in $q$ and ending in $r$, we let $\pi_1 \cdot \pi_2$ denote the concatenation of the two runs. Formally if $\pi_1=p_0 \era{a_0,\vect{v_0}} p_1 \,\cdots\, p_{n-1} \era{a_{n-1},\vect{v_{n-1}}} p_n$ and $\pi_2=q_0 \era{b_0,\vect{w_0}} q_1\, \cdots\, q_{m-1} \era{b_{n-1},\vect{v_{n-1}}} q_n$, we define:
\[
\pi_1 \cdot \pi_2 = p_0 \era{a_0,\vect{v_0}} p_1 \cdots p_{n-1} \era{a_{n-1},\vect{v_{n-1}}} p_n  \era{b_0,\vect{w_0}} q_1 \cdots q_{m-1} \era{b_{n-1},\vect{w_{n-1}}} q_n.
\]
In particular, if $p \eRb{u_1,\vect{v_1}}{\pi_1} q$ and $q \eRb{u_2,\vect{v_2}}{\pi_2} r$ then $p \eRb{u_1\cdot u_2, \vect{v_1} + \vect{v_2}}{\pi_1 \cdot \pi_2} r$.

\subsubsection{Generalized $\PA$}
\label{sec:generalized-pa}

We consider a variant of Parikh automata, called \intro{generalized Parikh automata} (generalized PA for short), in which transitions are no longer labeled by a vector in $\NN^d$ but by a semilinear set in $\NN^d$. As a consequence, a run $\pi$ of a generalized PA with is labeled by a couple $(w, S)$, where $w$ is a word and $S$ is a semilinear equal to the sum of all the semilinear sets labelling the transitions of $\pi$. A run $\pi$ labeled by $(w,S)$  is accepting if it starts in the initial state and ends in a final state and $S \cap C \neq \emptyset$ where $C$ is the semilinear constraint of the generalized PA.

The model of generalized PA is implicitly used in work of \cite{Klaedtke02,Klaedtke03}.

\begin{definition}[generalized Parikh automaton]
	A generalized Parikh automaton $\AA$ is given by a tuple  
	$(\Sigma, Q, q_I, F, C, \Delta)$  where $\Sigma$ is the input alphabet, $Q$ is the finite set of states, $q_I \in Q$ is the initial state, $F \subseteq Q$ is the set of final states, $C$ is a semilinear set in $\NN^d$ and $\Delta$ is the set of transitions of the form $p \era{a,S} q$ with $p$ and $q \in Q$, $a \in \Sigma$ and $S$ a semilinear set in $\mathbb{N}^d$.
	
A run $p_0 \era{a_0,S_0} p_1 \cdots p_{n-1} \era{a_{n-1},S_{n-1}} p_n$
is labeled by $(w,S)$ with $w=a_0 \cdots a_{n-1}$, $S=S_0+\cdots+S_{n-1}$. By convention, a run $\pi=q$ is labeled by $(\varepsilon,\{\vect{0}\})$.

It is accepting if it starts in the initial state (i.e., $p_0=q_I$) and ends in an accepting state (i.e., $p_n \in F$) and $S \cap C \neq \emptyset$.	

	A generalized $\PA$  is called (weakly)-unambiguous if every word has at most one accepting run.
\end{definition}

A $\PA$ can be seen as a generalized $\PA$ by replacing every transition $(p,(a,\vect{v}),q)$ with $(p,(a,\{\vect{v}\}),q)$.
This transformation preserves both the accepted language and weak-unambiguity. We will give in Proposition~\ref{prop:equivalence-PA-generalized-PA} a converse translation from generalized $\PA$ to $\PA$ which preserves weak-unambiguity.

For this, we need a technical lemma on semilinear sets.
For a  subset $L$  of $\NN^d$, we define $L^r$ by $L^0=\{\vect{0}\}$ and $L^{r+1}=L^{r}+L$ for $r \geq 0$. In other terms, $L^r$ contains the vectors that can be obtained as the sum of $r$ vectors in $L$. 

\begin{lemma}[\cite{Klaedtke02}]
\label{lem:semilinear-relation}
For a semilinear set $S$ in $\NN^d$, the relation $\vect{v} \in S^r$ is semilinear, meaning that there exists a Presburger formula $\varphi(x,y_1,\ldots,y_d)$ such that for all $r,v_1,\ldots,v_{n-1}$ and $v_n \in \NN$, $\varphi[r,v_1,\ldots,v_n]$ holds if and only if $(v_1,\ldots,v_d)$ belongs to $S^r$. 	
\end{lemma}

\begin{proof}
We first consider the case where $S=\vect c + P^*$ is a linear set. For all $\vect{v} \in \NN^d$ and $r > 0$, we have $\vect{v} \in S^r$ if and only if $\vect{v}-r \vect{c} \in P^*$. This is due to the fact that $(P^*)^r=P^*$ when $r>0$. Hence the formula $\varphi_S$, given below, defines the relation $\vect{v} \in S^r$.
\[
\begin{array}{lcl}
\varphi_S(x,y_1,\ldots,y_d) & \eqdef & x=0 \rightarrow v_1=\cdots=v_d=0 \\
&&  \wedge\; x\neq0 \rightarrow \varphi_P(v_1-c_1 x,\ldots,v_d-c_d x)  \\
\end{array}
\]
where $\varphi_P(x_1,\ldots,x_d)$ is a Presburger formula defining $P$. Remark that the multiplication by is $c_i$  allowed as the vector $\vect{c}$ is a constant.

Let us now handle the general case where $S$ is semilinear, of the form $S=\cup_{i=1}^n S_i$ where the $S_i$'s are linear sets. We define: 
	\begin{align*}
		\varphi_{S}(r,\vect v)&=\exists \ell_1, \ldots, \ell_n,\ \exists \vect{z_1}, \ldots, \vect{z_n},\\
		& \phantom{\wedge}\ \ \ell_1+\ldots+ \ell_n=r \wedge \vect v = \vect{z_1}+ \ldots + \vect{z_n} \\
		& \wedge \bigwedge_{i=1}^n \varphi_{S_i}(\ell_i,\vect{z_i})
	\end{align*}
	where $\vect{z_i}$ designates the vector of $d$ variables $(z_{i,1}, \ldots, z_{i,d})$.
	
	Let $\vect v\in\NN^d, r\in\NN$ such that $\vect v\in S^r$.
	If $r=0$ then $\vect v=\vect 0$ and $\varphi_{S}[\vect v, r]$ holds.
	Otherwise, $r>0$ and $\vect v$ is the sum of $r$ vectors from $S$: we can write $\vect v= \vect{v_1}+\ldots+\vect{v_r}$, where $\vect{v_j}\in S$ for $1\leq j\leq r$. 
	
	For $1\leq i\leq n$, let's denote $D_i$ the disjoint sets of vectors among $\vect{v_1}, \ldots, \vect{v_r}$ that belong to $L_i$, and don't belong to $L_k$ for $k<i$. Reordering the sum $\vect v= \vect{v_1}+\ldots+\vect{v_r}$ by grouping the vectors according to the sets $D_i$, we can rewrite it into $\vect{v}= \vect{z_1}+\ldots+\vect{z_n}$, where $\vect{z_i}$ is the sum of every vectors in $D_i$ (and is the null vector if $D_i$ is empty). Let's denote $\ell_i=|D_i|$. Then $\vect{z_i}\in L_i^{\ell_i}$ and $\ell_1+\ldots +\ell_n=r$. So $\varphi_{S}[r,\vect v]$ is true.
	
	Reciprocally, let $r\in\NN$ and $\vect v\in\NN^d$ satisfying $\varphi_{S}$.
	Then we can find $\ell_1, \ldots, \ell_n \in\NN$ such that $\ell_1+\ldots+ \ell_n=r$, and vectors $\vect{z_1}, \ldots, \vect{z_n}\in\NN^d$ such that $\vect v = \vect{z_1}+ \ldots + \vect{z_n}$ and $\vect{z_i}\in L_i^{\ell_i}\subseteq S^{\ell_i}$.
	
	If $r=0$, then every $\ell_i$ equals $0$, so $\vect{z_i}=\vect 0$ for every $i$. Then $\vect v=\vect 0$, and consequently $\vect v\in S^r$.

	If $r>0$, then some $\ell_i$ are not null. The equality $\vect v = \vect{z_1}+ \ldots + \vect{z_n}$ can be rewritten  as $\vect v= \vect{v_1}+\ldots+\vect{v_r}$, where $\vect v_i\in S$. So $\vect v\in S^r$.

\end{proof}

We can now provide a translation from generalized $\PA$ to $\PA$ which preserves weak-unambiguity.

\begin{proposition}
\label{prop:equivalence-PA-generalized-PA}
	Let $\AA$ a generalized $\PA$ (resp. weakly-unambiguous generalized $\PA$), the language $L(\AA)$ is recognized by a PA (resp. weakly-unambiguous $\PA$) $\BB$ with the same number of states and whose dimension is equal to the number of different semilinear sets appearing in the transitions of $\AA$.
\end{proposition}

\begin{proof}
    Fix a generalized $\PA$ $\AA=(Q, q_0, F, \Sigma, \Delta, S)$. Let $S_1,\ldots,S_r$ be an enumeration of the different semilinear sets appearing in the transitions of $\AA$.

    We construct a $\PA$ $\BB=(Q,q_0,F,\Sigma,\Delta',S')$ of dimension $r$ accepting $L(\AA)$. The automaton $\BB$ shares the same structure as $\AA$ (i.e., the same set of states, initial state, final states) and only differs by the transitions and semilinear constraint.  For every transitions $p \xrightarrow{a,S_i} q$ of $\AA$, there is a
    transition $p \xrightarrow{a,\vect{e_i}}q$, where $\vect{e_i}$ is the vector that has a $1$ in its $i$th coordinate, $0$ everywhere else. The semilinear constraint  is $S':=\{(n_1, \ldots, n_r)\in\NN^r\ :\ \exists \vect{v}\in S\cap \sum_{i=1}^r S_i^{n_i}\}$. It is indeed semilinear, since it is recognized by the formula:
		\[\varphi_{S'}(n_1, \ldots, n_r)=\exists \vect{v_1},\ldots,\vect{v_r},\ \vect{v_1}+\ldots+\vect{v_r}\in S\ \wedge\  \bigwedge_{i=1}^r \vect{v_i}\in S_i^{n_i}\]
	where $\vect{v_i}\in S_i^{n_i}$ is syntactic sugar for the formula constructed in Lemma~\ref{lem:semilinear-relation}.

	We denote by $\varphi$ the transformation taking a run $\pi$ of $\AA$ and translating it into a run of $\BB$ by replacing every transition in $\Delta$ by its associated transition in $\Delta'$. The mapping $\varphi$ is a bijection between the runs of $\AA$ and the runs of $\BB$ preserving the states and the input word.

To show that $\AA$ and $\BB$ accept the same languages, we are going to prove that for every run $\pi$ of $\AA$, $\pi$ is accepting for $\AA$ if and only if $\varphi(\pi)$ is accepting for $\BB$. In particular, as $\varphi$ is a bijection preserving the input word, this implies both that $L(\AA)=L(\BB)$ and that if $\AA$ is weakly unambiguous then so is $\BB$.
	
Let $\pi$ be an accepting run in $\AA$ labeled by $(w, S_\pi)$. Let $m_i$ denote the number of occurrences of $S_i$ in the transitions of $\pi$, for $1\leq i\leq r$ and let $\vect{m}$ denote the vector $(m_i)_{i \in [1,r]}$. By definition,  the semilinear set $S_\pi$ labelling the run $\pi$ is $S_\pi=\sum_{i=1}^r S_i^{m_i}$. The run $\varphi(\pi)$ of $\BB$ is labeled by $(w, \vect m)$ and starts in the initial state and ends in a final state. To show that $\varphi(\pi)$ is accepting for $\BB$, we need to show that $\vect{m}\in S'$. By definition of $S'$, this is equivalent to  $S \cap \sum_{i=1}^r S_i^{m_i} = S \cap S_\pi \neq \emptyset$ which holds as $\pi$ is accepting for $\AA$.

Let $\pi'$ be an accepting run in $\BB$ labeled by $(w, \vect{m})$, with $\vect{m}=(m_i)_{i} \in S'$. Then $\pi=\varphi^{-1}(\pi')$ is a run in $\AA$ labeled by $(w, S_\pi)$ and starting in the initial state and ending in a final state. Furthermore $m_i$ is the number of occurrences of the semilinear set $S_i$ in the transition of $\pi$ and consequently $S_\pi=\sum_iS_i^{m_i}$. To show that $\pi$ is accepting for $w$, we need to prove that $S_\pi \cap S\not=\emptyset$. By definition, as $\vect{m}$ belongs to $S'$, $\sum_iS_i^{m_i} \cap S = S_\pi \cap S \neq \emptyset$.
\end{proof}

\subsection{$\PA$ with $\varepsilon$-transitions}
\label{sec:eps-pa}

A $\PA$ with $\varepsilon$-transitions is defined is the same fashion as a $\PA$ except that we allow transitions of the form $(p,\varepsilon,\vect{v},q)$ where $\varepsilon$ denotes the empty-word. The notions of run and accepted languages are adapted in the usual way. 

In the non-deterministic case, it is known   that adding $\varepsilon$-transitions does not increase the expressivity of the model \cite{Klaedtke03}. A little more care is needed to ensure that the $\varepsilon$-removal preserves weak-unambiguity.

\begin{proposition}
\label{prop:eps-removal-wupa}
For every weakly-unambiguous $\PA$ with $\varepsilon$-transitions, there (effectively) exists an equivalent weakly-unambiguous $\PA$ $\BB$ without $\varepsilon$-transitions. Furthermore $\BB$ can be constructed with $n$ states and dimension $n^2 t$ if $\AA$ has $n$ states and $t$ transitions. 
\end{proposition}

\begin{proof}
Let $\AA=(\Sigma, Q, \qinit, F, C, \Delta)$ be a weakly-unambiguous PA with $\varepsilon$-transitions. Without loss of generality, we assume that the initial state $\qinit$  has no incoming transitions. Thanks to Proposition~\ref{prop:equivalence-PA-generalized-PA}, it is enough to construct a  weakly-unambiguous generalized $\PA$ $\BB$ without $\varepsilon$-transitions accepting $L(\AA)$.

For every pair of states $(p,q)$ of $\AA$, we consider the 
set, denoted by $C_{p,q}^\varepsilon$, of vectors $\vect{v}$ in $\NN^d$ which can label a run $\pi$ of the form $p \eRa{\varepsilon,\vect{v}} q$. 

Notice that $C_{p,q}^\varepsilon$ is semilinear. Indeed, it is recognized by the automaton over $(\NN^d,+)$ obtained from $\AA$ by setting the initial state to $p$, the final state to $q$, and  keeping only the $\varepsilon$-transitions. 

The automaton $\BB$ is defined as follows: it has the same set of states, initial and final states as $\AA$ does.

For every pair $(q, r)$ of states of $\AA$, for every transition $(p, (a, \vect{d}), q)\in\Delta$, where $a\not=\varepsilon$ and $p\not=q_I$, we add to $\Delta'$ the transition $(p, (a, (\{\vect d\}+C_{q,r}^\varepsilon), r)$, and the transition $(q_I, (a, (\{\vect{d}\}+C_{q_I,p}^\varepsilon+ C_{q,r}^\varepsilon), r)$. Since $q_I$ has no incoming transition, there can only be exactly one use of a transition of the form $(q_I, (a, (\{\vect{d}\}+C_{q_I,p}^\varepsilon+ C_{q,r}^\varepsilon), r)$ in a run of $\AA'$, at the very beginning.

Let us show that $L(\AA)\subseteq L(\BB)$. Let $w=w_1\ldots w_r\in L(\AA)$, and $q_I\eRb{(w, \vect d)}{\pi} q_f$ its accepting run  in $\AA$, with $\vect d\in S$. We can decompose $\pi=\pi^0_{\varepsilon}\pi^1_{w_1}\ldots \pi^{r-1}_{\varepsilon}\pi^r_{w_r}\pi^r_{\varepsilon}$, where each $\pi^i_{w_i}$ is a transition that reads the letter $w_i$, and each $\pi^i_{\varepsilon}$ is a (possibly empty) path of $\varepsilon$-transitions.
\\
At the beginning of the run, we have a run of the form $q_I\eRb{(\varepsilon, \vect{v_0})}{\pi^0_{\varepsilon}}p\xrightarrow{(w_1,\vect{d_1})}q\eRb{(\varepsilon, \vect{v_1})}{\pi^1_{\varepsilon}}r$, with by definition $\vect{v_0}\in C_{q_I,p}^\varepsilon$ and $\vect{v_1}\in C_{q,r}^\varepsilon$. By construction, we have in $\Delta'$ the transition $(q_I, (w_1, (\{\vect{d}\}+C_{q_I,p}^\varepsilon+ C_{q,r}^\varepsilon), r)$. For written convenience we write $L_0=C_{q_I,p}^\varepsilon$ and $L_1=C_{q,r}^\varepsilon$.
\\
For any $2\leq j\leq r$, $\pi^j_{w_j}$ and $\pi^j_{\varepsilon}$ are of the form $\pi^j_{w_j}=p\xrightarrow{w_j,\vect{d_j}}q$ and  $\pi^j_{\varepsilon}=q\eRa{(\varepsilon, \vect{v_j})} r$, where $p,q,r$ are three states and $\vect{v_j}\in C_{q,r}^\varepsilon$, so that in $\Delta'$ we have the corresponding transition $(p, (w_j, (\{\vect{d_j}\}+C_{q,r}^\varepsilon), r)$. For written convenience we write $L_j=C_{p,q}^\varepsilon$. Notice that this decomposition is unambiguous, as there is only one way to decompose $\pi$ under this form.

This leads to a run in $\BB$ of the form $q_I\eRb{(w, L)}{\pi_\BB} q_f$ with $L=L_0+\ldots+L_r+\{\vect{d_1}\}+\ldots \{\vect{d_r}\}$. It is accepting if $L\cap S\not = \emptyset$. It is immediate since $\vect{v_0}+\ldots + \vect{v_r}+\vect{d_1}+\ldots +\vect{d_r}=\vect d \in S$, and $\vect{v_i}\in L_i$ for $0\leq i \leq r$.

Let us now show that $L(\BB)\subseteq L(\AA)$. 
Let $w=w_1\ldots w_r\in L(\BB)$, and $q_I\eRb{(w, L)}{\pi'} q_f$ its accepting run in $\BB$, with $S\cap L\not = \emptyset$. We can decompose $\pi'=\pi'_1\pi'_2\ldots \pi'_r$, where $\pi'_1$ is of the form $q_I \xrightarrow{w_1, \{\vect{d_1}\}+L_0+L_1}r$, where $p,q,r$ are three states, $L_0=C_{q_I,p}^\varepsilon$ and $L_1=C_{q,r}^\varepsilon$. Likewise, for $j\leq 2$, each $\pi'_j$ is of the form $p \xrightarrow{w_j, \{\vect{d_j}\}+L_j}r$ for $p,q,r$ three states, and $L_j=C_{q,r}^\varepsilon$.

As $S\cap L\not = \emptyset$, let's take a $\vect d \in S\cap L$. As $L=L_0+\ldots+L_r+\{\vect{d_1}\}+\ldots+ \{\vect{d_r}\}$, we can find $\vect{v_i}\in L_i$ for $0\leq i \leq r$ such that $\vect{d}=\vect{d_1}+\ldots+ \vect{d_r}+\vect{v_0}+\ldots +\vect{v_r}$.

As $\vect{v_0}\in L_0$ and $\vect{v_1}\in L_1$, it means that from $\pi'_1$ we can build a run in $\AA$ of the form  $\pi_1=q_I\eRb{(\varepsilon, \vect{v_0})}{\pi^0_{\varepsilon}}p\xrightarrow{(w_1,\vect{d_1})}q\eRb{(\varepsilon, \vect{v_1})}{\pi^1_{\varepsilon}}r$.

Likewise, for $2\leq j\leq r$, as $\vect{v_j}\in L_j$, we can write from $\pi'_j$ a run in $\AA$ of the form $\pi_j=p \xrightarrow{w_j,\vect{d_j}}q\eRb{(\varepsilon, \vect{v_j})}{\pi_j^\varepsilon} r$.

This leads to run in $\AA$, $q_I\eRb{w, \vect d}{\pi}q_f$, which is accepting since $\vect d \in S$. So $w\in L(\AA)$.

Notice that we have defined a construction that transforms a run $\pi$ of $\AA$ into a run $\pi'$ of $\BB$, that is obtained by contracting the $\varepsilon$ paths in $\pi$. This construction is well defined, there is only one way to contract the run into a run of $\BB$. Reciprocally, every run from $\BB$ is the contraction of at least one run in $\AA$.

Suppose that $w\in L(\BB)$ is accepted by two different runs $\pi'_1, \pi'_2$ in $\BB$. Then we can create two accepting runs $\pi_1, \pi_2$ in $\AA$ labeled by $w$, such that $\pi'_1$ is the contraction of $\pi_1$, and $\pi'_2$ is the contraction of $\pi_2$. By weak unambiguity of $\AA$, $\pi_1=\pi_2$. As there is only one way to contract a run of $\AA$ into a run of $\BB$, $\pi'_1=\pi'_2$.
\end{proof}

\subsection{Proof of Proposition~\ref{prop:unambparbcm}}
\label{app:unambparbcm}

The aim of this section is to prove that the class of languages accepted by unambiguous (two-way) RBCM and by weakly-unambiguous PA coincide. This equivalence is stated in Proposition~\ref{prop:unambparbcm} which is restated for convenience below.

\equivparbcm*

The result is known to hold in the non-deterministic case \cite{Klaedtke03}. The proof in the unambiguous case follows the same general outline but some care is need to ensure that the different steps preserve unambiguity/weak-unambiguity. Using the normalizations of RBCM obtained in \cite{Ibarra78}, we can show that unambiguous RBCM accept the same languages as weakly-unambiguous PA with $\varepsilon$-transitions.

\begin{proposition}
The unambiguous RBCM and the weakly-unambiguous PA with $\varepsilon$-transitions accept the same languages.
\end{proposition}

\begin{proof}
\newcommand{\MM}{\mathcal{M}}

For the direct inclusion, consider an unambiguous two-way RBCM $\MM$ with $k$-counters each making at most $r$ reversals. In \cite{Ibarra78}, the author constructs an equivalent one-way RBCM $\MM'$ with $k'\geq k$ counters, each making at most one reversal. With very little adaptation, this construction can be shown to preserve unambiguity. To simplify the presentation, we present the construction of an equivalent weakly-unambiguous $\PA$ $\AA$ in the case of one counter. The adaptation to several counters is straightforward. 

The weakly-unambiguous PA $\AA$ is of dimension 3 and has states of the form $(q,\tau)$ with $q$ a state of $\MM'$ and $\tau \in \{0,1,2\}$. Intuitively, $\AA$ simulates $\MM'$ using the first coordinate of the output vector to store the number of increments, and the second coordinate stores the number of decrements. In particular, the value of the counter is obtained by subtracting the second coordinate to the first one. During the simulation, $\AA$ keeps track using the $\tau$-component of whether $\MM'$ is incrementing or decrementing. More precisely, 
\begin{itemize}
	\item $\tau=0$ if the machine $\MM'$ has never incremented (and the counter is equal to $0$),
	\item $\tau=1$ if the machine has incremented its counter but has never decremented,
	\item $\tau=2$ if the machine has decremented but the counter is not-null,
	\item $\tau=3$ if the machine has decremented and the counter is null.
\end{itemize}
The automaton $\AA$ increments its third component by $\tau$ upon accepting.

When $\tau=0$ or $\tau=3$, $\AA$ can only simulate transitions of $\MM'$ guarded by $c=0$ and when $\tau=1$ or $\tau=2$, only transitions guarded by $c\neq 0$.

The automaton $\AA$ can deterministically compute if $\tau=0$ or $1$ but it needs to non-deterministically guess if and when, it should enter $3$. The semilinear constraint of $\AA$ is used to check that the guess on value of $\tau$ is correct. More precisely, the semilinear constraint is given by the formula $\phi(x_1,x_2,x_3) \eqdef x_1 = x_2 \Leftrightarrow x_3=0 \vee x_3=3$. To ensure that the simulation is weakly-unambiguous, we ask that the automata enters $\tau=3$ from $\tau=1$ and $\tau=2$ only when performing a decrement on the counter.

For the converse inclusion, consider a weakly-unambiguous $\PA$ $\AA=(\Sigma, Q, q_I, F, C, \Delta)$ with $\varepsilon$-transitions of dimension $d$. We construct an unambiguous RBCM $\MM$ accepting $L(\AA)$. The machine $\MM$ is one-way and has $d$ counters which are used to store the coordinates of the vector in $\mathbb{N}^d$ computed by $\AA$. The machine $\MM$ works in two phases: it first simulates $\AA$ on the input word only incrementing the counters and then checks that the computed vector belongs to the semilinear set $C$ by simulating an unambiguous automaton over $\mathbb{N}^d$ accepting $C$. During the first phase, for every transition $(p^\AA,a,\vect{v},q^\AA)$ of $\AA$, the machine $M$ in state $p^\AA$ on the letter $a$ can increment the counter $i$ by $\vect{v}(i)$ and move the head to right and entering state $q^\AA$. Upon reaching the end-marker in a final state of $\AA$,
the machine $\MM$ enters the second phase. In this phase, it simulates an unambiguous automaton over $\mathbb{N}^d$ accepting $C$ \cite{EILENBERG1969}. The machine $\MM$ enters the initial state of $\CC$ and for each transition $(p^\CC,\vect{v},q^\CC)$ of $\CC$, the machine in state $p^\CC$ on the end-marker the machine can decrement the $i$-th counter by $\vect{v}(i)$. If such a decrement would decrease the value of a counter below 0, the machine rejects.
The machine $\MM$ accepts on the end-marker in a final state of $\CC$ if the value of all $d$ counters is 0. 

As the automaton accepting $C$ is assumed to be unambiguous, the accepting runs of $\MM$ are in bijection with the accepting runs of $\AA$. Hence $\MM$ accepts $L(\AA)$ and is  unambiguous as $\AA$ is weakly-unambiguous.
\end{proof}

Combined with Proposition~\ref{prop:eps-removal-wupa}, the previous proposition proves Proposition~\ref{prop:unambparbcm}.

\subsection{Proof omitted in Section~\ref{sec:rcm}}
\subsubsection{Proof of equivalence between RCM and weakly-unambiguous PA}
\label{app:rcm-wupa}
\eqparcm*
We prove it by showing how to recognize a language from RCM by a weakly-unambiguous PA, and how to do the reverse. Both directions are quite syntactical, since the two classes are very similar in their definition. The key idea is to see that the constraint of injectivity of the morphism for RCM corresponds to the unicity of an accepting run in the world of Parikh automata.

\paragraph*{From RCM to PA.} This direction is easy. Let $L\in RCM$. Let $\Gamma=\{a_1\ldots a_r\}$, $\Sigma$ two alphabets, $R$ a regular language over $\Gamma$, $C$ some semilinear set in $\NN^r$ regarded as constraints over the number of occurrences of the symbols of $\Gamma$, and finally $\mu:\Gamma^*\longrightarrow\Sigma^*$ a length preserving morphism that is injective over $R\cap [C]$ (where $[C]$ denotes the set of words verifying the constraints in $C$), such that  $L=\mu(R\cap [C])$.

We will write in the following for a word $w\in\Gamma^*$,  $\bm{\pi_{kh}}(w)=(|w|_{a_1}, \ldots, |w|_{a_r})$ its Parikh image.
We suppose that $R$ is recognized by a deterministic automaton $\mathcal A$, and that $C$ is represented by a Presburger formula $\Phi_c$. With these notations $[C]=\{w : \Phi_c(\bm{\pi_{kh}}(w))\}$.

The Parikh automaton recognizing $L$ can be easily obtained from this by simply replacing in $\mathcal A$ every transition of the form $(q_1, a_i, q_2)$ with $(q_1, (\mu(a_i), \vect{e_{a_i}}), q_2)$ where $\vect{e_{a_i}}$ is a vector in dimension $|\Gamma|=r$ filled with $0$ except for the $i$th coordinate, which is $1$. We keep the same Presburger formula for $P$: $\Phi_C$.

There is a bijection between the transition of $\mathcal A$ and $P$, and consequently between the paths in each automaton. For every run in $P$ labeled by $(w\in\Sigma^*,\vect d\in\NN^r)$, there exists a run in $\mathcal A$ taking exactly the same states, labeled by a word $w_2\in\Gamma^*$ such that $\mu(w_2)=w$ and $\vect{d}=\bm{\pi_{kh}}(w_2)$. On the other direction, every run in $\mathcal A$ labeled by $w$ can be translated in a run in $P$ labeled by $(\mu(w), \bm{\pi_{kh}}(w))$, taking the same states.

\begin{lemma}$L=\mathcal L(P)$ and $P$ is weakly-unambiguous.\end{lemma}
\begin{proof}
By double inclusion.
\begin{itemize}
\item[\textcolor{itemi}{\bm{$\subseteq$}}] Let $w\in L$. There exists a word $w_2\in R$ such that $\mu(w_2)=w$ and $\phi_c(\bm{\pi_{kh}}(w_2)))$ is true. So there exists an accepting run in $\mathcal A$ for $w_2$, then a run in $P$ with the same states labeled by $(w, \bm{\pi_{kh}}(w_2))$. As $\phi_c(\bm{\pi_{kh}}(w_2)))$ is true, this run is accepting in $P$. So $w\in\mathcal L(P)$.
\item[\textcolor{itemi}{\bm{$\supseteq$}}] Let $w=b_1\ldots b_s\in\mathcal L(P)$. There exists an accepting run in $P$ labeled by $((b_1,\vect{d_1}), \ldots,$ $(b_s,\vect{d_s}))$ so that $\phi_c(\vect d)$ is true, with $\vect{d}=\vect{d_1}+ \ldots+ \vect{d_{s}}$. Then there exists a path in $\mathcal A$ labeled by a word $w_2\in\Gamma^*$ such that $\mu(w_2)=w$ and $\vect{d}=\bm{\pi_{kh}}(w_2)$, going through exactly the same states.  This run is accepting in $\mathcal A$, so $w_2\in R$. Consequently $w_2\in R\cap[C]$ and finally $w\in\mu(R\cap[C])=L$.

Let's suppose that there exists an other accepting run in $P$ labeled by $((b_1,\vect{d'_1}), \ldots, (b_s,\vect{d'_s}))$ such that $\phi_c(\vect d')$ is true with $\vect{d'}=\vect{d'_1}+ \ldots+ \vect{d'_{s}}$. As there is a bijection between the transitions of $\mathcal A$ and $P$, this path is in bijection with a path in $\mathcal A$, labeled by a word  $w_3\in R \cap [C]$. As $\mu(w_3)=\mu(w_2)=w$, the injectivity of $\mu$ over $R\cap[C]$ leads to $w_3=w_2$. Consequently $\vect{d'_1}=\vect{d_1}$, \ldots, $\vect{d_s}=\vect{d'_s}$. Besides, $\mathcal A$ is deterministic so both paths go through the same states. So the two runs in $P$ are identical.

For every word there is at most one accepting run. So $P$ is weakly-unambiguous.
\end{itemize}
\end{proof}

\paragraph*{From PA to RCM}

Let $P$ be a weakly-unambiguous Parikh automaton, $\Phi_P$ its formula.

\begin{lemma}\label{vecteurunitepa}
There exists a weakly-unambiguous PA $P'$ recognizing the same language as $P$ such that its transitions are labeled with vectors of the form $\vect{e_i}$, \textit{e.g.} every transition only changes one coordinate by incrementing it. 
\end{lemma} 
\begin{proof}
It suffices to see $P$ as a generalized PA and apply the construction of  Proposition~\ref{prop:equivalence-PA-generalized-PA}, which only uses vectors of the announced form. 
\end{proof}

We can now suppose that $P$ is a weakly-unambiguous PA, whose transitions are of the form $(p, (a, \vect{e_i}),q)$ with $i\in\intint{1}{d}$, and $d$ is the dimension of $P$.

We then create the automaton $\mathcal A$ over the alphabet $\Gamma=\Sigma\times\{\vect{e_i}\}_{i\in\intint{1}{d}}$, simply obtained from $P$ by regarding it as a regular automaton over the alphabet $\Gamma$ (seeing the vectors as letters to concatenate). We define the morphism $\mu$ as the projection over the first coordinate: $\mu((a, \vect{e_i}))=a$.

Let $R$ be the language recognized by $\mathcal A$. A word in $R$ is nothing else than a path in $P$, and a word in $R\cap C$ is an accepting path of $P$. So $\mathcal L(P)=\mu(R\cap [C])$.

Let's finally prove that $\mu$ is injective over $R\cap [C]$. Let $w_1, w_2\in R\cap [C]$ such that $\mu(w_1)=\mu(w_2)$. If $w_1\not=w_2$, then we would have two different accepting runs in $P$ that would be labeled by the same word $\mu(w)$, contradicting that $P$ is weakly-unambiguous.

\subsubsection{Equivalence of definitions for RCM}
\label{app:eq-def-RCM}

The definition of RCM in \cite{Castiglione17} actually differs from the definition given in this article (see 
Section~\ref{sec:rcm}). In \cite{Castiglione17}, the RCM class is defined with a subset of the semilinear sets: the authors only consider semilinear sets which can be expressed as boolean combinations of inequalities: they do not allow quantifiers nor equalities modulo constants. For this section, we will call such semilinear sets \intro{modulo-free semilinear sets}.
Modulo-free semilinears form a strict subclass of the semilinears.
For instance the semilinear set $\{2n \;|\; n \geq 0\}$ cannot be defined using only inequalities and is therefore not a modulo-free semilinear. However this difference does not impact the resulting RCM class.

To show this, we will use the equivalence between RCM and weakly-unambiguous PA. Remark that the translation in both directions transforms a modulo-free semilinear set into a modulo-free semilinear set. Hence to prove the equivalence of RCM defined with general semilinears and RCM defined with modulo-free semilinears, it is enough to establish the same equivalence for weakly-unambiguous PA. A similar idea can be found for affine PA in \cite{Cadilhac12a}.

\begin{lemma}
For every weakly-unambiguous PA $P=(Q,\{q_I\}, F, \Delta, \mathcal S)$, there exists an equivalent 
weakly-unambiguous PA $P'$ with a modulo-free semilinear constraints set.
\end{lemma} 
\begin{proof}
 
 Without loss of generality, we can assume that the semilinear constraint of $P$ is defined by a Presburger formula $\Phi_P$ with no quantifier \cite{Pres29}, but is using predicates of the form $t \equiv 0 [K]$, where $t(x_1, \ldots, x_d)= a_0 + a_1x_1+ \ldots + a_dx_d $ is a semilinear term.
	
	Let suppose that $\Phi_P$ contains a predicate $\mathcal P$ of the form $t \equiv 0 [K]$, where $t(\vec x)=a_0 + a_1x_1+ \ldots + a_dx_d $.
	
	We construct an equivalent weakly-unambiguous automaton $P'$, verifying that $\Phi_{P'}$ is simply obtained from $\Phi_{P}$ by replacing the predicate $\mathcal P$ by a predicate without modulo. Iterating this transformation over every predicate that is using a modulo in $\Phi_P$ leads to the result.
	
	So we focus on eliminating the modulo in the predicate $\mathcal P$. The idea to build $P'$ is to add in the states of $P$ the value of $t[n_1, \ldots, n_d] \mod K$, where $(n_1, \ldots, n_d)$ denotes the current values of the counters of the automaton. It relies on the form of $t$, so that we only need to know the value of  $t[\vec n] \mod K$ to know the value of $t[\vec n+\vec d] \mod K$, for any $\vec d$ appearing in a transition of $P$. Indeed, $t[\vec n +\vec d]\equiv t[\vec n]+t[\vec d]-a_0 \mod K$.
	
	More formally, $P'=(Q', \{(q_I,a_0)\}, F', \Delta', \mathcal S')$ is a PA of dimension $d+1$. The set of states of $P'$ is $Q'= \{q_f\}\cup (Q_P\times \intint{0}{K-1})$, the final state is $q_f$, and the set of transitions is $\Delta'=\Delta_1\cup \Delta_2$, where $\Delta_1$ is defined by:
	\[\Delta_1=\{\big((p,i), (a, (\vect{d},0)), (q,k)\big)\,:\, (p,(a,\vect{d}),q)\in\Delta, k=(i+t[\vec d]-a_0\hspace{-0.2cm}\mod K)\}\]
	and $\Delta_2$ is defined by:
	\[\Delta_2=\{\big((p,i), (a, (\vect{d},k)), q_f\big)\,:\, q\in F, (p,(a,\vect{d}),q)\in\Delta, k=(i+t[\vec d]-a_0\hspace{-0.2cm}\mod K)\}.\]
	The role of $\Delta_2$ is just to use the new dimension to be able to reach in the semilinear the value of $t[n_1, \ldots, n_d] \mod K$, where $(n_1, \ldots, n_d)$ denotes the values of the counters when entering the accepting state.
	
 Hence $P'$ simulates $P$ by storing in its state the value of $t[n_1, \ldots, n_d] \mod K$, and non deterministically transfers this value in its new dimension before entering the final state, from which there are no outgoing transitions. The Presburger formula $\Phi_{P'}$ defining $\mathcal S'$ is obtained by replacing the predicate $\mathcal P$ in $\Phi_P$ by the predicate $x_{d+1}=0$.
 
 It is straightforward that $L(P)=L(P')$  and that the new automaton $P'$ is weakly-unambiguous.
\end{proof}
\section{Proof omitted in Section~\ref{sec:wuppa}}\label{sec:C}

\begin{proposition}
\label{prop:non-closure-wuppa}
The class of weakly-unambiguous \PPA is not closed under union nor intersection. 
\end{proposition}
\begin{proof}
    Recall the language $\mathcal D$ over the alphabet $\{a,b\}$, defined in Section~\ref{sec:2 examples} by:
\[\mathcal D= \{\underline{n_1}\ \underline{n_2}\ldots\underline{n_{k}}\ :\ k\in\NN^*,\ n_1=1\ \textnormal{and}\ \exists j<k, n_{j+1} \not= 2n_{j} \},\text{ where }\underline{n}=a^nb,\]
 and the language $\overline{\mathcal D}=ab(a^*b)^*\setminus\mathcal D$. It has been explained in Section~\ref{sec:2 examples} that the generating series of $\overline{\mathcal D}$ is not holonomic.
 
To show that the class of weakly-unambiguous \PPA is not closed under union and intersection, we use some languages introduced in \cite{flajolet87}, that have been useful to prove inherent ambiguity for context-free languages. Both the languages $\mathcal I_1=\{\underline{n}_1\ldots\underline{n}_k\ |\  (n_1=1 \textnormal{ and }\forall j, n_{2j}=2n_{2j-1})$ and $\mathcal I_2=\{\underline{n}_1\ldots\underline{n}_k\ |\ (\forall j, n_{2j+1}=2n_{2j})\}$ are deterministic context-free languages, so they are recognized by two weakly-unambiguous \PPA. Furthermore their generating series are holonomic (and even algebraic).

Notice that the language $\mathcal I_1\cap \mathcal I_2$ is exactly the language $\overline{\mathcal D}$, so its series is not holonomic. So by Proposition~\ref{prop:PPA_holonomic}, $\mathcal I_1\cap \mathcal I_2$ is not recognized by a weakly-unambiguous \PPA.

For the non-closure under union, 
consider the language $\mathcal L=\mathcal I_1\cup \mathcal I_2$. If we denote by $L(x)$ (resp. $I_1(x)$, $I_2(x)$ and $ \overline D(x)$) the generating series of $\mathcal L$ (resp. $\mathcal I_1$, $\mathcal I_2$ and $\overline{\mathcal D}$), we have that:\[L(x)=I_1(x)+I_2(x)-\overline D(x).\]

If $L(x)$ was holonomic, as $I_1(x)+I_2(x)$ is algebraic hence holonomic, by closure under sum $\overline D(x)$ would be holonomic, which is not the case. So $L(x)$ is not holonomic, and $\mathcal I_1\cup \mathcal I_2$ is not recognized by a weakly-unambiguous \PPA. 

So the class of weakly-unambiguous \PPA is not closed under union or intersection.

Notice that the language $\mathcal I_1\cup \mathcal I_2$ is recognized by an ambiguous \PPA, so the arguments above show that it is inherently weakly-ambiguous.
\end{proof}

\weakpparbcm*

Notice that the result is  easy if we allow $\varepsilon$-transitions in the definition of a \PPA, and a direct adaptation of the equivalence of PA and RBCM with $\varepsilon$-transitions. Hence, the heart of the proof is to eliminate $\varepsilon$-transitions in a weakly-unambiguous \PPA.

\subsection{Removing the $\varepsilon$-transitions in a weakly-unambiguous \PPA}

To remove  $\varepsilon$-transitions in \PPA, 
we consider the equivalent formalism of constrained context-free grammars, CCFG for short (i.e., context-free grammars with semilinear constraints). We show in Proposition~\label{prop:PPAtoccfg}
 that every language accepted by a weakly-unambiguous \PPA is a accepted by a weakly-unambiguous CCFG. Then we show that CCFG can be transformed in Greibach normal form while preserving weak-unambiguity (see Proposition~\ref{prop:greibach}). Finally of weakly-unambiguous CCFG in Greibach normal form can be translated back to a weakly-unambiguous \PPA without $\varepsilon$-transitions.

Here we need to adapt standard proofs for context-free grammars to the world of weighted grammars with semilinear constraints, and ensure that they preserve (weak)-unambiguity. The Greibach normal form has already been studied for weighted grammars, over semirings (see for instance \cite{10.1007/Bonsangue2012} or \cite{STANAT1972}), or in the particular case of probabilistic grammars (see \cite{Abney1999}, or  \cite{ponty2012}). In an other direction, the survival of the unambiguity property during the normalization process has been studied for context-free grammars in \cite{cheung}. 

In this proof, we study the normalization process in the case where the weights are vectors or semilinear sets in $\NN^d$.

\subsubsection{\ccfg: constrained context-free grammars}

\begin{definition} A \intro{constrained context-free grammar} of dimension $d$ is a tuple $(N,\Sigma, S, D, \mathcal S)$ where:
\begin{enumerate}
\item $N$ is the set of non terminals, $\Sigma$ the set of terminal letters,
\item $S\in N$ is the axiom
\item $\mathcal S\subseteq \NN^d$ is a semilinear set
\item $D$ is a set of production rules of the form $A\xrightarrow{\vect v}A_1\ldots A_r$, where $A\in N$, $r\geq 0$, $\vect v\in\NN^d$, and $A_i\in N\cup\Sigma$ for $1\leq i\leq r$.
\end{enumerate}

In order to give the definition of what a derivation tree is in this formalism, we need to introduce briefly what a tree labeled by semilinear set is. Let's fix a dimension $d$.

\begin{definition}
	Let $P$ be a set of symbols. The set $Tree_{P}$ of trees labeled by semilinear sets over $\NN^d$, with nodes in $P$, is defined inductively:
	\begin{itemize}
		\item every symbol $v\in P$ is an ordered semilinear-labeled tree over $P$
		\item if $r\geq 1$, $v\in P$, $L$ is a semilinear set of dimension $d$, and $t_1, \ldots, t_r$ are in $Tree_{P}$, then $(v, L, t_1, \ldots, t_r)\in Tree_{P}$. 
	\end{itemize}
\end{definition}
\begin{definition} We define by induction:
	\begin{itemize}
		\item for $v\in P$, $root(v)=v$ and if $t=(v, L, t_1, \ldots, t_r)$, $root(t)=v$.
		\item for $v\in P$, $S(v)=\{0\}$ and if $t=(v, L, t_1, \ldots, t_r)$, $S(t)=L+\sum_{i=1}^r S(i)$ denotes the sum of semilinear sets appearing in $t$.
		\item for $v\in P$, $left(v)=v$ and if $t=(v, L, t_1, \ldots, t_r)$, $left(t)=left(t_1)$ denotes the left-most symbol in the tree.
		\item for $v\in P$, $Fr(v)=v$ and if $t=(v, L, t_1, \ldots, t_r)$, $Fr(t)=Fr(t_1)\ldots Fr(t_r)$ denotes the frontier of the tree.
		\item for $v\in P$, $subtrees(v)=\{v\}$ and if $t=(v, L, t_1, \ldots, t_r)$, $subtrees(t)=\{t\}\cup subtrees(t_1)\cup\ldots\cup subtrees(t_r)$ denotes the set of subtrees of $t$.
		\item for $v\in P$, $ht(v)=0$ and if $t=(v, L, t_1, \ldots, t_r)$, $ht(t)=1+\max(ht(t_1),\ldots, ht(t_r))$ denotes the height of a tree.
	\end{itemize}
\end{definition}
\begin{figure}[h]
	\centering
\usetikzlibrary{shapes.geometric}
\tikzset{lefttria/.style={
  draw,shape border uses incircle,shape border rotate=90,
  isosceles triangle, yshift=0.645cm, xshift=0.08cm}}
 \tikzset{lefttriaun/.style={
  draw,shape border uses incircle,shape border rotate=90,
  isosceles triangle, yshift=0.645cm, xshift=0.07cm}}
  \tikzset{righttria/.style={
  draw,shape border uses incircle,shape border rotate=90,
  isosceles triangle, yshift=0.645cm, xshift=-0.08cm}}

\begin{tikzpicture}[thick,scale=0.6, every node/.style={scale=0.6}]
    		[level 1/.style={sibling distance=1.5cm}]
			\node[label={-90:\color{blue}$L$}] {$v$}
  				child{ node{}
  						child{node[lefttria] {$t_1$} edge from parent[draw=none]}
       			}
       			child{ node{}
  						child{node[lefttriaun] {$t_2$} edge from parent[draw=none]}
				}
       			child{ node {$\cdots$}
       			}
  				child{ node{}
  						child{node[righttria] {$t_r$} edge from parent[draw=none]}
				};
			\end{tikzpicture}
			\caption{Representation of a tree of root $v$ with $r$ children, labeled by a semilinear set $L$}
\end{figure}
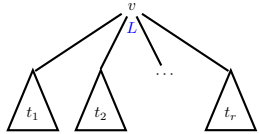

\begin{definition}[derivation tree]
	Let $G=(N, \Sigma, S, D, \mathcal S)$ a constrained context-free grammar.
	For all $v\in N\cup\Sigma$, the set of derivation trees rooted in $v$, denoted $Tree_G(v)$, is defined by induction:
			\begin{itemize}
				\item for $v\in N\cup \Sigma$, $v\in Tree_G(v)$
				\item if $v\in N$, and $v\xrightarrow{L}v_1\ldots v_r \in D$ is a production rule, if $t_i\in Tree_G(v_i)$ for $1\leq i \leq r$, then $(v, L, t_1, \ldots, t_r)\in Tree_G(v)$.
			\end{itemize}Let's denote $Tree_G=\cup_{v\in N\cup\Sigma}Tree_G(v)$ the set of all derivation trees of $G$.
\end{definition}

A word $w\in\Sigma^*$ is generated by the constrained grammar if there exists a derivation tree $t\in Tree_G(S)$ such that $Fr(t)=w$ and the sum of the vectors appearing in $t$ belongs to $\mathcal S$.

A \ccfg\ is said \intro{weakly-unambiguous} if for every word $w$ there exists at most one derivation tree accepting $w$. 
\end{definition}

\begin{proposition}\label{prop:PPAtoccfg}
	Every weakly-unambiguous \PPA can be turned into an equivalent weakly-unambiguous \ccfg{} of same dimension.
\end{proposition}
\begin{proof}
 We simply adapt the classical Ginsburg triplets method. We consider a weakly-unambiguous \PPA of dimension $d$, $\BB=(Q,q_I,F, \Sigma,\Gamma,\bot,\delta,\mathcal S)$, where $Q$ is the set of states, $q_I$ is the initial state, $F$ the set of final states, $\Sigma$ the set of letters, $\Gamma$ the set of stack symbols, $\bot\in\Gamma$ the starting stack symbol, $\delta$ the set of transitions, and $\mathcal S$ the semilinear set.
 
 It is not difficult to transform $\BB=(Q,q_I,F, \Sigma,\Gamma,\bot,\delta,\mathcal S)$ into an equivalent automaton having the additional property that when it accepts a word, the stack of the automaton is empty: the transformation consists in adding two new states $q_{\emptyset}$ and $q_f$ to $Q$, and replace the set of final states $F$ by $F'=\{q_f\}$. The automaton adds a new symbol $\$\notin \Gamma$ in the stack, then simulates $\BB$. Whenever the automaton is in a state $q\in F$, it can non deterministically go to state $q_{\emptyset}$ by an $(\varepsilon,\vect 0)$-transition. From state $q_{\emptyset}$, the automaton no longer reads any letter, and only uses $(\varepsilon,\vect 0)$-transitions. It empties its stack until it sees $\$$, then goes to $q_f$, which has no outgoing transitions. It is not difficult to see that this transformation preserves the recognized language and the weak-unambiguity of the original automaton $\BB$.

Without loss of generality, we can also suppose that the transitions of the \PPA add to the stack at most one symbol, by adding intermediate states.

So without loss of generality, we suppose that $\BB$ adds to its stack at most one symbol, and only accepts with an empty stack.

The following transformation into a grammar is an adaptation of the classical triplets method, with weights. The corresponding grammar is defined by $G_{\BB}=(N,\Sigma, S, D, \mathcal S)$ where $N\subseteq\{S\}\cup\{[q,q',z]\,:\,q,q'\in Q, z\in\Gamma\cup\{\varepsilon\} \}$, and $D$ is defined by:
\begin{itemize}
\item $[q, q', z] \xto{\vect d} y$ if $(q,z \xto{(y,\vect d)}q', \varepsilon) \in \delta$
\item $[q, q', z] \xto{\vect d} y[q'', q', z_1]$ if $(q,z \xto{(y,\vect d)}q'', z_1) \in \delta$
\item $[q, q', z] \xto{\vect d} y[q_1, q_2, z_1][q_2, q', z_2]$, for every $q_2\in Q$, if $(q,z \xto{(y,\vect d)}q_1, z_1z_2) \in \delta$
\item $S \xto{\vect 0}[q_I,q_F,\bot]$ for every $q_F\in F$
\end{itemize}

A left-derivation of a word $w$ from a symbol $V\in N$ is a sequence $w_1\xrightarrow{r_1} \ldots \xrightarrow{r_{n-1}} w_n$, where $w_1=V$, $w_n=w$, and for $i\in[n]$, $w_i\in(\Sigma\cup N)^*$  and $w_{i+1}$ is obtained from $w_i$ by applying the rule $r_i\in D$ to the leftmost non terminal symbol in $w_i$.

For any symbol $V\in N$, we define $D_{left}(V)$ the set of every left derivations starting from $V$.

It is classical to show by induction that to every element of $D_{left}([q,q',z])$ of the form $[q,q',z]\xrightarrow{r_1}w_2 \ldots \xrightarrow{r_{n-1}} w_n=w$, we can associate a run in $\BB$ that starts in state $q$, with top stack symbol $z$, such that this stack symbol will be popped for the first time when arriving in state $q'$, after reading exactly the word $w$, with vector value the sum of the vectors of $r_1, \ldots, r_n$. Furthermore this correspondence is a one-to-on correspondence, in the sense that the sequence of left derivations simulates exactly the runs of the automaton.

Hence, the left derivations of the grammar run the pushdown automaton, and every derivation of the grammar is a valid run of the pushdown automaton, both having the same weight. Consequently the notion of weak-unambiguity is preserved during this transformation.
\end{proof}

In the following, the weakly-unambiguous $\ccfg$ are under the form $G=(N,\Sigma, S, D, \mathcal S)$, and verify conditions $(H)$:
\begin{enumerate}
	\item  every non terminal symbols is productive, and such that every rule has at most three symbols in its right member, the first symbol being a letter in $\Sigma \cup \{\varepsilon\}$;
	\item without loss of generality the axiom $S$ cannot be found in any production rule of the grammars
\end{enumerate}

Let us modify the rules of such grammars to obtain an equivalent weakly-unambiguous grammar such that the head symbol of every right member is a letter in $\Sigma$ (this is the Greibach normal form).

\subsubsection{Removal of nullable symbols}

A symbol $A$ is said nullable if $A\not = S$ and if $A\to^*\varepsilon$. Notice that $\varepsilon$ is a nullable.

We want to modify the rules of a grammar in order to remove the nullable symbols, by using the usual algorithm.

If  $A$ is nullable, we denote by $T_{A}$ the set of vectors that label a derivation $A\to^*\varepsilon$.

\begin{claim}
	$T_{A}$ is semilinear.
\end{claim}
\begin{claimproof}
Let us create a simple context-free grammar $G'=(N,\Sigma', A, D')$ from $G$ over the alphabet $\Sigma'=\{a_1, \ldots, a_d\}$.

	The grammar $G'$ is obtained from $G$ by changing the axiom to $A$, and by removing from $D$ any production rule writing a letter in $\Sigma$. Finally the weights are suppressed by discarding $\mathcal S$, and changing every production rule of the form $B\xrightarrow{\vect v}B_1\ldots B_r$ by the rule $B\to a_1^{v_1}\ldots a_d^{v_d}B_1\ldots B_r$ where $r\geq 0$, $B_i\in N$ for $1\leq i\leq r$. A terminal derivation tree of $G'$ simulates the derivation of $G$ starting from $A$ and producing the word $\varepsilon$, and there is a one-to-one correspondence between them. Furthermore the Parikh image of the word recognized by a derivation tree of $G'$ corresponds to the vector labelling $\varepsilon$ in $G$. By Parikh's theorem, it is semilinear, so by the one-to-one correspondence, $T_{A}$ is semilinear.
\end{claimproof}

\begin{proposition}\label{prop:removalnullable}
	Let $G=(N,\Sigma, S, D, \mathcal S)$ a weakly-unambiguous $\ccfg$ verifying condition $(H)$. Then we can build a  generalized weakly-unambiguous $\ccfg$ $G'$, having no nullable symbol, still verifying condition $(H)$, such that $L(G')=L(G)\backslash\{\varepsilon\}$.
\end{proposition}
\begin{proof}
	
We adapt the classical ideas for context-free grammars. The $\ccfg$ $G'$ is defined by $G'=(N,\Sigma, S, D', \mathcal S)$. Notice that the only changes come from the derivation rules.

The set $D'$ is defined as follows:
\begin{itemize}
\item For every rule of $D$ of the form $t_i= A\xrightarrow{\vect{d_i}} \alpha$ with $\alpha\in\Sigma$ (\textit{i.e.} $\alpha\neq\varepsilon$), we add to $D'$ the rule $t_i^{\scriptscriptstyle 1}=A\xrightarrow{\{\vect{d_i}\}} \alpha$. This deletes every $\varepsilon$-production from the grammar.
\item  For every rule of $D$ of the form $t_i= A\xrightarrow{\vect{d_i}} \alpha B$ with $\alpha\in\Sigma\cup\{\varepsilon\}$ and $B\in N$, we add the following rules to $D'$:
	\begin{itemize}
	\item  $t_i^{\scriptscriptstyle 1}= A\xrightarrow{\{\vect{d_i}\}} \alpha B$;
	\item  $t^{\scriptscriptstyle 2}_i=A\xrightarrow{\{\vect{d_i}+T_B\}} \alpha$, if $B$ is nullable in $G$ and $\alpha\neq\varepsilon$.
	\end{itemize}
\item For every rule of $D$ of the form $t_i= A\xrightarrow{\vect{d_i}} \alpha BC$, with $B,C\in N$, and $\alpha\in \Sigma\cup\{\varepsilon\}$,  we add to $D'$ the following rules:
\begin{itemize}
	\item $t_i^{\scriptscriptstyle 1} = A\xrightarrow{\{\vect{d_i}\}} \alpha BC$;
	\item $t_i^{\scriptscriptstyle 2} = A\xrightarrow{\{\vect{d_i}\}+T_{B}} \alpha C$, if $B$ is nullable;
	\item $t_i^{\scriptscriptstyle 3} = A\xrightarrow{\{\vect{d_i}\}+T_{C}} \alpha B$, if $C$ is nullable;
	\item $t_i^{\scriptscriptstyle 4} = A\xrightarrow{\{\vect{d_i}\}+T_{B}+T_{C}} \alpha$, if $B$ and $C$ are nullable and $\alpha\neq \varepsilon$ (otherwise it would introduce an $\varepsilon$-production)
\end{itemize}
\end{itemize}

By construction, $G'$ has no nullable symbol, and still verifies condition $(H)$.

It is classical that any terminal derivation tree $T$ for a word $w\in\Sigma^*$ in $G$ can be turned into a terminal derivation tree $T'$ for $w$ in $G'$: it is obtained by erasing every maximal subtree of $T$ having $\varepsilon$ as a frontier, and by replacing the rules accordingly. It is clear that if $T$ is labeled by $\vect{d}$, then $T'$ is labeled by a semilinear set $L$ such that $\vect{d}\in L$, so that if $w\in L(G)$ then $w\in L(G')$.

Reciprocally, we can transform any terminal derivation tree $T'$ of $G'$ for a word $w\in\Sigma^*$, labeled by a semilinear set $L$, into a terminal derivation tree $T$ for $w$ in $G$, labeled by $\vect{d}$ such that $\vect{d}\in L$: we replace any rule of the form $t_i^r$ by the rule $t_i$ it comes from, and expand the tree by deriving, with rules in $D$, every symbol $B$ appearing in $t_i$ and not in $t_i^r$ into a tree of frontier $\varepsilon$ (which is labeled by a vector in $T_B$). If furthermore $w\in L(B)$, and $T'$ is an accepting derivation of $w$, then let $\vect{d}\in L\cap \mathcal S$. By definition, $L$ is the sum of every semilinear set appearing in its derivation rules. So we can choose wisely the values of the vector labelling the trees of frontier $\varepsilon$ when creating $T$ from $T'$ to obtain a derivation tree of vector value $\vect{d}\in \mathcal S$, which proves that $w\in L(G)$.

Finally, $G'$ is weakly-unambiguous: by the construction above, two different accepting derivation trees for a word in $G'$ would otherwise lead to two different accepting derivation trees for the same word in $G$, which is weakly-unambiguous.
\end{proof}

\subsubsection{Removal of unit rules}

Let a generalized weakly-unambiguous $\ccfg$ $G=(N,\Sigma, S, D, \mathcal S)$ of dimension $d$, verifying condition $(H)$, without any nullable symbol.

A unit rule is a rule of the form $A\to B$, where $A$ and $B$ are non terminal symbols.
For every couple of non terminal symbols $(A, B)\in N^2$, we denote by $T_{(A,B)}$ the union of every semilinear sets labelling a derivation in $G$ of the form $A\xto{unit^+}B$. 

\begin{claim}
The set $T_{A,B}$ is semilinear.	
\end{claim}
\begin{claimproof}
	It suffices to consider the automaton of dimension $d$ described as follows. The set of state contains $N$, the initial state is $A$ and the final state is $B$. For every unit rule of the form $B\xrightarrow{L}C$ in $G$, we link $B$ by a $\vect{0}$-transition to an automaton recognizing $L$, then we connect the final states of this automaton to state $C$ by a $\vect{0}$ transition too. It is immediate that the set of vectors recognized by this automaton is $T_{A,B}$. Hence $T_{A,B}$ is semilinear.
\end{claimproof}

\begin{proposition}\label{prop:removalunit}
	Let a generalized weakly-unambiguous $\ccfg$ $G=(N,\Sigma, S, D, \mathcal S)$ of dimension $d$, verifying condition $(H)$, without nullable symbols (in particular $\varepsilon\notin L(G)$). Then we can build an equivalent generalized weakly-unambiguous $\ccfg$ $G'$ without nullable symbols nor unit-rule.
\end{proposition}
\begin{proof}
	We adapt the classical ideas for context-free grammars. Let $G'=(N,\Sigma, S, D', \mathcal S)$. Here again, the only difference comes from $D'$, which is defined as follows:
	\begin{itemize}
	\item Add to $D'$ every rule of $D$ that is not a unit-rule.
	\item For every non terminal symbol $B$, for every non unit rule $t_i=B\xto{L_i}\beta$ with $\beta\in (N\cup\Sigma)^+$ and $\beta\not\in N$, and for every symbol $A\in N$ such that $T_{(A,B)}\not=\emptyset$, add to $D'$ the rule: $t_i^A=A\xto{L_i+T_{(A,B)}}\beta$.
\end{itemize}
It is straightforward that the derivation trees of $G'$ are obtained by contracting the chains of unit rules in the derivation trees of $G$. The definition of $D'$ is such that the semilinear sets labelling the rules witnessing these contractions replace exactly the semilinear sets that label the removed chains of unit-rules. The previous reasoning for the elimination of nullable symbols can easily be transposed to this case to prove that $L(G)=L(G')$ and that $G'$ is weakly-unambiguous.\end{proof}

\subsection{Greibach normalization}

The removal of left-recursion is more delicate, since we want to preserve weak-unambiguity. To simplify the proof, we consider a variant of the classical way to eliminate left-recursion.

In the following we fix a dimension $d\in\NN^+$.
			
\newcommand{\context}{Context}
			
\begin{definition}
	Let $G=(N, \Sigma, S, D, \mathcal S)$ a constrained context-free grammar.
	\begin{itemize}
		\item A derivation tree $C\in Tree_G$ is called a $V$-context for $V\in N\cup \Sigma$ if $left(t)=V$ and $Fr(t)\in V\Sigma^*$. %
		\item if $C\in Tree_G$ is a $V$-context with $V\in N$%
, and $t\in Tree_G(V)$ is a derivation tree of root $V$, we write $C[t]\in Tree_G$ the derivation tree defined by:
			\begin{itemize}
				\item $C[t]=t$ if $C=V$
				\item $C[t]=(v, L, t_1[t], t_2, \ldots, t_r)$ if $C=(v, L, t_1, \ldots, t_r)$. Notice that in this case $t_1$ is a $V$-context.
			\end{itemize} Notice that $C[t]$ is a $left(t)$-context.
		\item We write $\context_G[V]$ the set of all $V$-contexts for $V\in N\cup\Sigma$, and more generally $\context_G[A]$ the set of all $V$-contexts for $V\in A$, where $A$ is a set of symbols (in practice we will only take $A$ equals to $N$, $\Sigma$ or $N\cup \Sigma$).
		\item We write $\context_G^P[V]$ for $P,V\in N\cup \Sigma$ the set of $V$-contexts having $P$ as a root, and $\context_G^P[\Sigma]$ the set of terminal derivation trees having $P$ as a root.
	\end{itemize}
\end{definition}

In the following $G=(N, \Sigma, S, D, \mathcal S)$ is a weakly-unambiguous constrained context-free grammar, free of unit-rule and $\varepsilon$-productions.

\begin{proposition}\label{prop:decompostree} Let $V\in N\cup\Sigma$ and $t$ a $V$-context of height greater than $1$. Then there exists  a production rule of the form $B\xrightarrow{L}V v_1\ldots v_r$ with $r\geq 0$, derivation trees $t_i\in Tree_G(v_i)$, and $t'$ a $B$-context such that $t$ can be decomposed $t=t'[(B,L,V,t_1, \ldots, t_r)]$. The choice of the decomposition is unique. Furthermore:
\begin{itemize}
	\item $r\geq 1$ if $V\in N$, since $G$ is unit-free 
	\item if $r\geq 1$, for $1\leq i\leq r$, $Fr(t_i)\in\Sigma^+$. In particular $left(t_i)\in\Sigma$.
\end{itemize}
\end{proposition}
\begin{proof}
	It is immediate by induction.
\end{proof}

\begin{figure}[h]
\centering
\includegraphics[scale=0.4]{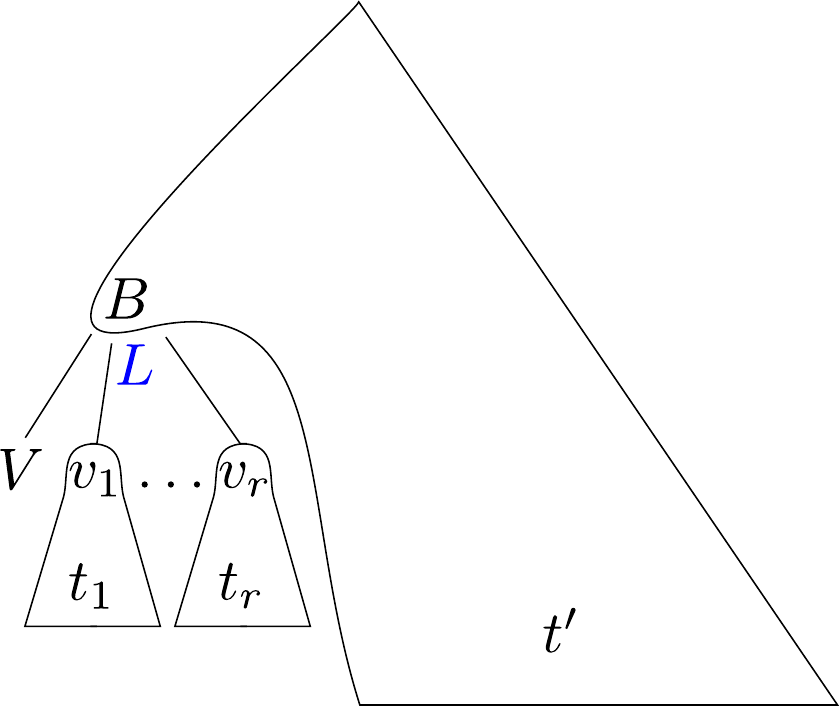}
\caption{Decomposition of a $V$-context under the form $t=t'[(B,L,V,t_1, \ldots, t_r)]$}
\end{figure}
  
\newcommand{\boxx}[1]{\langle #1 \rangle}

\begin{definition}
	Let's define $\Phi$ an application from $\context_G[\Sigma\cup N]$ to $Tree_{\Sigma\cup N}$ as follows:
	\begin{itemize}
		\item for $v\in\Sigma\cup N$, $\Phi(v)=v$
		\item if $ht(t)\geq 1$, then it has a unique decomposition of the form $t=t'[(B,L,V,t_1, \ldots, t_r)]$ following proposition \ref{prop:decompostree}, where $r\geq 0$, $B\xrightarrow{L}V v_1\ldots v_r \in D$, $t_i\in Tree_G(v_i)$, and $t'$ is a $B$-context:
			\begin{itemize}
				\item if $left(t)\in\Sigma$, $\Phi(t)=(\boxx{root(t')},L, left(t), \Phi(t_1), \ldots, \Phi(t_r), \Phi(t'))$, where we omit the last child $\Phi(t')$ if $t'$ is reduced to $left(t')$.
				\item if $left(t)\in N$, $\Phi(t)=(\boxx{root(t'),left(t)},L, \Phi(t_1), \ldots, \Phi(t_r), \Phi(t'))$, where we omit the last child $\Phi(t')$ if $t'$ is reduced to $left(t')$. Recall that if $left(t)\in N$ then $r\geq 1$.
			\end{itemize}			
	\end{itemize}
\end{definition}

\begin{figure}[h]
\centering
\includegraphics[scale=0.4]{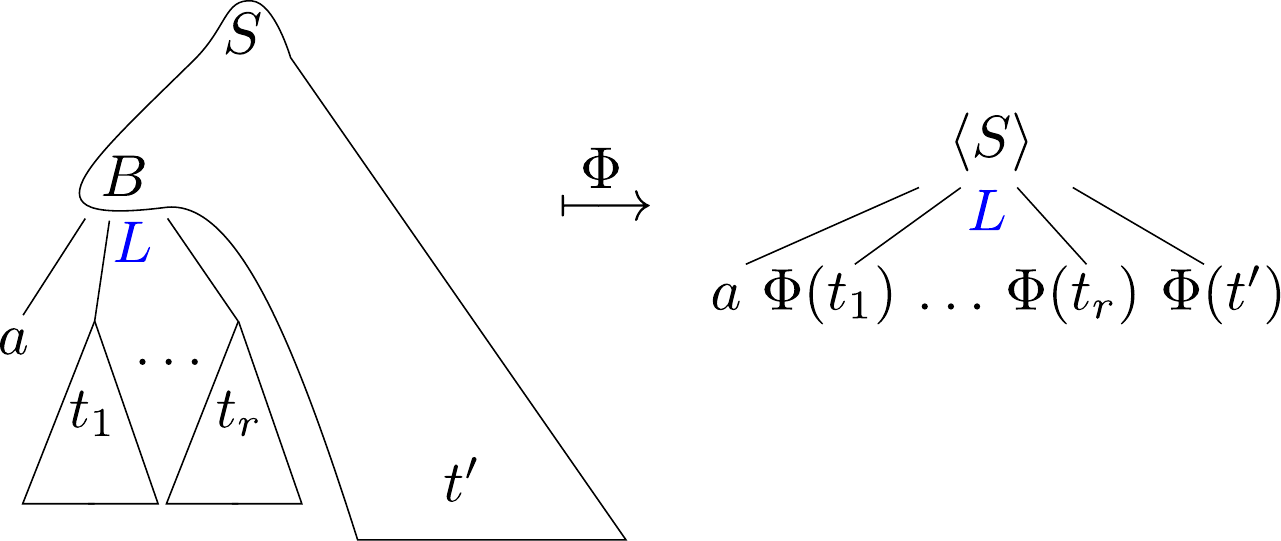}
\caption{Action of $\Phi$ on $t$ an $a$-context, where $left(t)=a\in\Sigma$, and $t'$ is not reduced to $root(t)=S$}
\end{figure}

\begin{figure}[h]
\centering
\includegraphics[scale=0.4]{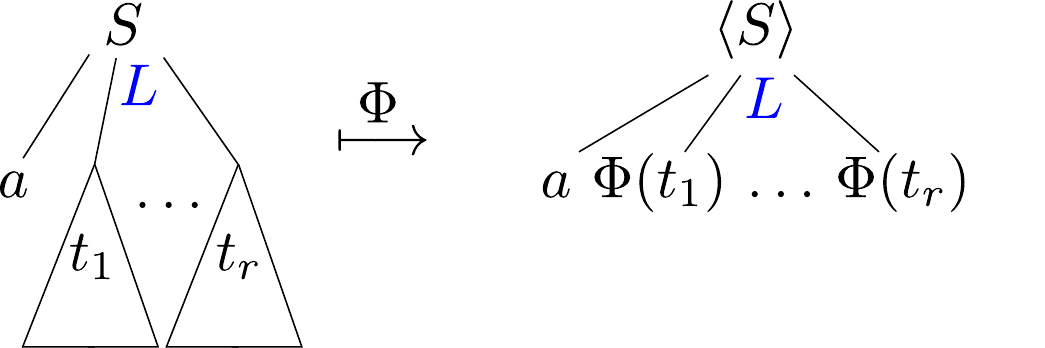}
\caption{Action of $\Phi$ on an $a$-context, where $a\in\Sigma$, and $t'$ is reduced to $S$}
\end{figure}

\begin{figure}[h]
\centering
\includegraphics[scale=0.4]{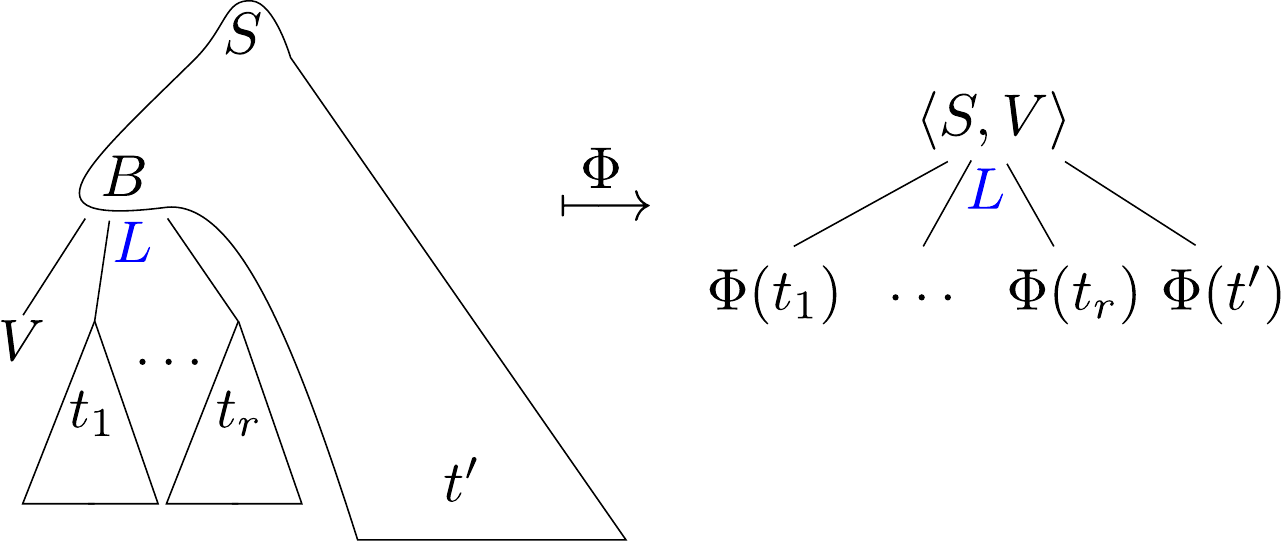}
\caption{Action of $\Phi$ on a $V$-context, where $V\in N$, and $t'$ is not reduced to $S$}
\end{figure}

\begin{figure}[h!]
\centering
\includegraphics[scale=0.4]{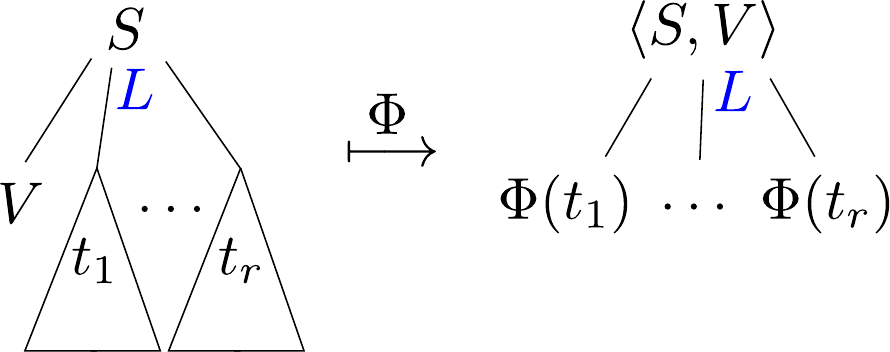}
\caption{Action of $\Phi$ on a $V$-context, where $V\in N$, and $t'$ is reduced to $S$}
\end{figure}

The following proposition is easily checked by induction: 
\begin{proposition} If $t\in \context_G[\Sigma\cup N]$, $t_\Sigma\in \context_G[\Sigma]$, and  $t_N\in \context_G[N]$:
	\begin{itemize}
		\item $|\Phi(t_\Sigma)|=|t_\Sigma|$ and $|\Phi(t_N)|=|t_N|-1$ where $|t|$ denotes the number of symbols in a tree.
		\item $\mathcal S(\Phi(t))=\mathcal S(t)$, and $ht(\Phi(t))=ht(t)$.
		\item $Fr(\Phi(t_\Sigma))=Fr(t_\Sigma)$, and $Fr(\Phi(t_N))=left(t_N)^{-1}Fr(t_N)$ if $ht(t_N) \geq 1$.
		\item if $t'$ is a subtree of $\Phi(t)$, then $t'$ is the image of a subtree of $t$ by $\Phi$.
	\end{itemize}
\end{proposition}
\begin{proof}
	By immediate induction.
\end{proof}

\begin{proposition}\label{prop:phiinjective}
	$\Phi$ is injective on $\context_G[N\cup \Sigma]$.
\end{proposition}
\begin{proof}
	By induction over the height of the trees in $\context_G[N\cup \Sigma]$. As $\Phi$ preserves height, it suffices to check it with trees of the same height.
	\begin{itemize}
		\item $\Phi$ is the identity over the trees of height $0$.
		\item Let $t_1, t_2 \in \context_G[N\cup \Sigma]$ of same height bigger than 1, such that $\Phi(t_1)=\Phi(t_2)$. By proposition \ref{prop:decompostree}, $t_1$ can be decomposed by $t_1=t'_1[(B_1,L_1,V_1,T_1, \ldots, T_r)]$, and $t_2=t'_2[(B_2,L_2,V_2,T'_1, \ldots, T'_R)]$.
	As $\Phi(t_1)=\Phi(t_2)$, we have automatically that $L_1=L_2$, and that they have the same root.
	
	If it is of the form $\boxx{B_1}$, this means that $B_1=B_2$, and that the first child of $\Phi(t_1)$ is a letter. By construction, it means that $left(t_1)=left(t_2)\in\Sigma$.
	Otherwise it is of the form $\boxx{root(t_1),left(t_1)}$. So $B_1=B_2$ and $left(t_1)=left(t_2)\in N$. So in both cases $left(t_1)=left(t_2)$ and $B_1=B_2$.
	
	Let's show that $t'_1=B_1$ if and only if $t'_2=B_2$. Without loss of generality let's suppose that $t'_1=B_1$ but $t'_2\not = B_2$. In this case, $r\geq 1$, and the last child of $\Phi(t_1)$ is $\Phi(T_r)$. The last child of $\Phi(t_2)$ is $\Phi(t'_2)$, so $\Phi(T_r)=\Phi(t'_2)$. By induction hypothesis, $T_r=t'_2$. This leads to a contradiction since $Fr(T_r)\in \Sigma^+$ whereas  $Fr(t'_2)\in N\Sigma^*$.
	
	Consequently $r=R$, and $\Phi(T_i)=\Phi(T'_i)$ for all $1\leq i\leq r$ by induction hypothesis. Furthermore either $t'_1=B_1$, $t'_2=B_2$ so $t'_1=t'_2$,  or $t'_1\not = B_1$ and then $t'_2\not = B_2$, so $\Phi(t'_1)=\Phi(t'_2)$ and by induction hypothesis $t'_1=t'_2$.
	
	So $t_1=t_2$.
	
			\end{itemize}
\end{proof}

\begin{remark}
	\label{remark:phibij}
	$\Phi:\context_G[\Sigma\cup N]\to \Phi(\context_G[\Sigma\cup N])$ is a bijection.
\end{remark}

\begin{definition}
	We define a new constrained context-free grammars $G_\Phi=(N', \Sigma, \boxx S, D', \mathcal S)$ 
as follows, where $N'=\{\boxx{V}\ :\  V\in N\} \cup \{\boxx{V, V'}\ :\  V,V'\in N\}$ and $D'$ is defined as follows:
\begin{itemize}
	\item For every production rule in $D$ of the form $Q\xrightarrow L a P_1\ldots P_r$, where $a\in \Sigma$, and $P_i\in \Sigma\cup N$, we add to $D'$ the production rule \[\boxx Q\xrightarrow L a \boxx{P_1}\ldots \boxx{P_r},\] where for writing convenience, we take as a convention that $\boxx b=b$ for $b\in \Sigma$. Furthermore, for every $P\in N$, we add to $D'$ the production rule:\[\boxx P\xrightarrow L a \boxx{P_1}\ldots \boxx{P_r}\boxx{P,Q}\]
	\item For every production rule in $D$ of the form $Q\xrightarrow L R P_1\ldots P_r$, where $R\in N$, $r\geq 1$, and $P_i\in \Sigma\cup N$, we add to $D'$ the production rule \[\boxx{Q,R}\xrightarrow L  \boxx{P_1}\ldots \boxx{P_r},\]and for every $P\in N$, we add to $D'$ the production rule:\[\boxx{P,R}\xrightarrow L  \boxx{P_1}\ldots \boxx{P_r}\boxx{P,Q}\]
\end{itemize}
\end{definition}

\begin{example}
	We consider $G=(N, \Sigma, S_a, D, \mathcal S)$ given by $N=\{S_a,S_b\}$, $a,b\in \Sigma$, and $D$ is composed of the the two rules $S_a\xrightarrow{(1,0)}a|S_ba$ and $S_b\xrightarrow{(0,1)}b|S_ab$.
	Then $G_\Phi=(N', \Sigma, \boxx{S_a}, D', \mathcal S)$ is defined by $N'=\{\boxx{S_a}, \boxx{S_b}, \boxx{S_a, S_b}, \boxx{S_a, S_a}, \boxx{S_b, S_b}, \boxx{S_b, S_a} \}$, and $D'$ by:
	\begin{align*}
		\boxx{S_a}&\xrightarrow{(1,0)}a|a\boxx{S_a,S_a} & \boxx{S_a, S_b}&\xrightarrow{(1,0)}a|a\boxx{S_a,S_a}\\
		\boxx{S_a}&\xrightarrow{(0,1)}b\boxx{S_a,S_b}&\boxx{S_a,S_a}&\xrightarrow{(0,1)}b\boxx{S_a,S_b}\\
		\boxx{S_b}&\xrightarrow{(0,1)}b|b\boxx{S_b,S_b}&\boxx{S_b,S_a}&\xrightarrow{(0,1)}b|b\boxx{S_b,S_b}\\
		\boxx{S_b}&\xrightarrow{(1,0)}a\boxx{S_b,S_a}&\boxx{S_b, S_b}&\xrightarrow{(1,0)}a\boxx{S_b,S_a}\\
	\end{align*}
	
	Notice that the construction may generate some non terminals $\boxx{S_b}, \boxx{S_b, S_a}$ and $\boxx{S_b, S_b}$ that are not reachable from the axiom $\boxx{S_a}$.
	
	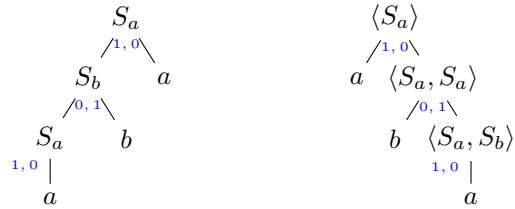
\begin{figure}[h]\centering
	\begin{tikzpicture}[level distance=0.8cm,
  level 1/.style={sibling distance=1cm},
  level 2/.style={sibling distance=1cm}]
  \node[label={[label distance=-0.10cm]-90:\color{blue}\tiny$1,0$}] {$S_a$}
   child { node[label={[label distance=-0.10cm]-90:\color{blue}\tiny$0,1$}] {$S_b$}
      		child { node[label={[label distance=-0.1cm]-100:\color{blue}\tiny$1,0$}] {$S_a$}
      				child {node {$a$}
      					  }
                  }
     	    child {node {$b$}}
    	  }
    child {node {$a$}};
\end{tikzpicture}\hspace{2cm}
\begin{tikzpicture}[level distance=0.8cm,
  level 1/.style={sibling distance=1cm},
  level 2/.style={sibling distance=1cm}]
  \node[label={[label distance=-0.10cm]-90:\color{blue}\tiny$1,0$}] {$\boxx{S_a}$}
   child {node {$a$}}
   child { node[label={[label distance=-0.10cm]-90:\color{blue}\tiny$0,1$}] {$\boxx{S_a, S_a}$}
   			child {node {$b$}}
      		child { node[label={[label distance=-0.1cm]-100:\color{blue}\tiny$1,0$}] {$\boxx{S_a, S_b}$}
      				child {node {$a$}
      					  }
                  }
    	  };
\end{tikzpicture}
\caption{On the left, a derivation tree $t$ of $G$, and on the right, a derivation tree $t'$ of $G_\Phi$, such that $\Phi(t)=t'$.}	
	\end{figure}
	
\end{example}

\begin{proposition}\label{prop:phiisccfg}
	$\context_{G_\Phi}[N']=\Phi(\context_G[N\cup\Sigma])$. More precisely, for $P,Q\in N$, 
	\begin{enumerate}
		\item The set of terminal derivation trees of $G_\Phi$ having $\boxx{P}$ as a root is the image by $\Phi$ of the set of terminal derivation tress of $G$ having $P$ as a root. In other words:
		\[\context_{G_\Phi}^{\boxx{P}}[\Sigma]=\Phi(\context_G^P[\Sigma]).\]
		\item The set of terminal derivation trees of $G_\Phi$ having $\boxx{P,Q}$ as a root is the image by $\Phi$ of the set of $Q$-contexts from $G$ having $P$ as a root. In other words \[\context_{G_\Phi}^{\boxx{P,Q}}[\Sigma]=\Phi(\context_G^P[Q]).\] 
	\end{enumerate}
\end{proposition}

\begin{proof}
	By double inclusion and induction. We have to prove both affirmations simultaneously.

Let's prove by induction that for every $P,Q\in N$, $\context_{G_\Phi}^{\boxx{P}}[\Sigma]\subseteq\Phi(\context_G^P[\Sigma])$ and $\context_{G_\Phi}^{\boxx{P,Q}}[\Sigma]\subseteq\Phi(\context_G^P[Q])$. Let's do the basic cases:

Let $t\in \context_{G_\Phi}^{\boxx{P}}[\Sigma]$ of height $1$, of the form $t=(\boxx{P}, L, a_1, \ldots, a_r)$ where $a_i\in \Sigma$ for $1\leq i \leq r$. As there is no symbol of the form $\boxx{P,Q}$, it means by construction that it is the application of the rule $\boxx P\xrightarrow{L}a_1\ldots a_r$, which comes from $P\xrightarrow{L}a_1\ldots a_r$ in $G$. Then $t_G=(P, L, a_1, \ldots, a_r)\in \context_G^P[\Sigma]$ and $\Phi(t_G)=t$.
\\
Reciprocally, let $t_G\in \context_G^P[\Sigma]$ of height $1$. Then it is of the form $t_G=(P, L, a_1, \ldots, a_r)$, with $a_i\in \Sigma$ for $1\leq i\leq r$. Then $\Phi(t_G)=(\boxx P, L, a_1, \ldots, a_r)$ and as $P\xrightarrow{L}a_1\ldots a_r$ is a derivation of $G$, then by construction $\boxx P\xrightarrow{L}a_1\ldots a_r$ is a derivation of $G_\Phi$. So $\Phi(t_G)\in \context_{G_\Phi}^{\boxx{P}}[\Sigma]$.

Let $t\in \context_{G_\Phi}^{\boxx{P,Q}}[\Sigma]$ of height $1$, of the form $t=(\boxx{P,Q}, L, a_1, \ldots, a_r)$ where $a_i\in \Sigma$ for $1\leq i \leq r$. As there is no symbol of the form $\boxx{P,Q}$ in the children, it means by construction that $t$ begins with the application of the rule $\boxx{P,Q}\xrightarrow{L}a_1\ldots a_r$, which comes from $P\xrightarrow{L}Qa_1\ldots a_r$ in $G$. Then $t_G=(P, L, Q, a_1, \ldots, a_r)\in \context_G^P[Q]$ and $\Phi(t_G)=t$.
\\
Reciprocally, let $t_G\in \context_G^P[Q]$ of height $1$. Then it is of the form $t_G=(P, L, Q, a_1, \ldots, a_r)$, with $a_i\in \Sigma$ for $1\leq i\leq r$. Then $\Phi(t_G)=(\boxx{P,Q}, L, a_1, \ldots, a_r)$ and as $P\xrightarrow{L}Qa_1\ldots a_r$ is a derivation of $G$, then by construction $\boxx{P,Q}\xrightarrow{L}a_1\ldots a_r$ is a derivation of $G_\Phi$. So $\Phi(t_G)\in \context_{G_\Phi}^{\boxx{P}}[Q]$.

Let's now handle the induction step.

\begin{itemize}
\item Let $t\in \context_{G_\Phi}^{\boxx{P}}[\Sigma]$ of height $>1$. There are two cases.
\begin{enumerate}
	\item First case: $t$ is of the form $t=(\boxx P, L, a, t_1, \ldots, t_r)$, and  begins with a derivation of the form $\boxx P\xrightarrow L a \boxx{P_1}\ldots\boxx{P_r}$, so that $root(t_i)=\boxx{P_i}$. By construction of $D'$, this means that $P\xrightarrow L a P_1\ldots P_r$ is a production rule of $G$. Furthermore, for $1\leq i\leq r$, $t_i\in \context_{G_\Phi}^{\boxx{P_i}}[\Sigma]$, so by induction hypothesis, there is $t'_i\in \context_G^{P_i}[\Sigma]$ such that $\Phi(t'_i)=t_i$.
	
	So $t_G=(P,L,a,t'_1, \ldots, t'_r)\in \context_G^P[\Sigma]$. And it is immediate that $\Phi(t_G)=t$.
	\item Second case: $t$ is of the form $t=(\boxx P, L, a, t_1, \ldots, t_r,T)$ and begins with a derivation of the form $\boxx P\xrightarrow L a \boxx{P_1}\ldots\boxx{P_r}\boxx{P,Q}$, so that $root(t_i)=\boxx{P_i}$, and $root(T)=\boxx{P,Q}$. By construction of $D'$, this means that $Q\xrightarrow L a P_1\ldots P_r$ is a production rule of $G$. Furthermore, for $1\leq i\leq r$, $t_i\in \context_{G_\Phi}^{\boxx{P_i}}[\Sigma]$, so by induction hypothesis, there is $t'_i\in \context_G^{P_i}[\Sigma]$ such that $\Phi(t'_i)=t_i$. Also by induction hypothesis, as $T\in \context_{G_\Phi}^{\boxx{P,Q}}[\Sigma]$,  there exists $T'\in \context_G^P[Q]$ such that $\Phi(T')=T$.
	
	So $t_G=T'[(Q,L,a,t'_1, \ldots, t'_r)]\in \context_G^P[\Sigma]$. And it is immediate that $\Phi(t_G)=t$.
	\end{enumerate}
\item Let $t\in \context_{G_\Phi}^{\boxx{P,Q}}[\Sigma]$ of height $>1$. There are also two cases.
\begin{enumerate}
	\item First case: $t$ is of the form $t=(\boxx{P,Q}, L, t_1, \ldots, t_r)$, and  begins with a derivation of the form $\boxx{P,Q}\xrightarrow L \boxx{P_1}\ldots\boxx{P_r}$, so that $root(t_i)=\boxx{P_i}$. By construction of $D'$, this means that $P\xrightarrow L Q P_1\ldots P_r$ is a production rule of $G$. Furthermore, for $1\leq i\leq r$, $t_i\in \context_{G_\Phi}^{\boxx{P_i}}[\Sigma]$, so by induction hypothesis, there is $t'_i\in \context_G^{P_i}[\Sigma]$ such that $\Phi(t'_i)=t_i$.
	
	So $t_G=(P,L,Q,t'_1, \ldots, t'_r)\in \context_G^P[Q]$. And it is immediate that $\Phi(t_G)=t$.
	\item Second case: $t$ is of the form $t=(\boxx{P,Q}, L, t_1, \ldots, t_r,T)$ and begins with a derivation of the form $\boxx{P,Q}\xrightarrow L \boxx{P_1}\ldots\boxx{P_r}\boxx{P,R}$, so that $root(t_i)=\boxx{P_i}$, and $root(T)=\boxx{P,R}$. By construction of $D'$, this means that $R\xrightarrow L Q P_1\ldots P_r$ is a production rule of $G$. Furthermore, for $1\leq i\leq r$, $t_i\in \context_{G_\Phi}^{\boxx{P_i}}[\Sigma]$, so by induction hypothesis, there is $t'_i\in \context_G^{P_i}[\Sigma]$ such that $\Phi(t'_i)=t_i$. Also by induction hypothesis, as $T\in \context_{G_\Phi}^{\boxx{P,R}}[\Sigma]$,  there exists $T'\in \context_G^P[R]$ such that $\Phi(T')=T$.
	
	So $t_G=T'[(R,L,Q,t'_1, \ldots, t'_r)]\in \context_G^P[Q]$. And it is immediate that $\Phi(t_G)=t$.
	\end{enumerate}
\end{itemize}
	
The other direction for the inclusion is on the same idea. We only treat one case, the other cases are similar.
Let $t\in \context_G^P[Q]$ of height strictly bigger $1$, with $Q\in N$. Let's denote $R=root(t)$. Then by Proposition~\ref{prop:decompostree} there exists  a production rule of the form $B\xrightarrow{L}Q P_1\ldots P_r$ with $r\geq 0$, derivation trees $t_i\in \context_G^{P_i}[\Sigma]$, and $t'\in Tree_G(R)[B]$ a $B$-context such that $t$ can be decomposed $t=t'[(B,L,Q,t_1, \ldots, t_r)]$. By definition, $\Phi(t)=(\boxx{R,Q}, L, \Phi(t_1), \ldots, \Phi(t_r), \Phi(t'))$ ($t'$ is not reduced to $R$ otherwise the height would be $1$).

By induction hypothesis, $\Phi(t_i)\in \context_{G_\Phi}^{\boxx{P_i}}[\Sigma]$, and $\Phi(t')\in \context_{G_\Phi}^{\boxx{R,B}}[\Sigma]$. Futhermore, by construction of $D'$, $\boxx{R,Q}\xrightarrow{L}\boxx{P_1}\ldots \boxx{P_r}\boxx{R,B}$ is a production rule in $D'$.

So $\Phi(t)\in \context_{G_\Phi}^{\boxx{R,Q}}[\Sigma]$.

\end{proof}

\begin{corollary}
	$L(G)=L(G_\Phi)$ and $G_\Phi$ is weakly-unambiguous.
\end{corollary}
\begin{proof}
	This comes from the fact that $\Phi$ is a bijection, by Proposition~\label{prop:phiinjective}, between $\context_G^S[\Sigma]$ and $\context_{G_\Phi}^{\boxx S}[\Sigma]$ that preserves both the frontier and the sum of the semilinear set, and that $G$ is weakly-unambiguous.
\end{proof}

\begin{proposition}\label{prop:greibach}
	There exists a weakly-unambiguous semilinearly constrained context-free grammars $G'_\Phi$ equivalent to $G_\Phi$ verifying that every production rule produces a letter first.
\end{proposition}
\begin{proof}
By construction of $G_\Phi$, every production having a symbol of the form $\boxx P$ as a left hand side begins with a letter in $\Sigma$. Every production having a symbol of the form $\boxx{P,Q}$ as a left hand side begins with either a letter or a non terminal of the form $\boxx{P_1}$. So $G_\Phi$ has no left-recursion. 

We define $G'_\Phi=(N', \Sigma, \boxx S, D'', \mathcal S)$. $D''$ is obtained from $D'$ by removing every rule of the form $\boxx{Q,R}\xrightarrow L  \boxx{P_1}\ldots \boxx{P_r}$ where $\boxx{P_1}$ is a non terminal.

Then for every deleted rule $\boxx{Q,R}\xrightarrow L  \boxx{P_1}\ldots \boxx{P_r}$, for every rule in $D'$ of the form $\boxx{P_1}\xrightarrow{L_1}a \boxx{Q_1}\ldots \boxx{Q_s}$, we add to $D''$ the rule 
\[ \boxx{Q,R}\xrightarrow{L+L_1}  a \boxx{Q_1}\ldots \boxx{Q_s}\ldots \boxx{P_r} \]

It is straightforward that $G'_\Phi$ recognizes the same language as $G_\Phi$ and is weakly-unambiguous.
\end{proof}

Thus we have proved the proposition:

\begin{proposition}\label{prop:greibach}
	Every weakly-unambiguous \ccfg{} $G$  is equivalent to a weakly-unambiguous \ccfg{} $G'$ such that every production rule of $G'$ starts with a symbol in $\Sigma$.
\end{proposition}

\subsubsection{From a \ccfg{} to a \PPA}

\begin{proposition}\label{prop:ccfgtoppa}
	Let $G=(N,\Sigma, S, D, \mathcal S)$ be a weakly-unambiguous \ccfg\ $G=(N,\Sigma, S, D, \mathcal S)$ such that every production rule begins with a non terminal letter (in particular $\varepsilon\notin L(G)$). Then $L(G)$ is recognized by a weakly-unambiguous \PPA{} without $\varepsilon$-transition.
\end{proposition}
\begin{proof}
Without loss of generality, we can suppose that the production rules of $G$ are in $\Sigma\times N^*$ without altering the language nor the weak-unambiguity of $G$.

	The classical construction considers the \PPA{} $\BB=(\{q_I\},q_I,\{q_I\}, \Sigma, N,S,\delta,\mathcal S)$ where $\Delta$ is defined as follows:
	\[\Delta=\{q_0, A \xto{a, L}q_0, \beta\,:\, (A\xto{L}a\beta)\in D, a\in\Sigma, \beta\in N^*\}\]
	
It is classical that the runs of the automaton are in bijection with the left-derivation of $G$, and that $L(G)=L(\BB)$. It is not difficult to see that the semilinear sets labelling the runs and the left-derivation that are in bijection coincide. Hence, if $G$ is weakly-unambiguous, $\BB$ is also weakly-unambiguous.
\end{proof}

\begin{proposition}
	Let $\BB=(Q,q_I,F, \Sigma,\Gamma,\bot,\delta,\mathcal S)$ a weakly-unambiguous automaton of dimension $d$. Then there exists an equivalent generalized weakly-unambiguous automaton of dimension $d$, without $\varepsilon$-transitions, recognizing the same language.
\end{proposition}
\begin{proof}
	By applying successively Propositions~\ref{prop:PPAtoccfg}, \ref{prop:removalnullable}, \ref{prop:removalunit}, \ref{prop:greibach} and \ref{prop:ccfgtoppa}, we can obtain a generalized weakly-unambiguous \ppa{} $\BB'=(\{q_I'\},q_I',\{q_I'\}, \Sigma,\Gamma',\bot',\Delta',\mathcal S)$ such that $L(\BB')=L(\BB)\backslash\{\varepsilon\}$.
	
	If $\varepsilon\notin L(\BB)$, the proof is complete.
	
	If $\varepsilon\in L(\BB)$, we want to add it to $L(\BB')$ while preserving the weak-unambiguity to finish the proof. We consider a copy $\BB''=(\{q_0\},\{q_I',q_0\},\{q_I',q_0\}, \Sigma,\Gamma',\bot',\Delta'',\mathcal S')$ of the automaton $\BB'$ that we modify by simply adding a new state $q_0$, which is the initial state and also a final state. The set of transitions $\Delta''$ is a copy of $\Delta'$, and for any transition $q_I, A \xto{a, L}q\in \Delta'$, we also add to $\Delta''$ the transition $q_0, A \xto{a, \{\vect{1}\}+L}q$ where $\vect{1}$ is the vector with $1$ in each of its coordinates. Thus, $\BB$ can only be in state $q_0$ at the beginning of a run, before reading any letter, and $q_0$ is basically a copy of $q_I$ except that it is final, to accept the empty word. Furthermore, the automaton leaves the state $q_0$ by adding $1$ to every coordinate. We then just have to adjust $\mathcal S'=\{\vect{0}\}\cup (\mathcal S + \{\vect{1}\})$ to be able to accept the run on the empty word, and then accept the runs outside $q_0$, which are exactly runs of $\BB'$ with an extra $1$ on every coordinate. It is straightforward that $L(\BB'')=L(\BB')\cup\{\varepsilon\}=L(\BB)$ and that $\BB''$ is weakly-unambiguous.
\end{proof}

\subsection{From $NPCM(\cdot, 0, \cdot)$ to \PPA and reciprocally}

Once the $\varepsilon$-transitions have been removed from the weakly-unambiguous \PPA, turning one model into the other is syntactic, as in the case without a stack  of proposition \ref{prop:unambparbcm}.

\section{Parikh tree automata}
\label{app:parikh-tree}

Let $\EE=(E,ar)$ a tuple, where $F$ is a set of symbols, and $ar:E\to \NN$ is a function that associates to each symbol $f\in E$ a number $ar(f)$, called its arity. Symbols of arity $0$ are called leaves. We suppose that there is at least one leaf in $\EE$.

We will write $\EE$ under the form $\{f(ar(f))\ :\ f\in E\}$. For instance, $\EE=\{a(0), b(0), f(2)\}$ is a set of three symbols: one binary node $f$ and two leaves $a$ and $b$.

In the following we will abusively write $f\in\EE$ instead of $f\in E$ where $\EE=(E,ar)$.

\begin{definition}
	The set of trees built on $\EE$, denoted $Tree_\EE$, is defined as follows:
	\begin{itemize}
		\item every leaf $a\in \EE$ is a tree built on $\EE$
		\item if $f\in E$ and $t_1, \ldots, t_r$ are in $Tree_{\EE}$, where $r=ar(f)$, then $(f, t_1, \ldots, t_r)\in Tree_{\EE}$. 
	\end{itemize}
\end{definition}

\begin{definition}
	A \emph{top-down Parikh tree automata} of dimension $d$ over a family of trees $Tree_\EE$ is a tuple $(Q, \EE, I, \Delta, S)$ where
	\begin{itemize}
		\item $Q$ is the set of states
		\item $I\subseteq Q$ is the set of initial states
		\item $\Delta$ is the set of transitions. Each transition is of the form: $(f,q)\xrightarrow{\mathbf v} q_1,\ldots,q_r$, where $f\in E$, $ar(f)=r\geq 0$, $\mathbf v\in\NN^d$, $q, q_1, \ldots, q_r\in Q$.\\ Notice that in the case where $ar(f)=0$, the transitions are of the form $(f,q)\xrightarrow v$. 
		\item $S\subseteq\NN^d$ is the accepting semilinear set
	\end{itemize}
\end{definition}

\begin{definition}
	Let $t\in Tree_\EE$ and $\AA=(Q, \EE, I, \Delta, S)$ a tree automaton. The set of runs of $A$ over $t$ starting by $q\in Q$, denoted by $Runs_\AA(t,q)$, is defined as follows:
		\begin{itemize}
			\item if $ht(t)=0$, \textit{e.g} $t$ is a leaf, and $((t,q)\xrightarrow{\textbf v})\in \Delta$, then $((t,q),\mathbf v)\in Runs_\AA(t,q)$
			\item if $t=(f,t_1, \ldots, t_r)$ with $ar(t)=r> 0$ and $(f,q)\xrightarrow{\mathbf v}q_1,\ldots,q_r \in\Delta$, and $\pi_i\in Runs_\AA(t_i,q_i)$ for $i=1\ldots r$, then $\pi=((f,q), \mathbf{v}, \pi_1, \ldots, \pi_r)\in Runs_\AA(t,q)$
		\end{itemize}
		
In other words, a run of $\AA$ over $t$ consists in labelling the nodes of $t$ with states in $Q$ and vectors in $\NN^d$ according to the transition rules of the automaton.

A run of $A$ over $t$ is said to be \emph{accepting} if the root is labeled by a state in $I$, and if the sum of the vectors appearing in the run belongs to $S$.
\end{definition}

\begin{example}
	Let $\EE=\{a(0), b(0), f(2)\}$, $\AA=(Q, \EE, I, \Delta, S)$ with $Q=\{q_1, q_2\}$, $I=\{q_1\}$, $S=\{(n,n)\,:\,n\in\NN\}$, and $\Delta$ given by the rules:	\begin{align*}
		&(f,q_1)\xrightarrow{(0,0)}q_2, q_2& &(b,q_1)\xrightarrow{(0,1)} &(a,q_2)\xrightarrow{(1,0)}\\
		&(f,q_2)\xrightarrow{(0,0)}q_1, q_1& &(b,q_2)\xrightarrow{(0,1)} & 
	\end{align*}
This automaton recognizes the set of trees in $Tree_\EE$ having the same number of $a$'s and $b$'s, and verifying that each leaf $a$ is at an odd distance of the root.
\end{example}

\begin{definition}
	A Parikh tree automaton is said to be \emph{unambiguous} if every accepted \emph{tree} has only one accepted run.
\end{definition}

Exactly like for Parikh automata, we can define an extended class of Parikh tree automata, labeled with semilinear sets in $\NN^d$ instead of vectors in $\NN^d$.

\begin{proposition}
	The class of (unambiguous) Parikh tree automata is equivalent to the class of (unambiguous) extended Parikh tree automata.
\end{proposition}
\begin{proof}
	Simple adaptation of the proof for Parikh automata.
\end{proof}

In terms of generating series, unambiguous Parikh tree automata have holonomic series:

\begin{proposition}
	Let $\AA$ be an unambiguous Parikh tree automaton. Let's denote by $f(x)=\sum_{n\in\NN}a_nx^n$ the generating series of $\AA$ counting the number of accepted \emph{trees} with $n$ nodes. Then $f$ is holonomic.
\end{proposition}
\begin{proof}
	Same techniques of proofs as for every Parikh versions we have seen so far. 
\end{proof}

However, the generating series coming from unambiguous Parikh automata do not provide a larger class of holonomic series than the ones coming from unambiguous CCFG.

\begin{definition}[prefix description of a run]
	Let $\AA=(Q, \EE, I, \Delta, S)$ a tree automaton. Let us write $\Sigma=E\times Q\times V$, where $V$ is the set of vectors labelling the transitions of $\AA$. The prefix description of a run $\pi$ of $\AA$ is the string $Prefix(\pi)\in\Sigma^*$, defined as follows:
		\begin{itemize}
			\item if $\pi=((t,q),\mathbf v)$, where $t$ is a leaf, then $Prefix(\pi)=(t,q,\mathbf v)$
			\item if $\pi=((f,q), \mathbf{v}, \pi_1, \ldots, \pi_r)$, then $Prefix(\pi)=(f,q, \mathbf v)Prefix(\pi_1)\ldots Prefix(\pi_r)$
		\end{itemize}
		It is easy to see that from $Prefix(\pi)$ we can rebuild $\pi$, because every symbol in $\EE$ has a fixed arity, and $|\pi|=|Prefix(\pi)|$ (\textit{e.g.} the number of nodes is $\pi$ is equal to the length of $Prefix(\pi)$).
		
		We call $Prefix(\AA)$ the language of the prefixes of the accepting runs in $\AA$. 
\end{definition}

\begin{proposition}
	If $\AA$ is an unambiguous tree automaton, then the language $Prefix(\AA)$ is recognized by a weakly-unambiguous CCFG.  
\end{proposition}
\begin{proof}
	It is a slight modification of the construction seen before. From the automaton $\AA=(Q, \EE, I, \Delta, S)$, we build a constrained context-free grammar $G=(N, \Sigma, A, D, S)$, where $\Sigma=\{(f,q, \mathbf v)\,:\,f\in \EE, q\in Q, \mathbf v\in V\}$, $N=\{[f,q]\,:\,f\in \EE, q\in Q\}\cup \{A\}$, $A$ is the axiom, and $D$ is defined from $\Delta$ by:
	\begin{itemize}
		\item for every $q\in I$, for every $f\in \EE$, we add to $D$ the transition $A\xrightarrow{\mathbf 0}[f,q]$,
		\item for every transition of the form $(f,q)\xrightarrow{\mathbf v}q_1,\ldots,q_r \in\Delta$ with $r>0$ and for every $(f_1, \ldots, f_r)\in \EE^r$, we add to $D$ the transition $[f,q]\xrightarrow{\mathbf v}(f,q, \mathbf v)[f_1, q_1]\ldots[f_r, q_r]$,
		\item for every transition of the form $((a,q)\xrightarrow{\mathbf v})\in\Delta$ (in this case $a$ is a leaf in $\EE$), we add to $D$ the transition $[a,q]\xrightarrow{\mathbf v}(a,q,\mathbf v)$.
	\end{itemize}	
	By induction we can easily prove that the terminal derivation trees starting at $[f,q]$ simulate excatly the runs of $\AA$ over trees of root $f$, starting from state $q$; furthermore if the terminal derivation tree simulates a run $\pi$, then it recognizes the word $Prefix(\pi)$, and the sum of the vectors of the derivation is the same as the sum of the vectors in $\pi$.
	
	The weak-unambiguity of $G$ comes directly from the unambiguity of $\AA$  and the one-to-one correspondance between the runs in $\AA$ and the derivation trees of $G$.
	\end{proof}
	
\begin{remark}
	This result shows that in terms of family of generating series, unambiguous tree automata do not extend the class of holonomic series that we have seen so far.
\end{remark}

\section{Proofs omitted in Section~\ref{sec:non-ambiguity}}

\subsection{Properties of the language \leven}\label{sec:leven}

\leven is the language of sequences of encoded numbers having two consecutive equal values,  the first one being at an odd position. This language is deterministic context-free, hence its generating series is algebraic:  the corresponding deterministic pushdown automaton checks pair by pair if two consecutive numbers are equal, starting from position 1. It is also recognized by a non-deterministic PA, since it can guess the correct odd position and use a 2-dimension vector to check that the two consecutive values are equal. This automaton is weakly-ambiguous as the word $\underline{1}\underline{1}\underline{1}\underline{1} $ has two accepting runs. 

Recall that the $a$-subpath decomposition of $S$ consists of the following steps:
\begin{itemize}
    \item Follow the path until a state $q$ is met for the second time (or the path ends), forming an initial subpath of the form $w_1\sigma_1$ (or $w_1$), where $\sigma_1$ is an $a$-piece of origin $q$.
    \item Iteratively repeat on the suffix $(w_1\sigma_1)^{-1}S$ of $S$.
    \item Once completed, we obtain a decomposition of the form $S=w_1\sigma_1 w_2 \sigma_2 \cdots w_{f}\sigma_tw_{t+1}$, which we factorize into $S=w_1\sigma_1^{s_1} w_2 \sigma_2 \cdots w_{f}\sigma_f^{s_f}w_{f+1}$
    by grouping as much as possible consecutive identical $\sigma_i$'s when they are separated by empty subpaths: $\sigma_4\varepsilon\sigma_4\varepsilon\sigma_4 \mapsto \sigma_4^3$.
\end{itemize}
This decomposition is unique, as it is produced by a deterministic algorithm.

\begin{figure}[t]
\noindent\begin{tikzpicture}
\node[draw,circle] (q0) at (0,0) {$q_0$};
\node[draw,circle] (q1) at (2,0) {$q_1$};
\node[] (inv1) at (3,0) {};
\draw[->,decorate,decoration={snake,amplitude=.4mm,segment length=2mm,post length=1mm}] (q0) -- node[above]{$w_1$} (q1);
\draw[->,decorate,decoration={snake,amplitude=.4mm,segment length=2mm,post length=1mm}] (q1) to [in=130,out=40,looseness=12] node[above]{$s_1\times\sigma_1$} (q1);
\draw[->,dotted] (q1) -- (inv1);

\node[] (inv2) at (4,0) {};
\node[draw,circle] (qi) at (5,0) {$q$};
\node[] (inv3) at (6,0) {};
\draw[->,dotted] (inv2) -- (qi);
\draw[->,dotted] (qi) -- (inv3);
\draw[->,blue,decorate,decoration={snake,amplitude=.4mm,segment length=2mm,post length=1mm}] (qi) to [out=40,in=130,looseness=12] node[above]{$s_i\times\sigma_i$} (qi);

\node[] (inv4) at (7,0) {};
\node[draw,circle] (qj) at (8,0) {$q$};
\node[] (inv5) at (9,0) {};
\draw[->,dotted] (inv4) -- (qj);
\draw[->,dotted] (qj) -- (inv5);
\draw[->,blue,decorate,decoration={snake,amplitude=.4mm,segment length=2mm,post length=1mm}] (qj) to [in=130,out=40,looseness=12] node[above]{$s_j\times\sigma_j$} (qj);

\node[] (inv5) at (10,0) {};
\node[draw,circle] (qf2) at (11,0) {$q_f$};
\node[draw,circle] (qf3) at (13,0) {$\ \ $};

\draw[->,decorate,decoration={snake,amplitude=.4mm,segment length=2mm,post length=1mm}] (qf2) to [in=130,out=40,looseness=12] node[above]{$s_f\times\sigma_f$} (qf2);
\draw[->,decorate,decoration={snake,amplitude=.4mm,segment length=2mm,post length=1mm}] (qf2) -- node[above]{$w_{f+1}$} (qf3);
\draw[->,dotted] (inv5) -- (qf2);
\end{tikzpicture}

\noindent\begin{tikzpicture}
\node[draw,circle] (q0) at (0,0) {$q_0$};
\node[draw,circle] (q1) at (2,0) {$q_1$};
\node[] (inv1) at (3,0) {};
\draw[->,decorate,decoration={snake,amplitude=.4mm,segment length=2mm,post length=1mm}] (q0) -- node[above]{$w_1$} (q1);
\draw[->,decorate,decoration={snake,amplitude=.4mm,segment length=2mm,post length=1mm}] (q1) to [in=130,out=40,looseness=12] node[above]{$s_1\times\sigma_1$} (q1);
\draw[->,dotted] (q1) -- (inv1);

\node[] (inv2) at (4,0) {};
\node[draw,circle] (qi) at (5,0) {$q$};
\node[] (inv3) at (6,0) {};
\draw[->,dotted] (inv2) -- (qi);
\draw[->,dotted] (qi) -- (inv3);
\draw[->,red,decorate,decoration={snake,amplitude=.4mm,segment length=2mm,post length=1mm}] (qi) to [out=40,in=130,looseness=12] node[above]{$s_j\times\sigma_j+s_i\times\sigma_i$} (qi);

\node[] (inv4) at (7,0) {};
\node[draw,circle] (qj) at (8,0) {$q$};
\node[] (inv5) at (9,0) {};
\draw[->,dotted] (inv4) -- (qj);
\draw[->,dotted] (qj) -- (inv5);

\node[] (inv5) at (10,0) {};
\node[draw,circle] (qf2) at (11,0) {$q_f$};
\node[draw,circle] (qf3) at (13,0) {$\ \ $};

\draw[->,decorate,decoration={snake,amplitude=.4mm,segment length=2mm,post length=1mm}] (qf2) to [in=130,out=40,looseness=12] node[above]{$s_f\times\sigma_f$} (qf2);
\draw[->,decorate,decoration={snake,amplitude=.4mm,segment length=2mm,post length=1mm}] (qf2) -- node[above]{$w_{f+1}$} (qf3);
\draw[->,dotted] (inv5) -- (qf2);
\end{tikzpicture}
\caption{If in the decomposition of an $a$-subpath there are two $a$-pieces $\sigma_i$ and $\sigma_j$ having same origin, then we can alter the path in way it computes the same vector. It is therefore impossible if the automaton is weakly-unambiguous and $S$ is part of an accepting path.\label{fig:alamain}}
\end{figure}
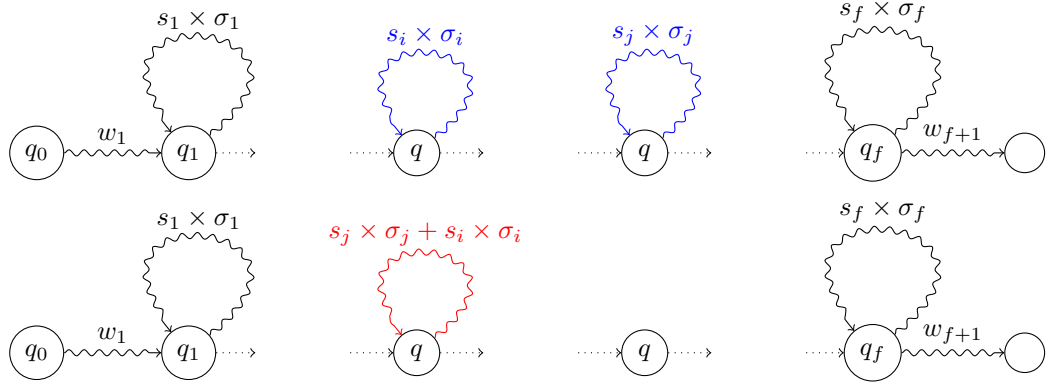

We now prove by contradiction that if $S$ has decomposition
$S=w_1\sigma_1^{s_1}w_2\cdots w_f\sigma_1^{s_f}w_{f+1}$ then  the $\sigma_i$'s have distinct origins: we assume that for $i< j$, the pieces $\sigma_i$ and $\sigma_j$ have same origin $q$.  Consider the $a$-path $S'$ defined by (it is depicted on Figure~\ref{fig:alamain}):
\[
S'=w_1\sigma_1^{s_1}w_2\cdots w_{i-1}\sigma_j^{s_j}\sigma_i^{s_i}w_i\cdots w_{j-1}w_{j} \cdots w_f\sigma_1^{s_f}w_{f+1}.
\]
The two paths are not equal: 
\begin{itemize}
    \item If $\sigma_i=\sigma_j$ then necessarily there is a non-empty $w_k$ or a $\sigma_k\neq \sigma_i$ between them, otherwise the path would contain $\sigma_i^{s_i+s_j}$ instead, by construction. 
    \item If $\sigma_i\neq\sigma_j$ then we do some $\sigma_j$-loops around $q$ before the $\sigma_j$-loops.
\end{itemize}
Moreover, the two $a$-paths $S$ and $S'$ have same length and compute the same vector, by commutativity of the vector addition. This is not possible as $S$ is part of an accepting run in a weakly-unambiguous automaton: all the $\sigma_i$'s have different extremities. This implies that there are at most $|Q_\AA|$ different $a$-pieces, which establishes the property $f\leq |Q_\AA|$.

To prove that there is a bounded number of distinct possible signatures for a given $\AA$, we just have to observe that there are finitely many possibilities for the $w_i$'s, as they are simple paths, and that there also are finitely many $a$-pieces.

\subsection{Decidability of the ambiguity of a PA}

\begin{proposition}\label{prop:decidability-weak-unambiguity}
Given a Parikh automaton $\AA$, it is decidable if it is weakly-ambiguous.
\end{proposition}
\begin{proof}
We prove that given a PA $\AA$, we can decide if it is weakly-ambiguous.

We use the same construction for the automaton of the intersection of two weakly-unambiguous automata: we starts by building the product-automaton $\CC$ recognizing $\AA_1 \cap \AA_2$, by taking $\AA_1=\AA_2=\AA$ and we modify it slightly by adding a new dimension. The construction is the same except for the transitions.

For every pair of transitions $(t_1, t_2)\in \delta_\AA\times\delta_\AA$, labeled by the same letter $a\in\Sigma$, with $t_1=(q_1,(a,\vect{v_1}),q'_1)$ and $t_2=(q_2,(a,\vect{v_2}),q'_2)$, then $((q_1,q_2),(a,(\vect{v_1}, \vect{v_2}, x_{t_1,t_2})),(q'_1,q'_2))\in\delta_{\CC}$, where $x_{t_1,t_2}=0$ if $t_1=t_2$, 1 otherwise. In other words, the last dimension counts by how many transitions the two simulated runs differ.

The semilinear set is the same as for the intersection, with the supplementary condition that the last value for the additional dimension must be bigger than $0$.

As $\AA$ is weakly-unambiguous if and only no accepted word has two accepted runs, it is direct to see that $\AA$ is weakly-unambiguous if and only the language recognized by $\CC$ is empty. And Emptiness is decidable for PA. 
\end{proof}

\subsection{Undecidability of inherent weak-ambiguity}

\begin{proposition}\label{prop:undecidability-inherent-weak-ambiguity}
Inherent weak-unambiguity is undecidable for PA.
\end{proposition}

The notion of weak-unambiguity, for languages recognized by Parikh automata, verifies all the criteria given in \cite{Greibach1968}, and thus is undecidable. We unroll here the key elements of the general proof from Greibach, for completeness.

We reduce the universality problem for languages recognized by PA, which is undecidable (see \cite{Klaedtke02}). Given an instance $L_1$ of the universality problem, we want to build an instance $L$ for the problem of  weak-unambiguity such that \[\hfill L_1=\Sigma_1^* \Leftrightarrow L \textnormal{\ is\ inherently\ weakly-unambiguous}.\hfill\]

The language $L$ is built as follows. Let's denote by $\Sigma_1$ the alphabet of $L_1$, $c$ a new symbol not in $\Sigma_1$, and finally let $L_2$ be an inherently weak-ambiguous language in PA (for instance the language $S$ of Section~\ref{sec:non-ambiguity} ) over an alphabet $\Sigma_2$ (we can have $\Sigma_1\cap\Sigma_2\not = \emptyset$).

We define then $L=L_1c\Sigma_2^*\ \cup\ \Sigma_1^*cL_2$. Then: $L_1=\Sigma_1^* \Leftrightarrow L \textnormal{\ is\ weakly-unambiguous}$. Indeed:
\begin{itemize}
\item[\textcolor{itemi}{\bm{$\Rightarrow$}}] If $L_1=\Sigma_1^*$ then $L=\Sigma_1^*c\Sigma_2^*$ is regular, so can be recognized weakly-unambiguously by a PA.
\item[\textcolor{itemi}{\bm{$\Leftarrow$}}] If $L$ is weakly-unambiguous, let's show that $L_1 = \Sigma_1^*$.

By contradiction, if $L_1\not = \Sigma_1^*$, then there would exist a word $y\in \Sigma_1^*$ such that $y\not\in L_1$.

As $yc\Sigma_2^*$ is regular,  $L\cap yc\Sigma_2^*$ is weakly-unambiguous since the class of weakly-unambiguous PA languages is stable by intersection.

As $y\not\in L_1$ and $c\not\in\Sigma_1$, $L\cap yc\Sigma_2^*=ycL_2$.

So $ycL_2$ is weakly-unambiguous, and by stability by left quotient with a word, $L_2=(yc)^{-1}(ycL_2)$ is weakly-unambiguous. This leads to a contradiction since $L_2$ in inherently weakly-ambiguous.
\end{itemize}

\begin{remark}
The proof needs at least 3 symbols. It can be proved that the problem is still undecidable for $2$-letters languages by encoding every letter with a sequence of 2 letters.
\end{remark}

\begin{remark}
One of the main point in the proof is the stability of weakly-unambiguous automaton under left quotient with a word. Notice that although this property is true (by adding $\varepsilon$-transitions in an automaton without $\varepsilon$-transitions), it cannot be generalized by replacing $w$ by a regular language, since the language $L=\{a^mba^nbva^nw\ :\ v, w\in\{a,b\}^*, |v|=m\}$ is weakly-unambiguously recognizable, whereas $(a^*b)^{-1}L=S$ is not.
\end{remark}

\section{Toolbox, and general bounds}\label{sec:toolbox}

\subsection{Notations, Hadamard's bound, Cramer's rule}
When working on multivariate polynomials in $\mathbb{Q}[x_1,\ldots,x_n]$, we
denote by $\boldsymbol\alpha\in\mathbb{N}^n$ the \intro{multi-index}
$({\alpha}_1,\ldots,{\alpha}_n)$ and by $x^{\vect\alpha}$ the monomial $x^{\boldsymbol
\alpha} = \prod_{i=1}^n x_i^{\alpha_i}$. The \intro{(total) degree} of
$x^{\boldsymbol \alpha}$ is defined as $\sum_{i=1}^n \alpha_i$, but in the sequel we
will mainly work with the \intro{maxdegree} $\mdeg(x^{\boldsymbol\alpha})=\max_{i=1}^n
\alpha_i$.

A nonzero polynomial $P$ in $\mathbb{Q}[x_1,\ldots,x_n]$ is a finite $\mathbb{Q}$-linear combination of monomials
\[
P=\sum_{\boldsymbol\alpha\in\Gamma_P}\lambda_{\boldsymbol \alpha}x^{\boldsymbol\alpha},
\]
where $\Gamma_P$ is a finite set of multi-indices, called the monomial support of~$P$. The \intro{degree} $\deg(P)$ of $P$ is the maximum degree of its monomials:
\[
\deg(P)  = \max\left\{\deg(x^{\boldsymbol\alpha}):\ \boldsymbol\alpha\in\Gamma_P\text{ and }\lambda_{\boldsymbol\alpha}\neq 0 \right\}.
\]
Similarly, the maxdegree of a polynomial is the maximum maxdegree of its monomials:
\[
\mdeg(P)  = \max\left\{\mdeg(x^{\boldsymbol\alpha}):\ \boldsymbol\alpha\in\Gamma_P\text{ and }\lambda_{\boldsymbol\alpha}\neq 0 \right\}.
\]
We will use the following classical norms for polynomials of  $\mathbb{Q}[x_1,\ldots,x_n]$:
\[
\|P\|_1  = \sum_{\boldsymbol\alpha\in\Gamma_P} |\lambda_{\boldsymbol\alpha}|, \; \text{ and } \;
\|P\|_\infty  = \max_{\boldsymbol\alpha\in\Gamma_P} |\lambda_{\boldsymbol\alpha}|.   
\]
Since there are $(m+1)^n$ monomials of maxdegree at most $m$, on $n$ variables, we get directly:
\begin{equation}\label{eq:norm 1 infty}
\|P\|_1 \leq (\mdeg(P)+1)^n \|P\|_\infty  .
\end{equation}
\begin{lemma}[product]\label{lm:multiplicationpolynomes}
	If $P$ and $Q$ are polynomials in $\mathbb{Q}[x_1,\ldots,x_n]$, then 
\[\|PQ\|_\infty\leq (m+1)^n\|P\|_\infty\|Q\|_\infty,\text{ where }m=\min(\mdeg(P), \mdeg(Q)).\]
\end{lemma}
\begin{proof}
	W.l.o.g., we can assume that $\mdeg(Q)\leq\mdeg(P)$ and hence $m=\mdeg(Q)$. For $P=\sum \lambda_{\boldsymbol\alpha}x^{\boldsymbol\alpha}$ and $Q=\sum \mu_{\boldsymbol\beta}x^{\boldsymbol\beta}$, we have $PQ=\sum c_{\boldsymbol \gamma}x^{\boldsymbol \gamma}$ with $|c_{\boldsymbol \gamma}|=|\sum_{\boldsymbol\alpha+\boldsymbol\beta=\boldsymbol\gamma} \lambda_{\boldsymbol\alpha}\mu_{\boldsymbol\beta}|\leq \|P\|_\infty\|Q\|_1\leq (m+1)^n\|P\|_\infty\|Q\|_\infty$.
\end{proof}

\begin{lemma}\label{lem:bound poly}
Let $P\neq 0$ be a polynomial in $\mathbb{Z}[x]$. 
Then $P(n) \neq 0$ for all $n \geq \|P\|_\infty+1$.
\end{lemma}
\begin{proof}
Let $n \geq \|P\|_\infty+1$. We write $P(x)=\sum_{k=0}^d a_k x^k$. Assume toward a contradiction that $P(n)=0$. In particular, we have:
\[ n^d \leq  \left|a_d n^d\right| = \left|\sum_{k=0}^{d-1} a_k n^k\right| \leq \|P\|_\infty \frac{n^d-1}{n-1} \leq n^d-1  \]	
which brings the contradiction.
\end{proof}

Let $A$ be a $\mathbb{C}$-valued $p\times p$ matrix. The classical \intro{Hadamard bound} states that:
\begin{equation}\label{eqq:Hadamard bound}
|\det(A)|\leq \prod_{i=1}^p \sqrt{\sum_{j=1}^p|A_{i,j}|^2},
\end{equation}
that is, the modulus of the determinant is bounded from above by the product of the length of the column vectors of $A$.

We adapt to several variables a result of Goldstein and Graham~\cite{Goldstein}, which gives bounds for the coefficients of a determinant of polynomials, generalizing the classical Hadamard bound. Note that similar bounds appear in \cite[Lemma 3.1]{Emiris} and in a different context, in~\cite[Eq.~(21)]{Brown} (without proof).

\begin{lemma}[determinant]\label{lm:hadamardbound}
	Let $A$ be a $p\times p$ matrix whose coefficients are polynomials in $\mathbb{Q}[x_1,\ldots,x_n]$. \[
	\| \det(A)\|_\infty\leq R_1^p p^{p/2} \;\textrm{where $R_1=\max_{i,j} \|A_{i,j}\|_1$.}
	\]
In particular, using the upper bound given by Equation~\eqref{eq:norm 1 infty}, it also satisfies:
\[
	\|\det(A)\|_\infty\leq R_\infty^pp^{p/2}(M+1)^{np} \;\textrm{where $R_\infty=\max_{i,j} \| A_{i,j} \|_\infty$ and $M = \max_{i,j}  \mdeg(A_{i,j})$}.
	\]
\end{lemma}	
\begin{proof}
Let $f$  be the map from $\mathbb{R}^n$ to $\mathbb{C}$ defined using
the polynomial $\det(A(x_1, \ldots, x_n))$  as follows:
\[
f(t_1,\ldots, t_n) = \det\left(A(e^{i t_1},\ldots,e^{i t_n})\right).
\]
For any given $t_1,\ldots, t_n\in\mathbb{R}$, $f(t_1,\ldots, t_n)$ is defined as a complex-valued determinant, so we can apply Hadamard bound, which yields that
\[
|f(t_1,\ldots,t_n)| \leq \prod_{k=1}^p\sqrt{\sum_{j=1}^p|A_{k,j}(e^{it_1}, \ldots, e^{it_n})|^2}.
\]
But when setting $(x_1,\ldots,x_n)\mapsto(e^{it_1}, \ldots, e^{it_n})$, every monomial is changed into a complex number of modulus~$1$. Hence $|A_{k,j}(e^{it_1}, \ldots, e^{it_n})|\leq \|A_{k,j}\|_1$, for all $k,j\in\{1,\ldots,p\}$. This yields:
\[
|f(t_1,\ldots,t_n)| \leq \prod_{k=1}^p\sqrt{\sum_{j=1}^p\|A_{k,j}\|_1^2}
\leq 
\prod_{k=1}^p\sqrt{\sum_{j=1}^p R_1^2}.
\]
Therefore, we get an upper bound for $|f(t_1,\ldots,t_n)|^2$:
\begin{equation}\label{eq:upper f2}
|f(t_1,\ldots,t_n)|^2 \leq p^p R_1^{2p}.    
\end{equation}
We now consider the real number
\[
I(A) = \frac{1}{(2\pi)^n}\int_{0}^{2\pi}\cdots\int_{0}^{2\pi}
|f(t_1,\ldots,t_n)|^2dt_1\cdots dt_n.
\]
By Eq.~\eqref{eq:upper f2}, we have $I(A) \leq p^{p}R_1^{2p}$. We will now prove that
$\|\det(A)\|_\infty^2 \leq I(A)$.

Observe that $\det(A)$ is a polynomial, and can therefore be written $\det(A)=\sum_{\boldsymbol\alpha\in\Gamma}\lambda_{\boldsymbol\alpha}x^{\boldsymbol\alpha}$, for a finite set of multi-indices $\Gamma$. Thus, when setting $(x_1,\ldots,x_n)\mapsto(e^{it_1}, \ldots, e^{it_n})$, using the fact that $|z|^2=z\overline{z}$ (where $\overline{z}$ is the complex conjugate of $z$), we obtain:
\[
|f(t_1,\ldots,t_n)|^2 =
\left(\sum_{\boldsymbol\alpha\in\Gamma}\lambda_{\boldsymbol\alpha}\exp(i\alpha_1t_1+\ldots+i\alpha_nt_n)\right)\left(\sum_{\boldsymbol\beta\in\Gamma}\overline{\lambda_{\boldsymbol\beta}}\exp(-i\beta_1t_1-\ldots-i\beta_nt_n)\right).
\]
Hence,
\[
I(A) = \frac{1}{(2\pi)^n}\sum_{\boldsymbol\alpha,\boldsymbol\beta\in\Gamma}
\lambda_{\boldsymbol\alpha}\overline{\lambda_{\boldsymbol\beta}}
\int_{t_1=0}^{2\pi}\cdots\int_{t_n=0}^{2\pi}
e^{i(\alpha_1-\beta_1)t_1}\cdots e^{i(\alpha_n-\beta_n)t_n} dt_1\cdots dt_n.
\]
The integral part is null if $\boldsymbol\alpha\neq\boldsymbol\beta$ and it is equal to $(2\pi)^n$ if $\boldsymbol\alpha=\boldsymbol\beta$. Hence, 
$I(A) = \sum_{\boldsymbol\alpha\in\Gamma}|\lambda_{\boldsymbol\alpha}|^2 \geq \|\det(A(x_1, \ldots, x_n))\|_\infty^2$, concluding the proof.
\end{proof}

Recall that a \intro{minor} of an $m\times n$ matrix with $m<n$ is the determinant of a matrix obtained by selecting $m$ columns of $A$, in the same order as in $A$. We have the following classical results:

\begin{proposition}[Cramer's rule]
\label{prop:cramer-rule}
Let $\mathbb K$ be any field. 
\begin{enumerate}
	\item Let $A$ be an invertible $n\times n$ matrix with coefficients in $\mathbb K$, $\boldsymbol b$ and $\vect{v} \in \mathbb K^n$ such that $A\boldsymbol v=\boldsymbol b$.
	For all $i \in [1,n]$, ${v}_i$ is equal to $\frac{\det(A_i)}{\det(A)}$ where $A_i$ designates the matrix obtained by replacing the $i$-th column of $A$ by the vector $\vect{b}$.
	\item Let $A$ be an $m \times n$ matrix with coefficients in $\mathbb K$ where $m<n$. There exists a vector $\boldsymbol v \in \mathbb{K}^m\setminus\{\vect 0\}$ whose coordinates are  minors or opposite of minors of $A$ such that $A\boldsymbol v=\boldsymbol 0$.
\end{enumerate}
\end{proposition}

\begin{proof}
See \cite{cramersrule} and \cite{Kauers2014}.
\end{proof}

\subsection{Bounds for the generating series of a vector automaton}
\newcommand{\VV}{\mathcal{V}}

We introduce a very simple model of non-deterministic finite state automaton, called a vector automaton, whose transitions are labeled by vectors in $\mathbb{N}^d \setminus \{\vect{0}\}$. A run of a vector automaton is labeled  by the sum of the vectors appearing in its transitions. In particular, a vector automaton accepts a set of vectors in $\mathbb{N}^d$. 

Formally, a \intro{vector automaton} $\AA$ over $\mathbb{N}^d$ is a tuple $(Q_\AA,q_I,F,\Delta)$ where $Q_\AA$ is a finite set of states, $q_I \in Q_\AA$ is the initial state, $F \subseteq Q_\AA$ is the set of final states and $\Delta$ is the set of transitions of the form $(p,\vect{v},q)$ with $p$ and $q$ in $Q_\AA$ and $\vect{v} \in \NN^d\setminus \{\vect{0}\}$. The runs of $\AA$ are defined in the usual way. We denote $\| \AA \|_\infty \eqdef \max_{(p,\vect{v},q) \in \Delta}  \| \vect v \|_\infty$ and $d^{out}_\AA$ the maximum out-degree of $\AA$.

The (multivariate) generating series of $\AA$, denoted $A(x_1,\ldots,x_d)$, is defined by:
\[
A(x_1,\ldots,x_d) \eqdef \sum_{\vect{v} \in \mathbb{N}^d} a(\vect{v})x^{{v}_1}_1\cdots x^{{v}_d}_d,
\]
where $a(\vect{v})$ counts the number of runs of $\AA$ starting in the initial state, ending in a final state and labeled by $\vect{v}$. As we do not allow transitions labeled by $\vect{0}$, the number of such runs is finite and the series is well-defined.

\begin{lemma}\label{lm:automatontorational}
Let $\AA$ be a vector automaton over $\NN^d$ with state set $Q_\AA$. Its multivariate generating series  $A(x_1,\ldots,x_d)$ is equal to $\frac{P}{Q}$, where $P$ and $Q$ are two polynomials in $\mathbb{Z}[x_1,\ldots,x_d]$ such that:
\[
    \mdeg(P), \mdeg(Q)\leq |Q_\AA|\cdot \|\AA\|_\infty,\text{ and }
    \|P\|_\infty, \|Q\|_\infty \leq{(1+d^{out}_{\AA})}^{|Q_{\AA}|}{|Q_{\AA}|}^{|Q_{\AA}|/2}.
\]
\end{lemma}
\begin{proof}
Let $\AA=(Q_{\AA},q_I,F,\Delta)$ be a vector automaton. For every state $q\in Q_\AA$, we consider the series 
\[q(x_1, \ldots, x_d)=\sum_{\vect{v} \in \mathbb{N}^d}a_q(\vect{v})x_1^{{v}_1}\ldots x_d^{{v}_d},\] 
where $a_q(\vect{v})$ counts the number of runs from the state $q$ to a final state, labeled by the vector $\vect{v}$. In particular, the multivariate generating series $A(x_1,\ldots,x_d)$ of $\AA$ is equal to $q_I(x_1,\ldots,x_d)$.

Classically, for any state $q$, we have the relation:
	\begin{align*}q(x_1, \ldots, x_d)= &\llbracket q\in F\rrbracket 
	+ \sum_{(q,\vect{v},q')\in \Delta}x_1^{{v}_1}\cdots x_d^{{v}_d}\cdot q'(x_1, \ldots, x_d),
	\end{align*}
where $\llbracket P\rrbracket = 1$ if $P$ holds and $\llbracket P\rrbracket = 0$ otherwise.

As a consequence, the vector of generating series $\vect q=(q(x_1, \ldots, x_d))_{q\in Q_\AA}$ satisfies the equation:
\[
\vect{q}=M\vect{q}+\vect{f},
\]
where  $\vect{f}=(\llbracket q\in F\rrbracket)_{q\in Q_\AA}$, and $M$ is the $|Q_\AA|\times |Q_\AA|$ matrix of polynomials defined by $M_{p,q}(x_1,\ldots,x_d) = \sum_{(p,\vect{v},q) \in \Delta} x_1^{\vect{v}_1} \cdots x_d^{\vect{v}_d}$ for all $p$ and $q$ in $Q_\AA$. In particular, for all $p$ and $q \in Q_\AA$,  $\|M_{p,q}\|_1 \leq d^{out}_{\AA}$ and $M$ is the null matrix if we set $x_1=\ldots=x_d = 0$.

The vector $\vect{q}$ satisfies the equation $(Id-M)\vect{q} = \vect{f}$. The matrix $Id-M$ is invertible on the field of rational fractions $\mathbb{Q}(x_1,\ldots,x_d)$. Indeed $\det(Id-M)(x_1,\ldots,x_d)$ is not the null polynomial as $\det(Id-M)(0,\ldots,0)=\det(I)=1$. Therefore, by Cramer's rule (see Proposition~\ref{prop:cramer-rule}), we have 
\[
A(x_1,\ldots,x_d)=q_I(x_1,\ldots,x_d) = \frac{P(x_1,\ldots,x_d)}{Q(x_1,\ldots,x_d)},
\]
with $P(x_1,\ldots,x_d)=\det(M')$, where $M'$ is the matrix obtained by replacing the column $q_I$ of $Id-M$ by the vector $\vect{f}$, and with $Q(x_1,\ldots,x_d)=\det(Id-M)$.

For all $p,q \in Q_\AA$, we have $\|(Id-M)_{p,q}\|_1 \leq 1 + 
d_\AA^{out}$ and $\mdeg(Id-M)_{p,q} \leq \|\AA\|_\infty$. As the vector $\vect{f}$ substituted in $Id-M$ to create $M'$ only contains $0$'s and $1$'s, the same bounds also hold for $M'$. Hence, using Lemma~\ref{lm:hadamardbound}, we can bound the $\infty$-norm of both $P$ and $Q$ by $ (1 + 
d_\AA^{out})^{|Q_\AA|}|Q_\AA|^{|Q_\AA|/2}$. 

By definition of the determinant,  $P$ and $Q$ are linear combinations of products of $|Q_\AA|$ polynomials, each of maxdegree at most $\|\AA\|_\infty$. As a consequence $\mdeg(P)$ and $\mdeg(Q)$ are at most $|Q_\AA|\cdot \|\AA\|_\infty$.
\end{proof}

\subsection{Bounds for the transformation of a differential equation into a recurrence equation}

In one variable, the coefficient sequence of a holonomic series $H(x)$ satisfies a linear recurrence relation with polynomial coefficients. In this section, we provide a bound on the order of the recurrence relation and on the degree and $\infty$-norm of the leading polynomial in terms of the similar quantities of a differential equation satisfied by $H(x)$.

\begin{restatable}{proposition}{boundsdiffeqtorec}\label{prop:boundsdiffeqtorec}
	Let $H(x)=\sum u_nx^n$ be a series verifying an equation of the form:
	\[q_r(x)\partial_r H(x)+\ldots + q_0(x)H(x)=0\]
Then $u_n$ verifies a recurrence of the form: 
	\[\sum_{k=-s}^{S} t_k(n)u_{n+k}=0, \, \text{for $n\geq s$, with $t_S\not= 0$,}\]
	and $0\leq s\leq \max_{i}(\deg(q_i))$, $0\leq S\leq r$, $\deg(t_S)\leq r$, $\|t_S\|_\infty\leq \max_i(\|q_i\|_\infty)\cdot r^{r+1}$.
\end{restatable}

\begin{proof}
Without loss of generality, we can assume that for some $k \in [0,r]$, $q_k(0)\neq 0$. If this is not the case, we can divide the equation by the largest $x^m$ which divides all $q_k$'s. This operation does not impact the $\infty$-norm of the polynomials and decreases their degree.

Let $D = \max_i \deg(q_i)$. We can rewrite the equation satisfied by $H(x)$ as:	
\[\sum_{k=0}^r\sum_{k'=0}^D a_{k,k'}x^{k'}\partial_x^kH(x)=0.\]	

This equation translates into the following recurrence relations on the coefficient of $H(x)$ which holds for $n \geq D$
\[
\sum_{k=0}^r\sum_{k'=0}^D a_{k,k'} (n-k'+1)(n-k'+2)\ldots (n-k'+k)u_{n-k'+k}=0.\]
By taking $j=k-k'$, we can rewrite the relation as:
	\[\sum_{j=-D}^{r} t_j(n)u_{n+j}=0 \;\text{with}\; t_j(n)=\sum_{k=\max(0,j)}^{r} a_{k,k-j} \prod_{\ell = 1}^{k}
(n+j-\ell+1)
	\]

Some of the polynomials $t_j$ are possibly zero. Therefore,  we need  to characterize the leading non-zero polynomial in the recurrence.

Consider $S \eqdef \max \{ k-k' \mid k \in [0,r], k' \in [0,D] \;\text{and}\; a_{k,k'} \neq 0 \}$. The assumption that  $q_{k_0}(0)\neq 0$ for some $k_0$ translates to $a_{k_0,0}\neq 0$ for some $k_0 \geq 0$. This implies that $S \geq k_0 \geq 0$.

We now prove that $t_S(n)$ is the leading polynomial in the recurrence relation. By maximality of $S$,  we have $t_j(n)=0$ for $S<j\leq r$. Moreover by definition, $t_S(n) = \sum_{k=S}^{r} r_k(n)$ with $r_k(n)= a_{k,k-S} \prod_{\ell = 1}^{k}
(n+S-\ell+1)$ for all $k \in [S,r]$. By definition of $S$, at least one on the $r_k$ is non-zero. As $\deg(r_k)=k$ for all $k \in [S,r]$, we have that $t_S(n)\neq 0$ and  $t_S(n)$ has degree at most $r$.

It only remains to bound the $\infty$-norm of $t_S$. Notice that when developing the product $\prod_{\ell = 1}^{k}(n+S-\ell+1)=\sum_{m=0}^kb_mn^m$, then for $m\in [k]$, $b_m=\sum_{I\subseteq [k],|I|=k-m} \prod_{\ell\in I}(S-\ell+1)$, so that $|b_m|\leq \binom{k}{m}r^{k-m}\leq k^mr^{k-m}\leq r^r$. Hence:
\[
\begin{array}{lcl}
\|t_S\|_\infty &\leq & \sum_{k=S}^{r} |a_{k,k-S}| \cdot  \| \prod_{\ell = 1}^{k} (n+S-\ell+1) \|_\infty \\
 &\leq & \sum_{k=S}^{r} |a_{k,k-S}| \cdot r^r \\
 &\leq & \max_i(\|q_i\|_\infty) \cdot r^{r+1}. \\
 \end{array}
\]
\end{proof}

\subsection{Bound for the sum of two univariate holonomic series}
\newcommand{\hauteur}{\mathrm{ht}}

In this section, we specialize the bounds of \cite[Theorem~2]{Kauers2014} for the sum of holonomic series to our notations.

\begin{proposition}[\cite{Kauers2014}]
\label{prop:bound-holonomic-sum}
Given two holonomic series $A(x)$ and $B(x)$ satisfying non-trivial differential equations of the form:
\[
\begin{array}{lcr}
p_r^A(x) \partial^r A(x) + \cdots + p_0^A(x) A(x) = 0  & \textrm{and} & 
p_r^B(x) \partial^r B(x) + \cdots + p_0^B(x) B(x) = 0 \\
\end{array},
\]
let $D=\max_{i\in[r]} (\deg(p_i^A),\deg(p_i^B))$ and $S_\infty=\max_{i\in[r]} (\|p_i^A)\|_\infty,\|p_i^B\|_\infty)$. Then, $C(x)=A(x)+B(x)$ satisfies a non-trivial differential equation of the form:
\[
q_{2r}(x) \partial^{2r} C(x) + \cdots + q_0(x) C(x) = 0,
\]
with for all $i \in [2r]$, $\deg(q_i) \leq 2(r+1)D$ and $\log(\|q_i\|) \leq O(r^2(1+\log(r)+\log(D)+\log(S_\infty))$.
\end{proposition}
\begin{proof}
By \cite[Theorem 2]{Kauers2014}, we immediately have that $C(x)$ satisfies a non-trivial differential equation of the form: \[
q_{2r} \partial^{2r} C(x) + \cdots + q_0 C(x) = 0 
\]
with for all $i \in [2r]$, $\deg(q_i) \leq 2(r+1)D$.
Furthermore for all $i \in [2r]$,
\[
\hauteur(\|q_i\|_\infty) \leq \hauteur(2r)+\hauteur((2r+1)!)+(2r+1)\hauteur(d)+(2r+2)((2r)(\hauteur(1)+\hauteur(d))+\hauteur(S_\infty))
\]
with $\hauteur(x)=\log(1+|x|)$.

As for all $x \geq 1$, $\log(1+x) \leq 1 + \log(x)$, we have that:
\[
\log(\|q_i\|_\infty)) \leq 8(r+1)^2  + (2r+2)\log(r) + (4r^2+6r+1)\log(D) + (2r+2)\log(S_\infty).
\]
\end{proof}

\subsection{Bounds for specialization}

In this section, we consider a holonomic multivariate series $H(x,y_1,\ldots,y_d)$ such that $H(x)=H(x,1,\ldots,1)$ is well-defined. As stated in Proposition~\ref{prop:specialisation_holonome}, $H(x)$ is holonomic. We derive bounds on the differential equation satisfied by $H(x)$, given an equation satisfied by $H(x, \vect{y})$.

Before proceeding, we need a technical lemma.

\begin{lemma}
\label{lem:polynomial-div-instance}
Let $P_0, P_1$, \ldots, $P_d$ be a sequence of polynomials such that for all $\ell \geq 0$, $P_\ell$ belongs to $\mathbb{Q}[y_\ell,\ldots,y_{d}]$ and for all $\ell \geq 0$, $P_{\ell+1}(y_{\ell+1},\ldots,y_{d})$ is obtained by dividing 
$P_{\ell}(y_\ell,\ldots,y_d)$ by $(y_{\ell}-1)^{k_{\ell}}$ for some $k_{\ell} \geq 0$ and setting $y_{\ell}$ to $1$.

We have $\deg(P_d) \leq \deg_m(P_0)$ and $\| P_d \|_{\infty} \leq (\deg_m(P_0)+1)^d \; 2^{\deg(P_0)} \| P_0 \|_{\infty}$.
\end{lemma}

\begin{proof}
The polynomial $P_0$ can be written as the sum of its monomials $\sum_{\vect{\alpha} \in A_0} a_{\vect{\alpha}} y_0^{\alpha_0}\cdots y_d^{\alpha_d}$ for some set $A_0 \subset \mathbb{N}^{d+1}$.

We will prove by induction on $\ell$ that for all $\ell \in [0,d]$, \[
P_\ell(y_\ell,\ldots,y_{d}) = \sum_{\vect{\alpha} \in A_\ell} a_{\vect{\alpha}} \binom{\alpha_0}{k_0} \cdots \binom{\alpha_{\ell-1}}{k_{\ell-1}} y_\ell^{\alpha_\ell}\cdots y_d^{\alpha_d}.
\]
where $A_\ell = \{ \vect{\alpha} \in A_0 : \alpha_{\ell'} \geq k_{\ell'}\;\textrm{for $\ell' \in [0,\ell-1]$} \}$.

The property trivially holds for $\ell=0$. For the induction step, it is enough to remark that:
\[
P_{\ell+1}(y_{\ell+1},\ldots,y_{d}) = \frac{1}{k_\ell!} (\partial_{y_{\ell}}^{k_\ell} P_{\ell})(1,y_{\ell+1},\ldots,y_{d}).
\]

In particular, we have the following characterization of $P_d$ in terms of $P_0$.
\[
P_d(y_{d}) = \sum_{\vect{\alpha} \in A_d} a_{\vect{\alpha}} \binom{\alpha_0}{k_0} \cdots \binom{\alpha_{d-1}}{k_{d-1}}  y_d^{\alpha_d}.
\]
Clearly, the degree $\deg(P_d)$ is bounded by the maxdegree of $P_0$. For all $k \in [0,\deg(P_d)]$, consider the coefficient of $c_k$ of $y_d^k$ in $P_d$. We have:
\[
\begin{array}{lcl}
|c_k| & = &  |\sum_{\vect{\alpha} \in A_d, \alpha_d=k} a_{\vect{\alpha}} \binom{\alpha_0}{k_0} \cdots \binom{\alpha_{d-1}}{k_{d-1}}| \\
 & \leq & \sum_{\vect{\alpha} \in A_d, \alpha_d=k}  |a_{\vect{\alpha}}|\cdot 2^{\alpha_0+\cdots+\alpha_{d-1}} \\
& \leq & \|P_0\|_\infty \cdot 2^{\deg(P_0)} \cdot ( \sum_{\vect{\alpha} \in A_d, \alpha_d=k} 1 )\\
& \leq & \|P_0\|_\infty \cdot 2^{\deg(P_0)} \cdot (\mdeg(P_0)+1)^d.
\end{array}
\]
This implies the announced bound on $\|P_d\|_\infty$.

\end{proof}

\begin{restatable}{proposition}{boundsspecialization}\label{prop:boundsspecialization}
	Let $H(x,y_1, \ldots, y_d)$ a multivariate series satisfying a non-trivial equation of the form:
	 \[p_r(x,y_1, \ldots, y_d)\partial_{x}^r H(x,y_1, \ldots, y_d)+\ldots+ p_0(x,y_1, \ldots, y_d)H(x,y_1, \ldots, y_d)=0.\] 
If $H(x)=H(x,1,\ldots,1)$ is well defined, then $H(x)$ satisfies a non-trivial equation of the form:
	\[q_s(x)\partial_x^s H(x)+\cdots+ q_0(x)H(x)=0\]
with $s \leq r$ and for $i \in [0,s]$, $\deg(q_i) \leq \mdeg(p_i)$ and $\|q_i\|_\infty \leq \|p_i\|_\infty   (\deg_m(p_i)+1)^d \; 2^{\deg(p_i)}$.
\end{restatable}
	
\begin{proof}
For all $\ell \in [0,d]$, we inductively define a series $H_\ell$ in the variables $x,y_1,\ldots,y_{d-\ell}$ by taking
$H_0(x,y_1,\ldots,y_d)=H(x,y_1,\ldots,y_d)$ and for $\ell \in [1,d]$, $H_\ell$ is obtained by setting the variable $y_{d-\ell+1}$ to 1 in $H_{\ell-1}$ (i.e., $H_{\ell}(x,y_1,\ldots,y_{d-\ell})=H_{\ell-1}(x,y_1,\ldots,y_{d-\ell},1)$).
In particular, $H_d(x)$ is equal to $H(x)$.

We prove by induction that for all $\ell \in [0,d]$, $H_\ell$ satisfies a non-trivial equation of the form:
\[p_r^{[\ell]}(x,y_1, \ldots, y_{d-\ell})\partial_{x}^r H_\ell(x,y_1, \ldots, y_{d-\ell})+\ldots+ p^{[\ell]}_0(x,y_1, \ldots, y_{d-\ell})H_\ell(x,y_1, \ldots, y_{d-\ell})=0\]
and for all $\ell \in [1,d]$, there exists $k_\ell \geq 0$ such that for all $j \in [0,r]$, $p_j^{[\ell]}$ is obtained by dividing $p_k^{[\ell-1]}$ by $(y_{d-\ell+1}-1)^{k_\ell}$ and then setting $y_{d-\ell+1}$ to $1$.

The announced bounds on the polynomial $q_i$ are then obtained by applying Lemma~\ref{lem:polynomial-div-instance} to the sequence of polynomial $p_i^{[0]}=p_i, p_i^{[1]}, \ldots, p_i^{[d]}=q_i$.

The base case is immediate. Assume that the property holds for some $\ell$. That is to say that $H_\ell$ satisfies a non-trivial equation:
\[p_r^{[\ell]}(x,y_1, \ldots, y_{d-\ell})\partial_{x}^r H_\ell(x,y_1, \ldots, y_{d-\ell})+\ldots+ p^{[\ell]}_0(x,y_1, \ldots, y_{d-\ell})H_\ell(x,y_1, \ldots, y_{d-\ell})=0\]

Let $k_{\ell+1} \geq 0$ be the greatest $k_{\ell+1}$ such for all $j \in [0,r]$, $p_j^{[\ell]}$ can be written as $(y_{d-\ell}-1)^{k_{\ell+1}} \cdot s_j^{[\ell]}(x,y_1,\ldots,y_{d-\ell})$. Hence the series $H_\ell$ satisfies the non-trivial equation:
\[s_r^{[\ell]}(x,y_1, \ldots, y_{d-\ell})\partial_{x}^r H_\ell(x,y_1, \ldots, y_{d-\ell})+\ldots+ s^{[\ell]}_0(x,y_1, \ldots, y_{d-\ell})H_\ell(x,y_1, \ldots, y_{d-\ell})=0\]

The series $H_{\ell+1}$ satisfies the equation:
\[s_r^{[\ell]}(x,y_1, \ldots, y_{d-\ell-1},1)\partial_{x}^r H_{\ell+1}(x,y_1, \ldots, y_{d-\ell-1})+\ldots+ s^{[\ell]}_0(x,y_1, \ldots, y_{d-\ell-1},1)H_{\ell+1}(x,y_1, \ldots, y_{d-\ell-1})=0\]
This equation is non-trivial as otherwise the polynomials
$s_j^{[\ell]}$ would all be divisible by $(y_{d-\ell}-1)$ which would contradict the maximality of $k_{\ell+1}$.
Hence we can take $p_j^{[\ell+1]}(x,y_1, \ldots, y_{d-\ell-1})=s_j^{[\ell]}(x,y_1, \ldots, y_{d-\ell-1},1)$ for all $j \in [0,r]$.
\end{proof}

\subsection{Bounds for the Hadamard product of two rational series}\label{sec:hadamard}

In this section, we are interested in giving upper bounds on 
the differential equations for 
the Hadamard product of two rational series $F_1$ and $F_2$: we therefore aim at providing upper bounds for the maxdegrees and the coefficients of the polynomials that appear in the description of $F_1\odot F_2$
as partial differential equations in its variables $x_1$, \ldots $x_n$.

Our approach consists in taking Lipshitz's construction that proves that the Hadamard product of two holonomic series is holonomic, and studying it in details when the two series are rational. Let $\vect x = (x_1,\ldots,x_n)$ and $\vect y=(y_1,\ldots,y_n)$ be vectors of formal variables. In a nutshell, 
starting with $F_1(\vect x)$ and $F_2(\vect x)$, Lipshitz introduces the series $F(\vect x,\vect y)$ defined by
\begin{equation}\label{eq:F}
F(\vect x,\vect y) = \frac1{y_1\cdots y_n}
F_1\left(\frac{x_1}{y_1},\ldots,\frac{x_n}{y_n}\right)\cdot F_2(y_1,\ldots,y_n).
\end{equation}
Though this doubles the number of variables,
it transforms the Hadamard product into a simple product and a coefficient extraction, as one can establish that
\begin{equation}\label{eq:hadamard lipshitz}
F_1(\vect x)\odot F_2(\vect x) = [y_1^{-1}\cdots y_n^{-1}]F(\vect x, \vect y).
\end{equation}
Indeed, let us write $a_{i_1,\ldots, i_n}=[x^{i_1}\ldots x^{i_n}]F_1$ and $b_{i_1,\ldots, i_n}=[x^{i_1}\ldots x^{i_n}]F_2$, so that: 
\[F(\vect x,\vect y)=\frac{1}{y_1\ldots y_n}\sum_{i_1,\ldots, i_n}\sum_{j_1,\ldots, j_n}a_{i_1,\ldots, i_n}b_{j_1,\ldots, j_n}\big(\tfrac{x_1}{y_1}\big)^{i_1}\ldots \big(\tfrac{x_n}{y_n}\big)^{i_n}y_1^{j_1}\ldots y_n^{j_n}.\]
Hence $\displaystyle[(y_1\ldots y_n)^{-1}]F=[(y_1\ldots y_n)^{0}]\sum_{\substack{i_1,\ldots, i_n\\j_1,\ldots, j_n}}a_{i_1,\ldots, i_n}b_{j_1,\ldots, j_n}x_1^{i_1}\ldots x_n^{i_n}y_1^{j_1-i_1}\ldots y_n^{j_n-i_n}$. The terms in $y_1^0$ in the sum appear when $i_1=j_1$. More generally, the terms in $y_1^0\ldots y_n^0$ in the sum appear every time in the sum we have $i_1=j_1$, \ldots, $i_n=j_n$.

So $[(y_1\ldots y_n)^{-1}]F=\sum_{i_1,\ldots, i_n}a_{i_1,\ldots, i_n}b_{i_1,\ldots, i_n}x_1^{i_1}\ldots x_n^{i_n}=F_1(\vect x)\odot F_2(\vect x)$.

\begin{remark}
	Multiplying by a factor $\frac{1}{y_1\ldots y_n}$ in the definition of $F$, to  extract later  $[(y_1\ldots y_n)^{-1}]F$, may seem useless. It turns out that Lipshitz's proof relies strongly on the fact that the extracted coefficient to retrieve the Hadamard's product is indeed $[(y_1\ldots y_n)^{-1}]F$. The main reason is that $(y_1\ldots y_n)^{-1}$ cannot come from the derivation of any term $y_1^{i_1}\ldots y_n^{i_n}$, so that the differential equations for the Hadamard's product can be extracted from the differential equations for $F$.
\end{remark}

\begin{remark}
We used the most straightforward approach, as our primary goal is to provide a first upper bound for the complexity of the inclusion test for weakly-unambiguous automata. Other techniques, based on integral representation of the Hadamard product, may give better bounds, as it was done for some combinatorial walks in~\cite{bostan2017}.
\end{remark}

In this section, as we  already did in the presentation above,  we use vector notations for readability:
$\vect x = (x_1,\ldots,x_n)$, $\vect x^{\vect \alpha}=x_1^{\alpha_1}\cdots x_n^{\alpha_n}$, $[\vect y^{-\vect 1}]=[y_1^{-1}\cdots y_n^{-1}]$, and so on.

\subsubsection{Statement and main steps}
The bounds for the Hadamard product of two rational series are given in the following theorem.%
\begin{restatable}{theorem}{boundhadamardproductoftworational}\label{thm:boundhadamardproductoftworational}
Let $P_1,Q_1,P_2$ and $Q_2$ be four polynomials of $\mathbb{Z}[\vect x]$. We define the quantities
\begin{align*}
M & = 1 + \max\left\{\mdeg(P_1), \mdeg(P_2), \mdeg(Q_1), \mdeg(Q_2)\right\}\\
\normF & = \max\left\{\|P_1\|_\infty,\|P_2\|_\infty,\|Q_1\|_\infty,\|Q_2\|_\infty\right\}
\end{align*}
For every $j\in\{1,\ldots,n\}$ there exist $r\in\mathbb{N}$ and $r+1$ polynomials $p_0(\vect x)$,\ldots, $p_r(\vect x)$ of $\mathbb{Z}[\vect x]$, not all null, such that the series  $\frac{P_1}{Q_1}\odot \frac{P_2}{Q_2}(\vect x)$ satisfies the equation:
\[p_r(\vect x)\partial_{x_j}^rf(\vect x)+\ldots+ p_0(\vect x)f(\vect x)=0,\]
with $r+\deg({p_{i}})< (nM)^{O(n)}$ and $\log\|p_{i}\|_\infty\leq (nM)^{O(n^2)}(1+\log(\normF))$.
\end{restatable}

\begin{remark}
In the statement of the previous theorem, the notation $O(n)$ (resp. $O(n^2)$) is only here to simplify the writing and does not imply that $n$ tends to infinity. It simply means that there exists a constant $c$ (independent of all the parameters) such that for all $n \geq 1$, the inequality holds when $O(n)$ (resp. $O(n^2)$) is replaced by $cn$ (resp. $c n^2$). 
\end{remark}

The proof consists of several steps. As announced at the beginning of the section, we first build a rational function $F$ in $2n$ variables. In the following lemma we compute bounds for $F$.

\begin{lemma}\label{lem:preparation_hadamard}
With the notations of Theorem~\ref{thm:boundhadamardproductoftworational},
if $F$ is defined by
\[
F(\vect x, \vect y) = \frac1{y_1\cdots y_n}
\frac{P_1}{Q_1}\left(\frac{x_1}{y_1},\ldots,\frac{x_n}{y_n}\right)\cdot \frac{P_2}{Q_2}(y_1,\ldots,y_n),
\]
then there exist two polynomials $P$ and $Q$ of $\mathbb{Z}[\vect x, \vect y]$ such that $F=P/Q$,
$\mdeg(P)< 2M$, $\mdeg(Q)< 2M$,
 $\|P\|_\infty\leq\normF^2$ and  $\|Q\|_\infty\leq \normF^2$.
\end{lemma}

\begin{proof}
	In $P_1(\frac{x_1}{y_1}, \ldots, \frac{x_n}{y_n})$, the variables $y_1, \ldots, y_n$ have negative exponents, the lowest negative exponent being $\mdeg(P_1)\leq M-1$. Hence $(y_1\ldots y_n)^{M-1}P_1(\frac{x_1}{y_1}, \ldots, \frac{x_n}{y_n})$ is a polynomial in $\vect x$ and $\vect y$, of maxdegree at most  $M$, and of same infinite norm as $P_1$. Hence the function $P$ defined by
	\[
	P(\vect x,\vect y) = (y_1\cdots y_n)^M P_1\left(\frac{x_1}{y_1},\ldots,\frac{x_n}{y_n}\right)P_2(y_1,\ldots,y_n),
	\]
	is a polynomial of $\mathbb{Z}[\vect x, \vect y]$, with $\mdeg(P)\leq 2M-2$. Moreover, as $P_2$ has no variable of the form $x_i$, when expanding the product of $P_1$ and $P_2$, both written as a sum of monomials, there is only one way to produce each monomial of $P$: each coefficient of $P$ is the product of a coefficient of $P_1$ and a coefficient of $P_2$. Hence 
	$\|P\|_\infty\leq\|P_1\|_\infty\cdot\|P_2\|_\infty\leq \normF^2$.
	
	The same argument applies for $Q=(y_1\ldots y_n)^{M }Q_1(\frac{x_1}{y_1}, \ldots, \frac{x_n}{y_n})Q_2(y_1, \ldots, y_n)$, except for the additionnal factor $(y_1\ldots y_n)$ which increases by $1$ the maxdegree: $\mdeg(Q)\leq 2M-1$
\end{proof}

Once we have bounds on the representation of $F$ as a quotient of two polynomials, the main computations consist in studying the partial differential equations in $F$ by following Lipshitz's constructions. This is done in the following proposition, whose technical proof is given in the next subsection.

\begin{restatable}{proposition}{boundhadamardlipshitz}\label{prop:boundhadamardlipshitz}
Let $F$ be a rational fraction given under the form $F=\frac{P}{Q}\in\mathbb{Z}(\vect x,\vect y)$, where $P(\vect x, \vect y)$ and $Q(\vect x, \vect y)$ are  two polynomials of $2n$ variables with integer coefficients. Let $M=\max(\mdeg(P), \mdeg(Q))+1$,  
$\normF = \max(\|P\|_\infty,\|Q\|_\infty)$, and $N=((2n+1)M)^{2n}-2n$. Then $f=[y_1^{-1}\ldots y_n^{-1}] F$ is an element of $\mathbb{Z}(\vect x)$ that satisfies for every $j\in\{1,\ldots,n\}$ a non-trivial partial differential equation in $x_j$ with polynomial coefficients (the $p_i$'s and $r$ depends on $j$): \[p_r(\vect x)\partial_{x_j}^rf(\vect x)+ p_{r-1}(\vect x)\partial_{x_j}^{r-1}f(\vect x)+\ldots+ p_0(\vect x)f(\vect x)=0,\]
with $r+\mdeg({p_{i}})< N$, and $\log\|p_{i}\|_\infty\leq 6 N^{2n+1}M^{2n} \log (N\normF)$.
\end{restatable}

The proof of Theorem~\ref{thm:boundhadamardproductoftworational} just consists in combining Lemma~\ref{lem:preparation_hadamard} and Proposition~\ref{prop:boundhadamardlipshitz}, keeping in mind that Lemma~\ref{lem:preparation_hadamard} doubles the maxdegree bound and squares the infinite norm bound.

\subsubsection{Proof of Proposition \ref{prop:boundhadamardlipshitz}}\label{sec:boundhadamardlipshitz}

In this section, we prove Proposition~\ref{prop:boundhadamardlipshitz}. We thus use the notations of its statement. We start with three technical lemmas.

\begin{lemma}\label{lm:linear1}
Let $\vect{\alpha}=(\alpha_1, \ldots, \alpha_{2n})\in\NN^{2n}$. We denote by $\partial_{\vect{\alpha}}$ the operator defined by $\partial_{\vect{\alpha}}=\partial_{x_1}^{\alpha_1}\ldots\partial_{x_n}^{\alpha_n}\partial_{y_1}^{\alpha_{n+1}}\ldots\partial_{y_n}^{\alpha_{2n}}$. There exists a polynomial $P_{\vect\alpha}$ of $\mathbb{Z}[\vect x,\vect y]$ such that $Q^{s+1}\partial_{\vect{\alpha}}F=P_{\vect{\alpha}}$, $\mdeg(P_{\vect{\alpha}})\leq (s+1)(M-1)$ and
$\|P_{\vect{\alpha}}\|_{\infty}\leq 2^ss!M^{s(2n+1)}\normF^{s+1}$,
where $s=\|\vect\alpha\|_1=\sum_i \alpha_i$.
\end{lemma}
\begin{proof}
By induction on $s=\|\vect\alpha\|_1$. If $s=0$, then $QF=P$ with $\mdeg(P)< M$, and $\|P\|_\infty\leq \normF$.

If the proposition is true for $s-1$,   let $\vect{\alpha}$ be such that $\|\vect\alpha\|_1=s$. We write $\vect{\alpha}=\vect{\beta} + \vect{e_i}$, for some $i$ such that $\alpha_i>0$, and we have $\|\vect\beta\|_1=s-1$. 
By induction hypothesis, $Q^{s}\partial_{\vect{\beta}}F=P_{\vect{\beta}}$, where $P_{\vect{\beta}}$ satisfies the hypotheses for $s-1$. Deriving this relation according to $x_i$ yields: 
\[(\partial_{x_i}Q)s Q^{s-1}\partial_{\vect{\beta}}F+Q^s\partial_{\vect{\alpha}}F=\partial_{x_i}P_{\vect{\beta}}.\]
By multiplying by $Q$, we obtain \[Q^{s+1}\partial_{\vect{\alpha}}F=Q\partial_{x_i}P_{\vect{\beta}}-(\partial_{x_i}Q)s Q^{s}\partial_{\vect{\beta}}F=Q\partial_{x_i}P_{\vect{\beta}}-(\partial_{x_i}Q)s P_{\vect{\beta}}.\]
Let $P_{\vect{\alpha}}=Q\partial_{x_i}P_{\vect{\beta}}-(\partial_{x_i}Q)s P_{\vect{\beta}}$. The polynomial
 $P_{\vect{\alpha}}$ has integer coefficients, and 
$\mdeg(P_{\vect{\alpha}})\leq M-1 + \mdeg(P_{\vect{\beta}})\leq (s+1)(M-1)$, as $\mdeg (P_{\vect\beta})\leq s(M-1)$ by induction hypothesis. Moreover, by Lemma~\ref{lm:multiplicationpolynomes}, we have
\[
\|Q\partial_{x_i}P_{\vect{\beta}} \|_\infty
\leq M^{2n}\|Q\|_\infty\cdot\|\partial_{x_i}P_{\vect{\beta}} \|_\infty \leq M^{2n}\|Q\|_\infty\cdot s(M-1)\|P_{\vect{\beta}} \|_\infty. 
\]
And similarly, 
\[
\|(\partial_{x_i}Q)sP_{\vect{\beta}} \|_\infty
\leq M^{2n}\|\partial_{x_i}Q\|_\infty\cdot s\|P_{\vect{\beta}} \|_\infty \leq M^{2n+1}\|Q\|_\infty\cdot s\|P_{\vect{\beta}} \|_\infty. 
\]
Thus, $\|P_{\vect{\alpha}}\|_\infty\leq2 M^{2n+1}s\|Q\|_\infty\cdot \|P_{\vect{\beta}}\|_\infty$. This concludes the proof as $P_{\vect\beta}$ satisfies the induction hypothesis.
\end{proof}

\begin{lemma}\label{lm:lipshitz_rationnel_family_derivative}
Let $N=((2n+1)M)^{2n}-2n$, 
let $\vect\alpha$ and $\vect\beta$ be two vectors of $\mathbb{N}^n$ and let $\gamma \in\NN$ such that $\sum_i\alpha_i + \sum_i\beta_i + \gamma < N$. %
Let also $j\in\{1,\ldots,n\}$. There exists a polynomial $R_{\vect\alpha,\vect\beta,\gamma,j}$ of variables $x_1$,\ldots, $x_n$, $y_1$, \ldots $y_n$, with integer coefficients, such that
\[Q^{N}x_1^{\alpha_1}\ldots x_n^{\alpha_n}\partial_{y_1}^{\beta_1}\ldots \partial_{y_n}^{\beta_n}\partial_{x_j}^{\gamma}F =  R_{\vect\alpha,\vect\beta,\gamma,j}.\]
Moreover, $\mdeg(R_{\vect\alpha,\vect\beta,\gamma,j})\leq NM$, and $\log \|R_{\vect\alpha,\vect\beta,\gamma,j}\|_\infty \leq 3 N \log (N  \normF)$.
\end{lemma}
\begin{proof}
Let  $\vect{u}=(\gamma, 0, \ldots, 0, \beta_1, \ldots, \beta_n)$, so that by Lemma~\ref{lm:linear1}, we have
$Q^{s+1}\partial_{\vect{u}}F=P_{\vect{u}}$, where
 $\partial_{\vect{u}}F=\partial_{y_1}^{\beta_1}\ldots \partial_{y_n}^{\beta_n}\partial_{x_j}^{\gamma}F$,  and $s=\|\vect\beta\|_1+\gamma$. Then 
\[
R_{\vect\alpha,\vect\beta,\gamma,j}=Q^{N}x_1^{\alpha_1}\ldots x_n^{\alpha_n}\partial_{\vect{u}}F=Q^{N-s-1}x_1^{\alpha_1}\ldots x_n^{\alpha_n}P_{\vect u}\]
is a polynomial as $N>s$.
Moreover, using Lemma~\ref{lm:linear1} again, we have
\begin{align*}
\mdeg(R_{\vect\alpha,\vect\beta,\gamma,j})
& \leq (N-s-1)(M-1)+\max(\alpha_1,\ldots,\alpha_n) + (s+1)(M-1)\\
& \leq N(M-1) + N = NM.
\end{align*}
For the bound on the infinite norm of $R_{\vect\alpha,\vect\beta,\gamma,j}$ we have, by Lemma~\ref{lm:multiplicationpolynomes}
\[
  \|R_{\vect\alpha,\vect\beta,\gamma,j}\|_\infty
   = \|Q^{N-s-1}P_{\vect u}\|_\infty 
  \leq M^{2n}\|Q\|_\infty\|Q^{N-s-2}P_{\vect u}\|_\infty 
  \leq M^{2n(N-s)}\|Q\|_\infty^{N-s-1} \cdot \|P_{\vect u}\|_\infty. 
\]
Lemma~\ref{lm:linear1} provides a bound on $\|P_{\vect u}\|_\infty$, yielding
\[
  \|R_{\vect\alpha,\vect\beta,\gamma,j}\|_\infty
  \leq M^{2n(N-s)}\|Q\|_\infty^{N-s-1} 2^s s!M^{s(2n+1)}\normF^{s+1}
  \leq 2^s s! M^{2nN+s}\normF^{N}.
\]
As $s<N$ and $M<N$, we get $\|R_{\vect\alpha,\vect\beta,\gamma,j}\|_\infty\leq 2^N N^N M^{(2n+1)N} \normF^N\leq N^{2N} (2M^{2n})^{N} \normF^N$. Therefore, as $2M^{2n}\leq 2M^{2n}+2n(M^{2n}-1)\leq (2n+2)M^{2n}-2n\leq N$,
\begin{align*}
\log \|R_{\vect\alpha,\vect\beta,\gamma,j}\|_\infty
& \leq \log\left(N^{3N} \normF^{3N}\right)
\leq 3N\log (N\normF),
\end{align*}
as announced.
\end{proof}

\begin{lemma}\label{lm:lipshitz_famille_liee}
	Let $N=((2n+1)M)^{2n}-2n$. %
	For every $j\in\{1,\ldots,n\}$, the family $\{x_1^{\alpha_1}\ldots x_n^{\alpha_n}\partial_{y_1}^{\beta_1}\ldots\partial_{y_n}^{\beta_n}\partial_{x_j}^\gamma F \ :\ \sum_i\alpha_i + \sum_i\beta_i+\gamma < N\}$ is dependent over $\mathbb Z$. More precisely, there is a non-trivial dependency relation of the form	
	\[\sum_{\|\vect \alpha\|_1+\|\vect \beta\|_1+\gamma < N}v_{\vect \alpha, \vect \beta, \gamma}\;x_1^{\alpha_1}\ldots x_n^{\alpha_n}\partial_{y_1}^{\beta_1}\ldots\partial_{y_n}^{\beta_n}\partial_{x_j}^\gamma F=0,\]
where each $v_{\vect \alpha, \vect \beta, \gamma}$,  is an integer. Moreover, we have
$\log |v_{\vect \alpha, \vect \beta, \gamma}|\leq 6 N^{2n+1}M^{2n} \log (N\normF)$.
\end{lemma}
\begin{proof}
	Let	$\mathcal D$ be the set of every derivative operators $x_1^{\alpha_1}\ldots x_n^{\alpha_n}\partial_{y_1}^{\beta_1}\ldots\partial_{y_n}^{\beta_n}\partial_{x_j}^\gamma$, with $\sum_i\alpha_i + \sum_i\beta_i+\gamma < N$.
	The number of such operators is $|\mathcal D|=\binom{N+2n}{2n+1}$.
	
	Let $\mathcal M$ be the set of all monomials of the form $\vect x ^{\vect i}\vect y^{\vect j} = x_1^{i_1}\ldots x_n^{i_n}{y_1}^{j_1}\ldots{y_n}^{j_n}$, such that
	$\mdeg(\vect x ^{\vect i}\vect y^{\vect j}) \leq NM-1$. We have $|\mathcal M| =(NM)^{2n}$.

By Lemma~\ref{lm:lipshitz_rationnel_family_derivative}, for every $\delta\in \mathcal D$,
$Q^{N}\delta F$  can be written as a linear combination over $\mathbb Z$ of elements of $\mathcal P$. These linear combinations can be synthesised in a $(NM)^{2n}\times \binom{N+2n}{2n+1}$ matrix $A$, over $\mathbb Z$. Our choice for $N$ implies that 
$\binom{N+2n}{2n+1}>(NM)^{2n}$, as detailed below, so $A$ has more columns than rows. Thus,  the family of all $Q^{N}\delta F$, for $\delta\in \mathcal D$, is linearly dependent.

According to Cramer's rule (see Proposition~\ref{prop:cramer-rule}), there exists a non-trivial solution $\vect v$ of the equation $A\vect v=\vect 0$, where every coordinate of $\vect v$ is a minor of $A$ or the opposite of a minor of $A$. As $A$ has integer entries, the coordinates of $\vect v$ are integers. Moreover, they are bounded by Hadamard's bound (in the non-polynomial case) by 
$\|A\|_\infty^p p^{p/2}$, where $p\leq (NM)^{2n}$ so that  $p\leq(2n+1)^{4n^2}M^{4n^2+2n}$; 
hence $\log p \leq 6 n^2 \log((2n+1)M)$. We have
\[
\log(p^{p/2}) \leq  3n^2 p \log\left((2n+1)M\right).
\]
Moreover, since each column of $A$ corresponds to a $R_{\vect \alpha,\vect\beta,\gamma,j}$ of Lemma~\ref{lm:lipshitz_rationnel_family_derivative}, we have
\[
\log \|A\|_\infty^p
\leq 3 N p \log (N \normF).\]
Hence, since the latter bound is greater than the former:
\[
\log|v_{\vect \alpha,\vect\beta,\gamma}|
\leq 6 N p \log (N \normF) = 6 N(NM)^{2n}\log(N\normF).
\]
This concludes the proof, provided we  verify that $\binom{N+2n}{2n+1}>(NM)^{2n}$:
\[
\frac{\binom{N+2n}{2n+1}}{(NM)^{2n}}
 = \frac{(N+2n)(N+2n-1)\cdots (N+1)N}{(2n+1)!N^{2n}M^{2n}}
 \geq \frac{N+2n}{(2n+1)^{2n}M^{2n}},
\]
using the fact that $k! \leq k^{k-1}$ for $k\geq 1$.
\end{proof}

\begin{proof}[Proof of Proposition \ref{prop:boundhadamardlipshitz}]
By Lemma~\ref{lm:lipshitz_famille_liee}, we have a non-trivial dependency relation  of the form	
	\[\sum_{\|\vect \alpha\|_1+\|\vect \beta\|_1+\gamma< N}\;v_{\vect \alpha, \vect \beta, \gamma}x_1^{\alpha_1}\ldots x_n^{\alpha_n}\partial_{y_1}^{\beta_1}\ldots\partial_{y_n}^{\beta_n}\partial_{x_j}^\gamma F=0,\]
where $v_{\vect \alpha, \vect \beta, \gamma}$ are integers with $\log|v_{\vect \alpha, \vect \beta, \gamma}|\leq 6N^{2n+1}M^{2n}\log(N\normF)$.

We  factorize this relation according to the derivatives in $y_1, \ldots, y_n$: it therefore rewrites 
\begin{equation}\label{eq:beta}
\sum_{\vect \beta}p_{\vect \beta}(\vect x, \partial_{x_j})\partial_{y_1}^{\beta_1}\ldots\partial_{y_n}^{\beta_n} F=0,
\end{equation}
where $p_{\vect \beta}(\vect x, \partial_{x_j})$ is a polynomial in the variables $\vect x$ and $\partial_{x_j}$, with integer coefficients, verifying $\log\|p_{\vect \beta}(\vect x, \partial_{x_1})\|_\infty\leq 6N^{2n+1}M^{2n}\log(N\normF)$ and $\mdeg(p_{\vect \beta}(\vect x, \partial_{x_j})) < N$ (seen as a polynomial $x_1$, \ldots $x_n$ and in $\partial_{x_j}$, where $\partial_{x_j}$ do not commute with the $x_i$'s).

\newcommand{\bbetamin}{\ensuremath{\vect{\beta^{\mathrm{min}}}}}
\newcommand{\betamin}{\ensuremath{\beta^{\mathrm{min}}}}
\newcommand{\vbeta}{\ensuremath{\vect{\beta}}}

Let $\bbetamin$ be the smallest vector, for the lexicographic order, amongst the $\vect{\beta}$ such that $p_{\vect{\beta}} \neq 0$ in Equation~\eqref{eq:beta}. It exists as the relation is non-trivial. 

We now show that $p_{\bbetamin}(\vect{x},\partial x_j)(f)=0$. We use the vector notation for monomials: $\vect y^{-\bbetamin}=y_1^{-\betamin_1}\cdots y_n^{-\betamin_n}$, $\vect y^{\boldsymbol{-1}} = y_1^{-1}\cdots y_n^{-1} $, \ldots

We rely on the fact, established at the end of this proof, that for all $\vect{\beta}$ we have
\begin{equation}
\label{eq:tech}
[\vect y^{-\bbetamin-\vect 1}] \partial_{y_1}^{\beta_1} \cdots \partial_{y_n}^{\beta_n}  F=
\begin{cases}
    0  & \text{if } \vbeta \neq \bbetamin  \\
    (-1)^{\beta_1+\ldots\beta_n}\beta_1!\ldots \beta_n![\vect y^{-\vect 1}]F & \text{if } \vbeta = \bbetamin \\
\end{cases}
\end{equation}
Once established, we
extract the coefficient in $\vect y^{-\bbetamin_1-\vect 1}$ in Equation~\eqref{eq:beta}, yielding:
\[
 (-1)^{\beta_1+\ldots\beta_n}\beta_1!\ldots \beta_n!\cdot p_{\bbetamin}(\vect{x},\partial x_j) [\vect y^{-\vect 1}]F = 0.
\]
Since $f=[\vect y^{-\vect 1}]F$ by definition, we have established that $p_{\bbetamin}(\vect{x},\partial x_j)(f)=0$. This concludes the proof; the bound on the max-degrees and on the coefficients is a consequence of the bounds stated in Lemma~\ref{lm:lipshitz_famille_liee}.

So, it only remains to prove Equation~\eqref{eq:tech}.
Recall that $\bbetamin$ is lexicographically smaller than $\vbeta$. By grouping monomials, we can write $F$ under the form:
\[
F = \sum_{\vect{\alpha}} H_{\vect \alpha}(\vect{x}) y_1^{\alpha_1} \cdots y_n^{\alpha_n}. 
\]
Hence  we have:
\[[\vect y^{-\bbetamin-\vect{1}}]\partial_{y_1}^{\beta_1} \cdots \partial_{y_n}^{\beta_n}  F  
  =  [\vect y^{-\bbetamin-\vect{1}}] \sum_{\vect{\alpha}} H_{\vect{\alpha}}(\vect{x}) \partial_{y_1}^{\beta_1} \cdots \partial_{y_n}^{\beta_n} y_1^{\alpha_1} \cdots y_n^{\alpha_n}.\]
At this point, one can see that if any of the $\alpha_i$'s is non-negative, then
$\partial_{y_1}^{\beta_1} \cdots \partial_{y_n}^{\beta_n} y_1^{\alpha_1} \cdots y_n^{\alpha_n}$ cannot contribute to the coefficient of $\vect y^{-\bbetamin-\vect{1}}$. Thus we can limit the sum over the $\vect\alpha$'s with negative coefficients only. Moreover, if $\vect\beta \neq \bbetamin$, then there is one coordinate $i$ such that $\beta_i>\betamin_i$, and  the exponent in $y_i$ of $\partial_{y_1}^{\beta_1} \cdots \partial_{y_n}^{\beta_n} y_1^{\alpha_1} \cdots y_n^{\alpha_n}$ is $\alpha_i-\beta_i < -1 - \betamin_i$: it does not contribute to the coefficient of $\vect y^{-\bbetamin-\vect{1}}$ either. Similarly, by minimality of $\bbetamin$ for the lexicographic order, $\vect\alpha$ has to be $-\vect 1$ to contribute.
Therefore, we have:
\begin{align*}
[\vect y^{-\bbetamin-\vect{1}}]\partial_{y_1}^{\beta_1} \cdots \partial_{y_n}^{\beta_n}  F  
  & =  [\vect y^{-\bbetamin-\vect{1}}]
  H_{-\vect{1}}(\vect{x}) \partial_{y_1}^{\betamin_1} \cdots \partial_{y_n}^{\betamin_n} \vect y^{\vect -1} \\
  & = H_{-\vect{1}}(\vect{x}) (-1)^{\betamin_1+\ldots+\betamin_n}\betamin_1!\cdots\betamin_n!\\
  & = (-1)^{\betamin_1+\ldots+\betamin_n}\betamin_1!\cdots\betamin_n!\;[\vect y^{-\vect 1}]F.
\end{align*}
This establishes Equation~\eqref{eq:tech}, and therefore concludes the proof.
\end{proof}

\section{Proof of Proposition \ref{prop:bound-rec-pa}}

This section is devoted to finding bounds on the order and coefficients of a differential equation verified by the generating series of a weakly-unambiguous automaton. 

For this section, we fix a weakly-umambiguous $\PA$ $\AA=(\Sigma, Q, q_I, F, C, \Delta)$ of dimension $d$ whose semi-linear constraint $C$ given in an unambiguous presentation $C\eqdef\uplus_{i=1}^p \vect{c_i} + P_i^*$.  Recall that we denote by $|\AA| \eqdef |Q|+|\Delta|+p+\sum_{i} |P_i|$ and by $\|\AA\|_{\infty}$ the maximum coordinate of a vector appearing in $\Delta$, the $\vect{c_i}$'s and the $P_i$'s. 

In this section, we assume that the automaton is non-trivial: it contains at least one state, at least one transition with a non-null vector,  it semilinear constraint set contains at least one non-null vector. Note that if the automaton is trivial, the computation of its generating series is straightforward.

Recall that in the the proof of Proposition~\ref{prop:unambPAholonomic}, we showed that the
generating series $A(x)$ of $\AA$ is equal to the specialization $G(x,1,\ldots,1)$ of the weighted series of $\AA$. Furthermore, the weighted series can be expressed as 
\[G(x,y_1,\ldots,y_d) = \overline{A}(x,y_1,\ldots,y_n) \odot  \overline{C}(x,y_1,\ldots,y_d) \]
where $\overline{A}(x,y_1,\ldots,y_n)$  and $\overline{C}(x,y_1,\ldots,y_d)$ are the two rational series defined by:
\begin{itemize}
\item for all $(n,i_1,\ldots,i_d) \in \mathbb{N}^{d+1}$, $[x^n y_1^{i_1} \cdots y_d^{i_d}] \overline{A}$ counts the numbers of runs of $\AA$ of length $n$ labeled by the vector $(i_1,\ldots,i_d)$ starting in the initial state and ending in a final state;
\item $\overline{C}(x,y_1,\ldots,y_d) = \frac{C(y_1,\ldots,y_d)}{1-x}$ where $C(y_1,\ldots,y_d)$ is the support series of the  semi-linear set $C$. 	
\end{itemize}

Using Lemma~\ref{lm:automatontorational}, we bound rational fractions for $\overline{A}$ and $\overline{C}$ in terms of $d$, $|\AA|$ and $\|\AA|_\infty$.

\begin{lemma}
\label{lem:bounds-a-c}
 The rational generating series $\overline{A}$ and $\overline{C}$ can be written as $\frac{P}{Q}(x,y_1,\ldots,y_d)$ with $\mdeg(P),\mdeg(Q) \leq |\AA| \cdot \|\AA\|_\infty -1$ and $\|P\|_\infty,\|Q\|_\infty \leq  |\AA|^{5|\AA|/2}$.	
\end{lemma}

\begin{proof} For $\bar{A}$, we define a vector automata $\VV$ of dimension $d+1$ whose generating series $V(x_1,y_1,\ldots,y_n)$ satisfies $V(x,y_1,\ldots,y_d) = A(x,y_1,\ldots,y_d)$.

The vector automata $\VV=(Q,q_I,F,\Delta_\VV)$ has the same state set $Q$ as $\AA$, the same inital state and final states. For every transition $(p,a,\vect{v},q)$ of $\AA$, we have a transition $(p,(1,v_1,\ldots,v_d),q)$ in $\Delta_\VV$. By construction, we have $V(x,y_1,\ldots,y_d) = A(x,y_1,\ldots,y_d)$.
From Lemma~\ref{lm:automatontorational}, we have that $A=\frac{P}{Q}(x,y_1,\ldots,y_d)$ with  $\mdeg(P), \mdeg(Q)\leq |Q|\cdot \|\VV\|_\infty \leq |\AA| \cdot \|\AA\|_\infty -1$ and $
    \|P\|_\infty, \|Q\|_\infty \leq {(1+d^{out}_{\VV})}^{|Q|}{|Q|}^{|Q|/2} \leq |\AA|^{|\AA|/2+1} \leq |\AA|^{5|\AA|/2}$.

We now derive similar bounds for $C(y_1,\ldots,y_d)$ the 
support series of the semi-linear set $C$. From the unambiguous presentation of $C$, we build an unambiguous vector automaton $\VV'$ accepting $C$. It immediately follows that $\VV'(y_1,\ldots,y_d)=C(y_1,\ldots,y_d)$. 

For all $j \in [1,p]$, we fix an enumeration $\vect{v}_1^j,\ldots,\vect{v}_{|P_j|}^j$ of the 
vectors in $P_j$. The vector automaton $\VV'$ has an initial state $q_I$ and a state for each vector in the $P_j$'s. For $j \in [1,p]$ and $k \in [1,|P_j|]$, we denote by $q_k^j$ the state corresponding to the vector $\vect{v}_k^j$.

If none of the $\vect{c_i} \in [1,p]$ is equal to $\vect{0}$, the construction is straightforward. 
From the initial state $q_I$, there is a transition labeled by $\vect{c}_j$ to every state $q_k^j$ for all $j \in [1,p]$ and $k \in [1,|P_j|]$. In addition, for all $j \in [1,p]$ and $k'\leq k \in [1,|P_j|]$, there is a transition from $q_{k'}^j$ to $q_k^j$ labeled by $\vect{v}_{k'}^j$. The unambiguity of the automaton $\VV'$ immediately follows from the unambiguity of the presentation of $C$. Note that we need to require that all the $\vect{c_i}$'s are non-null as our definition of vector automaton does not allow for transitions labeled by $\vect{0}$.

If one of the vectors $\vect{c_i}$ is equal to $\vect{0}$, then by unambiguity of the presentation, exactly one of the $\vect{c_i}$ is equal to $\vect{0}$.  W.l.o.g. we assume that $\vect{c_0}=\vect{0}$. From the initial state $q_I$, there is a transition labeled by $\vect{c}_j$ to every state $q_k^j$ for all $j \in [2,p]$ and $k \in [1,|P_j|]$. Furthermore, for all $k \in [1,|P_1|]$, there is a transition from $q_I$ to $q_k^1$ labelled by $\vect{v}_k^1$. In addition, for all $j \in [1,p]$ and $k'\leq k \in [1,|P_j|]$, there is a transition from $q_{k'}^j$ to $q_k^j$ labeled by $\vect{v}_{k'}^j$. The unambiguity of the automaton $\VV'$ immediately follows from the unambiguity of the presentation of $C$.

In both cases, the out-degree of $\VV'$ is at most $\sum_{j \in [1,p]} (|P_j|)\leq |\AA|-1$, its number of states is $\sum_{j \in [1,p]} (|P_j|)+1 \leq |\AA|$ and $\|\VV'\|_\infty \leq \|\AA\|_\infty$.
By Lemma~\ref{lm:automatontorational},  we have that $C=\frac{P}{Q}(x,y_1,\ldots,y_d)$ with  $\mdeg(P), \mdeg(Q)\leq |\AA|\cdot \|\AA\|_\infty-1$ and $
    \|P\|_\infty, \|Q\|_\infty \leq {|\AA|}^{|\AA|}{|\AA|}^{|\AA|/2}  \leq |\AA|^{5|\AA|/2}$.
    
For the bounds on $\overline{C}$, 
remark that $\overline{C}= \frac{1}{1-x} \cdot C = \frac{P(y_1,\ldots,y_d)}{(1-x) \cdot Q(y_1,\ldots,y_d)}$. As $\|(1-x) \cdot Q(y_1,\ldots,y_d)\|_\infty = \|Q(y_1,\ldots,y_d)\|_\infty$ and $\mdeg((1-x) \cdot Q(y_1,\ldots,y_d))=\mdeg(Q(y_1,\ldots,y_d))$, we can conclude.
\end{proof}
	
We can now apply the bounds for the Hadamard product of two rational series obtained in Theorem~\ref{thm:boundhadamardproductoftworational} and the bound for the specialization to $1$ obtained in Proposition~\ref{prop:boundsspecialization} to obtain a bound on differential equations satisfied by $A(x)$.

\propboundrecpa*

\begin{proof}
Combining the bounds for the Hadamard product obtained in Theorem~\ref{thm:boundhadamardproductoftworational} with the bounds for $\overline{A}$ and $\overline{C}$ in Lemma~\ref{lem:bounds-a-c}, we obtain that the weighted generating series $G(x,y_1,\ldots,y_d)$ satisfy a non-trivial equation of the form:
\[
p_r(x,\vect y)\partial_{x}^r G(x,\vect y)+\ldots+ p_0(\vect y)G(x,\vect y)=0,
\]
with for all $r+\deg({p_{i}})< ((d+2)M)^{O(d)}$ and $\log\|p_{i}\|_\infty\leq ((d+2)M)^{O(d^2)}(1+\log(\normF))$ where
$M  = |\AA| \cdots \|\AA\|_\infty$ and $\normF  = |\AA|^{5|\AA|/2}$. For all $i \in [0,r]$, these bounds simplify to:
$r+\deg({p_{i}})< ((d+1)|\AA|\|\AA\|_\infty)^{O(d)}$ and $\log\|p_{i}\|_\infty\leq ((d+1)|\AA| \|\AA\|_\infty)^{O(d^2)}$.

By Proposition~\ref{prop:boundsspecialization}, we obtain that $A(x)=G(x,1,\ldots,x)$ satisfies a non-trivial equation of the form:
	\[q_s(x)\partial_x^s A(x)+\cdots+ q_0(x)A(x)=0\]
with $s \leq r$ and for $i \in [0,s]$, $\deg(q_i) \leq \mdeg(p_i)$ and $\|q_i\|_\infty \leq \|p_i\|_\infty   (\deg_m(p_i)+1)^d \; 2^{\deg(p_i)}$.

In particular for all $i \in [0,s]$, $\log \|q_i\|_\infty \leq ((d+1)|\AA| \|\AA\|_\infty)^{O(d^2)}$.

\end{proof}

\section{Proof of Theorem \ref{thm:inclusionbound}}

\theoreminclusion*

\begin{proof}
By Proposition~\ref{prop:intersectionPA}, we can construct a weakly-unambiguous $\PA$ $\CC$ accepting $L(\AA) \cap L(\BB)$ with $|\CC| \leq |\AA| |\BB|$, $\|\CC\|_\infty = \max(\|\AA\|_\infty,\|\BB\|_\infty)$ of dimension $d_\AA+d_\BB$.  The series $H(x)\eqdef A(x)-C(x)$ counts the number of words of length $n$ in $L(\AA) \subseteq L(\BB)$. In particular, $L(\AA) \subseteq L(\BB)$ if and only if $H(x)=0$. 

Being the difference of two holonomic series, $H(x)=\sum_{n \geq 0} u_n x^n$ is holonomic. Therefore, its coefficients satisfy a recurrence equation of the form:
\[\sum_{k=-s}^{S} t_k(n)u_{n+k}=0, \, \text{for $n\geq s$, with $t_S\not= 0$}\] 
which we can be rewritten as:
\[
t_S(n-S)u_{n} = -\sum_{k=1}^{S+s}t_{S-k}(n-S) u_{n-k} \, \text{for $n\geq S+s$}.
\]
In particular if $t_S(n-S)$ is not null, if $u_{n-1}=\cdots=u_{n-S-s}=0$ then $u_n=0$. To ensure that $t_S(n-S)$ is not null, it is enough to take $n-S \geq \|t_S\|_\infty+1$ (see Lemma~\ref{lem:bound poly}). In particular, if we take $W = s+S+\|t_S\|+1$, we have that $H(x)=0$ if and only if for all $n \leq W$, $u_n=0$. In terms of languages, this equivalence can be restated as $L(\AA) \neq\subseteq L(\BB)$ if and only if $L(\AA)\setminus L(\BB)$ contains a word of length at most $W$.

To obtain a bound on $W$ in terms of the size of $\AA$ and $\BB$, we first obtain bounds on differential partial equation satisfies by $H$. If we take $d=d_\AA+d_\BB$, $M=|A||B|\|A\|_\infty \|B\|_\infty$, we have, by Proposition~\ref{prop:bound-rec-pa}, that both $A$ and $C$ satisfies a non-trivial differential equation of the form:
\[
p_t(x) \partial^t F(x) + \cdots + p_0(x) F(x) = 0
\]
with $t \leq (d M)^{O(d)}$, $\deg(p_i) \leq (dM)^{O(d)}$ and $\log(\|p_i\|)\leq (dM)^{O(d^2)}$.

By Proposition~\ref{prop:bound-holonomic-sum}, $H$ satisfies a non-trivial differential equation of the form:
\[
q_r(x) \partial^r H(x) + \cdots + q_0(x) H(x) = 0
\]
with $r \leq (d M)^{O(d)}$, $\deg(p_i) \leq (dM)^{O(d)}$ and $\log(\|p_i\|)\leq (dM)^{O(d^2)}$.

By Proposition~\ref{prop:boundsdiffeqtorec}, its coefficients satisfy a recurrence equation of the form:
\[\sum_{k=-s}^{S} t_k(n)u_{n+k}=0, \, \text{for $n\geq s$, with $t_S\not= 0$}\] 
with $s \leq (dM)^{O(d)}$, $S \leq (d M)^{O(d)}$ and $\log(\|t_S\|_\infty) \leq  (dM)^{O(d^2)}$.

Hence we can take $W = 2^{2^{O(d^2\log(dM))}}$. 

\end{proof}

\section{Algorithmic consequences: the inclusion problem}

The aim of this section is to prove our announced result on the complexity of testing the inclusion:

\corcomplexite*

To prove this statement, we need some preliminary algorithms to be able to count the numbers of accepted words of a fixed size $n$ in a weakly-unambiguous automaton.

\subsection{Counting the runs in a vector automaton}
\renewcommand{\VV}[0]{\mathcal V}

In this section, we are given a vector automaton $\VV=(Q, q_I, F, \Delta)$ over $\NN^d$. We recall that $Q$ is the set of states, $q_I$ the initial state, $F$ the final state, $\Delta$ the set of transitions of the form $(q,\vect v, q')$ with $\vect v\neq \vect 0$. We denote by $\|\VV\|_\infty$ the greatest coordinate appearing in the vectors in $\Delta$.

For two vectors $\vect{v},\vect{v'}$ in $\NN^d$, we note $\vect{v}\preceq\vect{v'}$ if for every $i\in [d]$, $v_i\leq v'_i$. 

For every $q\in Q$, for every vector $\vect v\in \NN^d$, we denote by $q[\vect v]$ the number of runs in the automaton from $q_I$ to $q$ labeled by the vector $\vect v$.

The goal of this section is to compute, for any given a bound $K$, the numbers $q[\vect v]$ for all state $q\in Q$, and all vector $\vect{v}$ such that $\|\vect v\|_\infty\leq K$.  

A naive approach consists in enumerating every run of length at most $(K+1)^d$ in the automaton, and keep track of the vectors labelling them, leading to an algorithm exponential in $K$. To improve on this, first observe that the numbers $q[\vect v]$ for $q\in Q$, $\vect v\in \NN^d$ verify the following recurrence:
	\begin{equation}\label{eq:recurrenceforvectorautomaton}q[\vect v]=\sum_{\substack{(q',\vect u,q)\in \Delta\\\text{ with } \vect u\preceq \vect v}}q'[\vect v - \vect u].
	\end{equation}
Using Equation~\eqref{eq:recurrenceforvectorautomaton} and
dynamic programming, we get the following statement.

\begin{proposition}\label{prop:computecoeffsvectorautomaton}
	Given $K\in\mathbb{N}$ and a vector automaton $\VV=(Q, q_I, F, \Delta)$, we can compute all the values $q[\vect v]$ for $q\in Q$ and $\vect v\in\mathbb{N}^d$ such that $\|\vect v\|_\infty\leq K$ in $\OO((K+1)^{2d}|\Delta|\log{|\Delta|})$ time.
\end{proposition}
\begin{proof}
	Notice that since there are no transitions labeled by $\vect 0$, a run in $\mathcal V$ labeled by $\vect v$ verifying $\|\vect v\|_\infty\leq K$ takes at most $(K+1)^d$ transitions. So $q[\vect v]\leq |\Delta|^{(K+1)^d}$ for every $q\in Q$ and every $\vect v$ such that $\|\vect v\|_\infty\leq K$.
	
	We simply unroll the recurrence formula in Equation~\eqref{eq:recurrenceforvectorautomaton}, by enumerating all vectors $\vect v$ such that $\|\vect{v}\|_\infty\leq K$ in lexicographic order. 
	 Enumerating such vectors is done with $(K+1)^d$ iterations. Inside these iterations, we make at most $|\Delta|$ operations consisting in computing the vectors $\vect v-\vect u$, which is done in $\OO(d\log K)$, and adding $q'[\vect v-\vect u]$ to $q[\vect v]$, which is done in $\OO((K+1)^d\log|\Delta|)$. 
\end{proof}

\subsection{Counting the accepted words in a weakly-unambiguous automaton}

In this section, we want an algorithmic solution to count the total numbers of words of length at most $n\in\NN$ accepted by a given weakly-unmabiguous automaton.

\begin{proposition}\label{prop:enumeration_weakly_unambiguous}Given a weakly-unambiguous PA $\mathcal{B}=(Q, q_I, F, C, \Delta)$ of dimension $d$, without $\varepsilon$-transitions, the total number of words of length at most $n$ accepted by $\BB$
can be computed in $(n\|\BB\|_\infty|\BB|)^{\OO(d)}$ time. 
\end{proposition}

\begin{proof}
    We rely on the same constructions as in the proof of Lemma~\ref{lem:bounds-a-c}. First, from $\BB$ we can build the vector automaton $\VV$ of dimension $d+1$ for the paths in $\BB$ from $q_I$ to a final state: such a path $(w,\vect v)$ in $\BB$ is in one-to-one relation with a path labeled $(|w|,v_1,\ldots,v_d)$ in $\VV$. By construction, $\VV$ and $\BB$ share the same transition system. As there is no $\varepsilon$-transitions in $\BB$, any run labeled by $(w,\vect v)$ has length $|w|$ and is such that $\|\vect v\|_\infty\leq |w|\,\|\BB\|_\infty$. Hence, applying Proposition~\ref{prop:computecoeffsvectorautomaton} to $\VV$ with $K=n\|\BB\|_\infty$ we can compute in time $\OO((n\|\BB\|_\infty+1)^{2d+2}|\Delta|\log|\Delta|)=\OO((n\|\BB\|_\infty+1)^{2d+2}|\BB|\log|\BB|)=(n\|\BB\|_\infty|\BB|)^{\OO(d)}$ every $q[\vect v]$ for $q\in Q$ and $\vect v\in\NN^{d+1}$ such that $\|v\|_\infty\leq n\|\BB\|_\infty$.

    For the semilinear part, we use the construction of the automaton $\VV'$, also given in the proof of Lemma~\ref{lem:bounds-a-c}. For a unambiguous representation
 $C=\cup_{i=1}^pC_i$, where $C_i=\vect{c_i}+P_i^*$, the associated vector automaton $\VV'$ has $p+1$ states, $s_I$ the initial state and $s_1,\ldots s_p$ for the $C_i$. Applying Proposition~\ref{prop:computecoeffsvectorautomaton}, we can compute all the $s_i[\vect v]$, for all $i$ and all $\vect v\in\mathbb{N}^d$ such that $\|\vect v\|_\infty\leq \|\vect \BB\|_\infty$, also in time $(n\|\BB\|_\infty|\BB|)^{\OO(d)}$.
 As $\VV'$ is unambiguous, $s_i[\vect v]$ can only value $1$ or $0$, whether $\vect v$ is in $C_i$ or not (except for $s_I[\vect v]$, which value $0$ iff $\vect v=\vect 0$ and $\vect 0\in C$).

Finally, to obtain the total number of words of length $k\leq n$ accepted by the $\BB$, we sum the values $q[\vect{v'}]$ for every state $q\in F$, and over every vector $\vect{v'}=[k,\vect v]$, such that $s[\vect v]=1$ for at least one state $s$ of $\VV'$. For the cost of computing these additions, notice that $q[\vect{v'}]\leq |\Delta|^k\leq |\BB|^n$. Thus, this sum can be done in $\OO((n\|B\|_\infty+1)^d(p+|F|n\log|\BB|))=(n\|\BB\|_\infty|\BB|)^{\OO(d)}$ time.
\end{proof}

\subsection{Proof of Corollary~\ref{cor:complexite}}

We can finally prove the statement.

\begin{proof}[Proof of of Corollary~\ref{cor:complexite}]
	As sketched in the main presentation, the proof relies on the fact that $L(\AA)\subseteq L(\BB)$ if and only if $L(\AA)\cap L(\BB)=L(\AA)$. By Theorem~\ref{thm:inclusionbound}, if 
	$L(\AA)\cap L(\BB)\neq L(\AA)$, there exists a word $u\in L(\AA) $ that is not in $L(\AA)\cap L(\BB)$ of length at most
	$2^{2^{\OO(d^2\log(dM))}}$, where $d=d_\AA+d_\BB$ and 
	$M=|\AA|\,|\BB|\,\|\AA\|_\infty\|\BB\|_\infty$: we get this bound by studying the maxdegree and infinite norm of the product automaton of $\AA$ and $\BB$ built to recognize the intersection.
	It means that there exists some constant $c>0$ such that looking for words up to size $W =\lfloor 2^{2^{cd^2\log(dM)}}\rfloor$ is sufficient to decide whether
	$L(\AA)\subseteq L(\BB)$ or not.
	
	Thus, we just have to compute the number of words in 
	$L(\AA)\subseteq L(\BB)$ and $L(\AA)$ up to size $W$ to check if they are equal. By Proposition~\ref{prop:enumeration_weakly_unambiguous}, this can be done in time  $(WM)^{\OO(d)}$. Hence there exists some constant $c'>0$ such that
	\[
	(WM)^{\OO(d} \leq (WM)^{c'd}
	\leq \left(2^{2^{cd^2\log(dM)}}\right)^{c'd}M^{c'd}
	= 2^{c'd2^{cd^2\log(dM)}+c'd\log M} \leq 2^{2^{c''d^2\log(dM)}}.
	\]
	for a sufficiently large constant $c''$. This concludes the proof.
\end{proof}

\end{document}